\documentclass[11pt]{book}

\usepackage{color}
\usepackage{graphicx}  
\usepackage[all]{xypic}
\usepackage{psfrag}
\usepackage{tocbibind} 
\usepackage{epsf}
\usepackage{epic}
\usepackage{eepic}
\usepackage{epsfig}
\usepackage{wrapfig}
\usepackage{a4}
\usepackage{amssymb,latexsym}
\usepackage[mathscr]{eucal}
\usepackage{citesort}
\usepackage[active]{srcltx}
\usepackage{deleq}
\usepackage{url}
\usepackage{fancybox}
\usepackage{ntheorem}
\usepackage{bigint}
\usepackage{mhequ}
\usepackage[g]{esvect}
\usepackage{epstopdf}

\newcommand{\restrict}[1]{\ptcx{here starts a restrict environment, not to be displayed in the arxiv version}{\color{red} #1}}

\newcommand{\ptcx}[1]{\mnote{{\bf ptc:} {\color{red} #1}}}

\newcommand{\mnotex}[1]
{\protect{\stepcounter{mnotecount}}$^{\mbox{\footnotesize
$
\bullet$\themnotecount}}$ \marginpar{
\raggedright\tiny\em
$\!\!\!\!\!\!\,\bullet$\themnotecount: #1} }

\newcommand{\jamesr}[1]{{\color{red}\mnote{{\color{red}{\bf jg:}
#1} }}}

\newcommand{\jamesx}[1]{}
\renewcommand{\jamesx}[1]{{\mnote{{\color{black}{\bf jg:}
#1} }}}

\newcommand{\h}[2]{#1\dotfill\ #2\\\ptc{fixme}}

\newcommand{\imcDpS}{\,\,\mathring{\!\!\mcD}{}^+(\hyp)}

\def\nz{\ifmmode {I\hskip -3pt N} \else {\hbox {$I\hskip -3pt N$}}\fi}
\def\zz{\ifmmode {Z\hskip -4.8pt Z} \else
       {\hbox {$Z\hskip -4.8pt Z$}}\fi}
\def\qz{\ifmmode {Q\hskip -5.0pt\vrule height6.0pt depth 0pt
       \hskip 6pt} \else {\hbox
       {$Q\hskip -5.0pt\vrule height6.0pt depth 0pt\hskip 6pt$}}\fi}
\def\rz{\ifmmode {I\hskip -3pt R} \else {\hbox {$I\hskip -3pt R$}}\fi}
\def\cz{\ifmmode {C\hskip -4.8pt\vrule height5.8pt\hskip 6.3pt} \else
       {\hbox {$C\hskip -4.8pt\vrule height5.8pt\hskip 6.3pt$}}\fi}
\def\au{{\setbox0=\hbox{\lower1.36775ex\hbox{''}\kern-.05em}\dp0=.36775ex\hs
kip0pt\box0}}
\def\ao{{}\kern-.10em\hbox{``}}

\newcommand\Gregbeq{\begin{eqnarray}}
\newcommand\Gregeeq{\end{eqnarray}}

\newcommand{\mod}{\ \mathrm{mod}\, } 

\newcommand{\scri}{{\mycal I}}%
\def\h1{{\hat 1}}
\def\h2{{\hat 2}}

\def\3f{\frac{3}{2}}

\newcommand{\mcDpmSI}{{ {\mcD}_{I}^\pm(\hyp)}}
\newcommand{\mcDSI}{{\mcD_{I}(\hyp)}}
\newcommand{\mcDSpI}{{\mcD_{I}^+(\hyp)}}

\newcommand{\mcDSmI}{{\mcD_{I}^-(\hyp)}}
\newcommand{\imcDSI}{{\,\,\,\mathring{\!\!\!\mcD}_{I}(\hyp)}}
\newcommand{\imcDpSI}{{\,\,\,\mathring{\!\!\!\mcD}{}_{I}^+(\hyp)}}

\newcommand{\mcDSJ}{{\mcD_{J}(\hyp)}}
\newcommand{\mcDSpJ}{{\mcD_{J}^+(\hyp)}}

\newcommand{\mcDpmSJ}{{ {\mcD}_{J}^\pm(\hyp)}}

\newcommand{\cI}{{\check{I}{}}}

\newcommand{\roscoff}[1]{}

{\catcode `\@=11 \global\let\AddToReset=\@addtoreset}
\AddToReset{figure}{section}

\DeclareFontFamily{OT1}{rsfs}{}
\DeclareFontShape{OT1}{rsfs}{m}{n}{ <-7> rsfs5 <7-10> rsfs7 <10-> rsfs10}{}
\DeclareMathAlphabet{\mycal}{OT1}{rsfs}{m}{n}

{\catcode `\@=11 \global\let\AddToReset=\@addtoreset}
\AddToReset{equation}{section}
\renewcommand{\theequation}{\thesection.\arabic{equation}}
\newcounter{mnotecount}[section]

\renewcommand{\themnotecount}{\thesection.\arabic{mnotecount}}

\newcommand{\xmcop}{X_{\mcO_p}}\newcommand{\xmcopi}{X_{\mcO_{p_i}}}\newcommand{\xmcoq}{X_{\mcO_q}}

\newcommand{\jlcax}[1]{}
\newcommand{\eean}{\nonumber\end{eqnarray}}

\newcommand{\ptcr}[1]{{\color{red}\mnote{{\color{red}{\bf ptc:}
#1} }}}

\newcommand{\kk}[1]{}

\newcommand{\mcH}{{\mycal H}}

\newcommand{\vx}{{\vec x}}
\newcommand{\sHh}{\hyp_{\mcH}}

\newcommand{\vz}{{\vec z}}
\newcommand{\vg}{\gamma}
\newcommand{\dG}{\dot \Gamma}
\newcommand{\pxhere}{{\phi_{\vx}(s)}}
\newcommand{\pxpm}{{\phi_{\pm}(s)}}

\newcommand{\hmcM}{\,\,\,\,\widehat{\!\!\!\!\mcM}}

\newcommand{\beq}{\begin{equation}}

\newcommand{\FS}       
                  {F}

\newcommand{\HS} 
       {H_{\mbox{\scriptsize volume}}}

{\ptc{this should be removed in the oberwolfach version}}%

\newcommand{\trh}{\mathrm{tr}_h}

\newcommand{\dotJ}{{\dot J}}

\newcommand{\mcE}{{\mycal E}}%

\newcommand{\ihyp}{\,\,\mathring{\!\!\hyp}}
\newcommand{\eeal}[1]{\label{#1}\end{eqnarray}}
\newcommand{\C}{{\mathbb C}}
\newcommand{\bed}{\begin{deqarr}}
\newcommand{\eed}{\end{deqarr}}
\newcommand{\bedl}[1]{\begin{deqarr}\label{#1}}
\newcommand{\eedl}[2]{\arrlabel{#1}\label{#2}\end{deqarr}}

\newcommand{\loc}{\textrm{\scriptsize\upshape loc}}

\newcommand{\mcO}{{\mycal O}}
\newcommand{\mcU}{{\mycal U}}

\newcommand{\mcK}{{\mycal K}}

\newcommand{\bel}[1]{\begin{equation}\label{#1}}
\newcommand{\bea}{\begin{eqnarray}}
\newcommand{\bean}{\begin{eqnarray}\nonumber}
\newcommand{\beal}[1]{\begin{eqnarray}\label{#1}}
\newcommand{\eea}{\end{eqnarray}}

\newcommand{\distb}{\dist_{\backmg}}
\newcommand{\dist}{\mathrm{dist}}
\newcommand{\lorentz}{\sigma}
\newcommand{\lorp}{\lorentz_p} 
\newcommand{\nn}{\nonumber}
\newcommand{\Eq}[1]{Equation~\eq{#1}}
\newcommand{\Eqsone}[1]{Equations~\eq{#1}}
\newcommand{\Eqs}[2]{Equations~\eq{#1}-\eq{#2}}

\newcommand{\qed}{\hfill $\Box$ \medskip}
\newcommand{\proof}{\noindent {\sc Proof:\ }}
\newcommand{\be}{\begin{equation}}
\newcommand{\eeq}{\end{equation}}
\newcommand{\ee}{\end{equation}}
\newcommand{\beqa}{\begin{eqnarray}}
\newcommand{\eeqa}{\end{eqnarray}}
\newcommand{\beqan}{\begin{eqnarray*}}
\newcommand{\eeqan}{\end{eqnarray*}}
\newcommand{\ba}{\begin{array}}
\newcommand{\ea}{\end{array}}

\newcommand{\backmg}{h} 
\newcommand{\const}{\mbox{\rm const}} 

\newcommand{\hyp}{\mycal S}

\newcommand{\mcM}{{\mycal M}}
\newcommand{\mcD}{{\mycal D}}

\newcommand{\mcW}{{\mycal W}}
\newcommand{\mcV}{{\mycal V}}

\newcommand{\cU}{{\cal U}}

\newcommand{\mnote}[1]
{\protect{\stepcounter{mnotecount}}$^{\mbox{\footnotesize
$
\bullet$\themnotecount}}$ \marginpar{
\raggedright\tiny\em
$\!\!\!\!\!\!\,\bullet$\themnotecount: #1} }

\newcommand{\warn}[1]
{\protect{\stepcounter{mnotecount}}$^{\mbox{\footnotesize
$
\bullet$\themnotecount}}$ \marginpar{
\raggedright\tiny\em
$\!\!\!\!\!\!\,\bullet$\themnotecount: {\bf Warning:} #1} }

\newcommand{\R}{\mathbb R}

\newcommand{\N}{\mathbb N}

\newcommand{\eq}[1]{(\ref{#1})}

\newcommand{\hmcK}{\;\;\;\widehat{\!\!\!{\mcK}}}

\newcommand{\ptc}[1]{\mnote{{\bf ptc:}#1}}

\newcommand{\diag}{\mbox{\rm diag}}
\newcommand{\Lg}{{g}} 
\newcommand{\Ric}{\mbox{\rm Ric}}

\newcommand{\beqar}{\begin{deqarr}}
\newcommand{\eeqar}{\end{deqarr}}

\newcommand{\beaa}{\begin{eqnarray*}}
\newcommand{\eeaa}{\end{eqnarray*}}

\newcommand{\ohyp}{\,\,\overline{\!\!\hyp}} 

\newcommand{\cP}{{\cal P}}

\newcommand{\hg}{{\hat g}}

\newcommand{\eg}{{\emph{e.g.,\/}}}

\newcommand{\chindex}[1]{\index{#1}}

\newcommand{\dotg}{{\dot \gamma}}
\newcommand{\Mgk}{(\mcM,g)_{C^k}}
\newcommand{\Mgz}{(\mcM,g)_{C^0}}
\newcommand{\Mcont}{Consider a spacetime $\Mgz$. }
\newcommand{\Mginfty}{(\mcM,g)_{C^\infty}}
\newcommand{\MCinfty}{Consider a spacetime $\Mginfty$. }

\newcommand{\Mgthreek}{(\mcM,g)_{C^2}}

\newcommand{\MCthreek}{Consider a spacetime $\Mgthreek$. }
\newcommand{\Mgtwo}{(\mcM,g)_{C^2}}

\newcommand{\MCtwok}{Consider a spacetime $\Mgtwo$. }
\newcommand{\Mgoneone}{(\mcM,g)_{C^{1,1}}}
\newcommand{\MConeone}{Consider a spacetime $\Mgoneone$. }

\newcommand{\bethm}{\begin{Theorem}}
\newcommand{\et}{\end{Theorem}}
\newcommand{\bl}{\begin{Lemma}}

\renewcommand{\ptcx}[1]{}
\renewcommand{\jamesx}[1]{}
\renewcommand{\restrict}[1]{\ptcx{here starts a restrict environment, not   displayed in the arxiv version} }

\renewcommand{\roscoff}[1]{#1}

\renewcommand{\mnotex}[1]{#1}

\newtheorem{theorem}{Theorem}[section]

\newtheorem{Theorem} {\sc  Theorem\rm} [section]
\newtheorem{Corollary} [Theorem] {\sc  Corollary\rm}
\newtheorem{Lemma} [Theorem] {\sc  Lemma\rm}
\newtheorem{Proposition} [Theorem] {\sc  Proposition\rm}

\newtheorem{Definition}[Theorem]{\sc  Definition\rm}

\theorembodyfont{\upshape}
\newtheorem{Remark}[Theorem]{\sc  Remark\rm}
\newtheorem{remark}[Theorem]{\sc  Remark\rm}
\newtheorem{Remarks}[Theorem]{\sc  Remarks\rm}

\newtheorem{Example} [Theorem] {\sc  Example\rm}

\theoremsymbol{\ensuremath{\diamondsuit}}
\newcommand{\fcoco}{\small}
\theorembodyfont{\fcoco}\theoremseparator{}\theoremindent0.5cm

\newtheorem{cocoExa}[Theorem]{\sc Example\rm}
\theoremstyle{nonumberplain}\theorembodyfont{\fcoco}
\theoremseparator{}\theoremindent0.5cm
\newtheorem{coco}{}

\begin{document}
\frontmatter
\title{Elements of causality theory\protect\thanks{Preprint UWThPh-2011-32}}
\author{Piotr T.\ Chru\'{s}ciel%
\thanks{Supported in part by the Polish Ministry of Science and
Higher Education grant Nr N N201 372736.}
\\
University of
Vienna\\{\small\url{http://homepage.univie.ac.at/piotr.chrusciel/}}
}

\maketitle

\setcounter{tocdepth}{2} \tableofcontents

\mainmatter

\chapter{Introduction}

These notes present some elements of causality theory, as useful  to study
general relativity.
They amount to an incremental compilation (and thus are far from being well
synchronized and balanced between topics) of notes for lectures  I held  at
various summer schools over the years.
While they are not as complete as other presentations of the topic
\cite{HE,Beem-Ehrlich:Lorentz2,BONeill,MinguzziSanchez,PenroseDiffTopo,Kriele},
there is some originality in that the whole treatment is based on a definition
of causal curves which
allows one to simplify many arguments.

Now, in light of studies of the
Einstein equations with metrics of low
differentiability~\cite{WangCones,WangRicci,KlainermanRodnianski:r1,%
Maxwell:Rough,Maxwell:Compact}, it is of interest to understand the
differentiability needed for the causality part of the theory. The standard
references are either vague about differentiability, or assume smoothness of the
metric. In our
presentation we keep track of the differentiability of the metric needed for the arguments.
This leads to a coherent causality theory for $C^2$ metrics.
This differentiability threshold  can be traced back to Proposition~\ref{Paccum}, page~\pageref{Paccum} below, as used to prove that accumulation curves of causal curves are causal, as well as to the deformation lemma \ref{Lpushup0}, page~\pageref{Lpushup0}. The threshold for the accumulation result, and some of its consequences, is relaxed in~\cite{ChGrant} using different methods,
and note that the approach there requires developing first a causality theory for smooth metrics in any case.

Given the number of alternative more complete treatments, it is not clear
that the above two reasons justify making the notes public. However, the notes
serve as a crossreference for the accompanying notes~\cite{ChGrant} on causality
for continuous metrics, which is the main reason for posting.

The reader is warned that some of our proofs  do not
apply to metrics which are not $C^2$, and that a few key results (e.g.,
deforming not-everywhere-null-causal curves to timelike ones keeping end points
fixed) are plain wrong for metrics with lower differentiability.

\chapter{Causality}
 \label{SC}

 \ptcx{make sure that completeness.tex, on the choquet-bruhat geroch theorem,
and continuous.tex, on causality with low differentiability, are included in
the levoca runs; some introductory comments from ContinuousCausality and Causality.tex could be moved here}

Unless explicitly indicated otherwise, or otherwise clear from the context. we consider manifolds
equipped with a  $C^\infty$ atlas and a continuous metric. As already pointed out, the considerations below give a coherent causality
theory for metrics which are $C^2$ and manifolds which are $C^3$. However, for many considerations  a metric of $C^0$ differentiability class suffices. We will strive to indicate explicitly the differentiability of the metric needed as the presentation evolves.

\section{Time orientation}\label{SCTo}

\ptcx{check Seifert's approach to causality; compare defs} Recall that at
each point $p\in\mcM$ the set of timelike vectors in $T_pM$ has
precisely two components. A \emph{time-orientation of\/} $T_p\mcM
$ is the assignment of the name ``future pointing vectors'' to
one of those components; vectors in the remaining component
are then called ``past pointing''. The set of future pointing
timelike, or causal, vectors, is stable under addition and
multiplication by positive numbers; similarly for past pointing
ones. (In particular this implies convexity.) In order to see
this, suppose that $X=(X^0,\vec X)$ and $Y=(Y^0,\vec Y)$ are
timelike future pointing, in an ON-frame this is equivalent to
$$ |\vec X|< X^0\;, \quad |\vec Y| < Y^0\;,$$
and the inequality
$$ |\vec X+\vec Y|\le |\vec X|+|\vec Y|<X^0+ Y^0$$ follows.
Two timelike vectors $X$ and $Y$ have the same time orientation
if and only if
\bel{eDC1} g(X,Y) < 0
 \;;
 \ee
this  is immediate\restrict{from \eq{elsp1}} in an ON frame in which $X$ is
proportional to $e_0$.

 A
time-orientation of $T_p\mcM $ can always be propagated to a
neighborhood of $p$ by choosing any continuous vector field $X$
defined around $p$ which is timelike and future pointing at
$p$. By continuity of the metric and of $X$, the vector field
$X$ will be timelike in a sufficiently small neighborhood
$\mcO_p$ of $p$, and for $q\in\mcO_p$ one can define future
pointing vectors at $q$ as those lying in the same component of
the set of
timelike
vectors as $X(q)$: for $q\in \mcO_p$ the vector $Y\in T_q\mcM $
will be said to be timelike future pointing if and only if
$g(Y,X(q))< 0$.
 A
Lorentzian manifold is said to be \emph{time-orientable\/} if
such locally defined time-orientations can be defined globally
in a consistent way; that is, we can cover $\mcM $ by
coordinate neighborhoods $\mcO_p$, each equipped with a vector
field $\xmcop $, such that $g(\xmcop ,\xmcoq)<0$ on $\mcO_p\cap
\mcO_q$.

Some Lorentzian manifolds will not be time-orientable, as is shown by
the flat metric\footnote{In two dimensions $-g$ is a
Lorentzian metric whenever $g$ is, and the operation $g\to -g$
has the effect of interchanging the role of space and of time.
\restrict{The reader will notice that while the M\"obius strip with the
flat metric $g$ of Figure~\ref{FC1} is not time-orientable, it
becomes time-orientable when equipped with $-g$.}} on the
M\"obius strip.
\restrict{, \emph{cf.}\/ Figure~\ref{FC1}. \ptc{figure missing; a figure can be obtained from $x  =   [R+scos(1/2t)]cost$,
$y   =   [R+scos(1/2t)]sint $, and $ z   =   ssin(1/2t)$,
according to
\url{http://mathworld.wolfram.com/MoebiusStrip.html}; a cool
mobius ring on
\url{http://www.ka-gold-jewelry.com/p-products/mobius-ring-silver.php},
file saved in the figures directory; make a drawing of two
metrics, one which is time orientable, and one which is not}
\begin{figure}\label{FC1}
\vspace{6cm} \caption{The M\"obius strip, with the flat metric is
 $-dt^2+dx^2$ (so that
 the light cones are at $45^o$) provides an example of a
two dimensional Lorentzian manifold which is not time-orientable.}
\end{figure}
}
On a
time-orientable manifold there are precisely two choices of
time-orientation possible, and $(\mcM ,g)$ will be said
\emph{time oriented\/} when such a choice has been made.
This leads us to the fundamental definition:
\begin{Definition}
  \label{DC1}
A pair $(\mcM ,g)$ will be called a \emph{space-time\/} if
$(\mcM ,g)$ is a time-oriented Lorentzian manifold. We write $\Mgk$ to denote a space-time with a metric of $C^k$-differentiability class.
\end{Definition}

\begin{coco}
\begin{Remark}
\label{RDoubleCover}A Lorentzian manifold $(\mcM,g)$ which is
\emph{not\/} time-orientable has a double cover which
is~\cite{Geroch:topology}. The proof goes as follows: Choose
any $p_0\in \mcM$ and set
$$
\hmcM:=\{(p,\gamma): \mbox{ $p$ is a point of $\mcM$ and $\gamma$ is a
continuous curve from $p$ to $p_0$} \}/\sim\;,
$$
where $\sim$ is the following equivalence relation:
$(p,\gamma)\sim (p',\gamma')$ if $p=p'$ and if there exists a
continuous timelike vector field defined along the curve
obtained by first following $\gamma$ from $p$ to $p_0$ and then
$\gamma'$ from $p_0$ to $p'=p$. The usual arguments from the
theory of covering spaces show that $\hmcM$ can be equipped
with a manifold structure, and covers $\mcM$ twice. $\hmcM$ is
then equipped with the pull-back $\hat g$ of $g$ using the
covering map; time-orientability of $(\hmcM, \hat g)$ should
be clear. Furthermore, any time-orientable cover of $\mcM$ also
covers $\hmcM$, so $\hmcM$ can be thought-of as the smallest
time-oriented covering of $\mcM$.
\end{Remark}
\end{coco}

On any space-time there always exists a globally defined future
directed timelike vector field
--- to show this, consider the locally defined timelike vector fields
$\xmcop $ defined on neighborhoods $\mcO_p$ as described above.
One can choose a locally finite covering of $\mcM $ by such
neighborhoods $\mcO_{p_i}$, $i\in\N$, and construct a globally
defined vector field $X$ on $\mcM $ by setting $$X=\sum_i
\phi_i \xmcopi\;,$$ where the functions $\phi_i$ form a
partition of unity dominated by the covering
$\{\mcO_{p_i}\}_{i\in\N}$. The resulting vector field will be
timelike future pointing everywhere, as follows from the fact
that the sum of an arbitrary number of future pointing timelike
vectors is a future pointing timelike vector.

Now, non-compact manifolds always admit a nowhere vanishing
vector field. However, compact manifolds possess a nowhere
vanishing vector field if and only if
they have vanishing Euler characteristic $\chi$.
\jamesx{Check whether this requires orientability.}
More generally, if $M$ is a compact, orientable manifold, then
the Poincar\'{e}--Hopf theorem (see, e.g.,~\cite{GuilleminPollack})
implies that the index of any smooth vector field, $X$, on $M$ (i.e. the zeroes of $X$ counted with signs) satisfies
\[
\mathrm{index}(X) = \chi(M).
\]
As such, if $M$ admits a non-vanishing vector field $X$, then $\mathrm{index}(X) = 0$
and, hence, $\chi(M) = 0$. Conversely, if $M$ has $\chi(M) = 0$, then any smooth
vector field $X$ on $M$ is of index zero. A theorem of Hopf then implies that there
exists a non-vanishing vector field on $M$ homotopic to $X$.
\jamesx{We still need a reference for the second part. Concerning the non-compact case, there is an incomprehensible article on
MathOverflow by Thurston. I will keep looking for a better source.}
\ptcx{it could be convenient to find a reference here; Markus
contains absurd statements, Guillemin Pollack do it only in one
direction; Yvonne gives Chern, Ann Math \cite{Chern} mais ca n'y est pas; Milnor and Stasheff or Hirsch's differential topology?
; there are some absurd
statements in~\cite{Markus}?}

\restrict{
\begin{remark}
\ptcr{this is a remark by james grant}
There is still a slight problem with the non-orientable case. If $M$ is non-orientable,
then we may consider the oriented double-cover $\pi \colon \tilde{M} \to M$. Since $\tilde{M}$ is
compact and orientable, it follows that $\chi(\tilde{M}) = 0$ if and only if there exists a non-vanishing,
continuous vector field $X$ on $\tilde{M}$. If we consider the push-forward $\pi_* X$, this gives a
double-valued vector field on $M$. In order to define a Lorentzian metric, we require that the two
values of $\pi_* X$ differ only by a sign. It is not clear, however, that $X$ can be chosen in this way.
On the other hand, if we have a line element field on $M$ (i.e. an assignment of a pair of equal and
opposite vectors $\pm V$ at each point in $M$~\cite[pp.~39]{HE}), then we should be able to construct this as an image of
a vector field on $\tilde{M}$.
\jamesr{I will think about this.}
\end{remark}
}

As such, vanishing of the Euler characteristic, $\chi$,
is a necessary and sufficient condition of topological nature for a compact, orientable manifold
to be a time-orientable Lorentzian manifold. We actually have the following:

\begin{Proposition}
 \label{PDC1}
A manifold $\mcM$ admits a space-time structure if and only if
there exists a nowhere vanishing, continuous vector field on $\mcM$.
\end{Proposition}
\proof The necessity of the existence of a nowhere vanishing
vector field on $\mcM$ has already been established.
Conversely, suppose that such a vector field $X$ exists, and
let $h$ be any Riemannian metric on $\mcM$. Then the formula
\bel{ePDC1} g(Y,Z)= h(Y,Z) - 2\frac{h(Y,X)h(Z,X)}{h(X,X)}
\ee
defines a Lorentzian metric on $\mcM$. Finally, the existence
of a globally defined timelike vector field $X$ on a Lorentzian
manifold $(\mcM,g)$ implies time-orientability of $\mcM$ in the
obvious way -- choose $\mcO_p=\mcM$ and $\xmcop =X$.
 \qed

Summarising, non-compact $\mcM$'s always admit both a
Lorentzian metric and a space-time structure. Now, because the
Euler characteristic of a double-cover of $\mcM$ is zero if and
only if that of $\mcM$ is, it follows from
Remark~\ref{RDoubleCover} and Proposition~\ref{PDC1} that
compact $\mcM$'s admit a Lorentzian metric if and only if they
have vanishing Euler characteristic. For example, no Lorentzian
metrics exist on $S^2$.

\section{Normal coordinates}\label{SCnc}

Given a $C^2$ metric, for $p\in \mcM $ the exponential map
$$
 \exp_p:T_p\mcM \to \mcM
$$
is defined as follows; if $X$ is a vector in the tangent
space $ T_p\mcM $, then  $\exp_p(X)\in \mcM $ is the point
reached by following a geodesic with initial point $p$ and
initial tangent vector $X\in T_p\mcM $ for an affine distance
one, provided that the geodesic in question can be continued
that far. Now an affinely parameterized geodesic solves the
equation \be\label{geodeq} \nabla_{\dot x}\dot x = 0 \quad
\Longleftrightarrow \quad \frac{d^2x^\mu}{d s^2} =
-\Gamma^\mu_{\alpha\beta}\frac{dx^\alpha}{d
  s}\frac{dx^\beta}{d s}\;,
  \ee
  where the
$\Gamma^\mu_{\alpha\beta}$'s are the Christoffel symbols of the
metric $g$, defined as
\begin{eqnarray}\displaystyle
\label{Chs} & \Gamma^\mu_{\alpha\beta} := \displaystyle\frac 12
g^{\mu\sigma}\left(\frac{\partial g_{\sigma \alpha}}{\partial
x^\beta} +\frac{\partial g_{\sigma \beta}}{\partial x^\alpha}-
\frac{\partial g_{ \alpha\beta}}{\partial x^\sigma}\right)\;, &
\\ &
g^{\mu\sigma}:=g^{\#}(dx^\mu,dx^\sigma)\;,\qquad
g_{\alpha\beta}:=g(\partial_\alpha,\partial_\beta)\;.
&\label{Chs2}
\end{eqnarray}
Here, as elsewhere, we use the symbol $g^{\#}$ to denote the
``contravariant metric", that is, the metric on $T^*\mcM$
constructed out from $g$ in the canonical way (\restrict{see
Appendix~\ref{Srli}; }the matrix $g^{\alpha\beta}$ is thus the
matrix inverse to $g_{\alpha\beta}$). However, it is usual in
the literature to use the same symbol $g$ for the metric
$g^\#$, as well as for all other metrics induced by $g$ on
tensor bundles over $\mcM$, and we will often do so.

Equations~\eq{Chs}-\eq{Chs2} show that when the metric is
 of
$C^{1,1}$ differentiability class, then the Christoffel symbols
are Lipschitz continuous, which guarantees local existence and
uniqueness of solutions of \eq{geodeq}. Due to the lack of
uniqueness\footnote{Examples of $C^{1,\alpha}$ metrics with
non-unique null
  geodesics for $0<\alpha <1$ can be found in~\cite[Appendix F]{SCC} and~\cite{ChGrant},
  compare~\cite{Hartman} for spacelike geodesics in a Riemannian context.
  Here $C^{k,\alpha}$ is the space of $k$ times differentiable functions (or
  maps, or sections --- whichever is the case should be clear from the
  context), the $k$'th derivatives of which satisfy, locally, a
  H\"older condition of order $\alpha$.} of the Cauchy problem for \eq{geodeq}
  for metrics which are not $C^{1,1}$, various problems arise when attempting
to develop causality theory on manifolds
with a metric with less regularity\footnote{One can construct large classes of
solutions to the Cauchy problem for the vacuum Einstein
equations which are {\em not} of $C^{1,1}$ differentiability
class~\cite{KlainermanRodnianski:r1,KlainermanRodnianski:r2,SmithTataru,BahouriChemin}\restrict{, see Section \ref{???}}.
This leads to a  mismatch in differentiability
between the Cauchy problem and causality theory which has not been completely clarified yet.} than
$C^{1,1}$,   addressed in~\cite{ChGrant} (see also~\cite{SorkinWoolgar,KeyeMartin}).
\ptcx{needs synchronizing, similarly for the footnote}

The domain $\mcU_p$ of $\exp_p$ is always the largest subset of
$T_p\mcM $ on which the exponential map is defined. By
construction, and by homogeneity properties of solutions of
\eq{geodeq} under a linear change of parameterization (see
\eq{just}), the set $\mcU_p$ is star-shaped with respect to the
origin (this means that if $X\in \mcU_p$ then we also have
$\lambda X\in\mcU_p$ for all $\lambda\in[0,1]$). When the
metric is $C^{1,1}$, continuity of solutions of ODE's upon
initial values shows that $\mcU_p$ is an open neighborhood of
the origin of $T_p\mcM $.

The exponential map is neither surjective nor injective in
general. For example, on $\R\times S^1$ with the flat metric
$-dt^2+dx^2$, the ``left-directed'' null geodesics
$\Gamma_-(s)=(s,-s \;\mathrm{mod} \, 2\pi )$ and the
``right-directed'' null geodesics $\Gamma_+(s)=(s,s
\;\mathrm{mod} \, 2\pi )$ meet again after going each ``half of
the way around $S^1$", and injectivity fails. Both in de-Sitter
and in anti-de-Sitter\ptcx{ ref? more detail?} space-time all
timelike geodesics meet again at a point, which leads to lack
of surjectivity of the exponential map.

 A Lorentzian manifold is said to be \emph{geodesically
complete} if all geodesics can be defined for all real values
of affine parameter; this is equivalent to the requirement that
for all $p\in \mcM $ the domain of the exponential map is
$T_p\mcM $. One also talks about \emph{timelike
  geodesically complete} space-times, \emph{future timelike
  geodesically complete} space-times, \emph{etc.}, with those notions defined in an
obvious way.

\begin{coco}
It follows from the \emph{Hopf--Rinow\/} theorem~\cite{MilnorMorse,HopfRinow}
that \emph{\underline{compact Riemannian} manifolds are
geodesically complete}. There is no Lorentzian analogue of
this, the standard counter-example proceeds as follows:
\begin{Example}\label{Exgeoinc}
Consider the following symmetric tensor field on $\R^2$:
\bel{PC1.0}g=\frac{2dx dy}{x^2+y^2}\;.\ee We have
\bel{PC1.01}g_{\mu\nu}=\frac{1}{x^2+y^2}\left[\begin{array}{cc}
0 & 1 \\ 1 & 0 \end{array}\right]\quad\Longrightarrow \quad
\det g_{\mu\nu}= - \frac{1}{(x^2+y^2)^2}\;,\ee which shows that
$g$ is indeed a Lorentzian metric. Note that for all $\lambda\in
\R^*$ the maps
$$\R^2\ni(x,y)\to \phi_\lambda(x,y):=(\lambda x, \lambda y)$$
are isometries of $g$:
$$\phi^*_\lambda g =\frac{2d(\lambda x) d(\lambda y)}{(\lambda x)^2+(\lambda y)^2}=\frac{2dx dy}{x^2+y^2} =
g\;.$$ It follows that for any $1\ne \lambda>0$ the metric $g$
passes to the quotient space
$$\left\{\R^2\setminus\{0\}\right\}/\phi_\lambda=\{(x,y)\sim (\lambda x,\lambda
y)\}\approx S^1\times S^1 = {\mathbb T}^2\;.$$ (Clearly the
quotient spaces with $\lambda$ and $1/\lambda$ are the same, so
without loss of generality one can assume $\lambda >1$.) In
order to show geodesic incompletess of $g$ we will use the
following result:
\begin{Proposition}\label{Pgeoinc}
\MCtwok
Let $f$ be a function such that $g(\nabla f,\nabla f)$ is a
constant. Then the integral curves of $\nabla f$ are affinely
parameterized geodesics.
\end{Proposition}
\proof Let $X:= \nabla f$, we have \beaa (\nabla_X X)^j & = &
\nabla ^i f \nabla _i \nabla^j f
 =  \nabla ^i f
\nabla ^j \nabla_i f
\\& = & \frac 12 \nabla ^j( \nabla^i f \nabla _i
 f)=\frac 12 \nabla ^j(g(\nabla f,\nabla f)) = 0\;.
\eeaa\qed

Returning to the metric \eq{PC1.0}, let $f=x$, by \eq{PC1.01}
we have
$$g^{\mu\nu}=(g^{\mu\nu})^{-1}=({x^2+y^2})\left[\begin{array}{cc} 0 & 1 \\ 1 &
0 \end{array}\right]\;,$$ so that
$$\nabla f =({x^2+y^2}) \partial_y \quad \Longrightarrow \quad g(\nabla f,
\nabla f)=0\;.$$ Proposition~\ref{Pgeoinc} shows that the integral
curves of $\nabla f$ are null affinely parameterized geodesics.
Let $\gamma(s)=(x^\mu(s))$ be any such integral curve, thus
$$\frac {dx^\mu}{ds} = \nabla^\mu f \quad \Longrightarrow \quad \frac
{dx}{ds}=0\;, \quad \frac {dy}{ds}= (x^2+y^2)\;.$$ It follows that
$x(s)=x(0)$ for all $s$. The equation for $y$ is easily
integrated; for our purposes it is sufficient to consider that
integral curve which passes through $(0,y_0)\in
\R^2\setminus\{0\}$, $y_0>0$ --- we then have $x(s)=0$ for all
$s$ and \bel{PC1.03}\frac {dy}{ds}= y^2\quad \Longrightarrow
\quad y(s) = \frac{y_0}{1-y_0s}\;.\ee This shows that $y(s)$
runs away to infinity as $s$ approaches
$$s_\infty:= \frac 1 {y_0}\;.$$
It follows that $\gamma$ is indeed incomplete on
$\R^2\setminus\{0\}$. To see that it is also incomplete on the
quotient torus
$\left\{\R^2\setminus\{0\}\right\}/\phi_\lambda$, $\lambda
>1$, note that the
image of $\gamma(s)=(0,y(s))$ under the equivalence relation
$\sim$ is a circle, and there exists a sequence $s_j\to
s_\infty$ such that $\gamma(s_j)$ passes again and again
through its starting point:
$$ y(s_j)=\lambda ^j y_0 \quad \Longrightarrow \quad (0,y(s_j))\sim (0,y_0)
\ \mbox{ in } \ \left\{\R^2\setminus\{0\}\right\}/\phi_\lambda
\;.$$ By \eq{PC1.03} we have $$\frac {dy}{ds}(s_j)=
\left(y(s_j)\right)^2 =(\lambda_jy_0 )^2
\displaystyle\longrightarrow_{s_j\to s_\infty}\infty\;,$$ which
shows that the sequence of tangents $(dy/ds)(s_j)$ at $(0,y_0)$
blows up as $j$ tends to infinity. This clearly implies that
$\gamma$ cannot be extended beyond $s_\infty$ as a $C^1$ curve.
\end{Example}
\end{coco}

When the metric is $C^2$,
the inverse
function theorem%
\footnote{It follows from the invariance-of-domain theorem that one can  construct normal coordinates for
$C^{1,1}$ metrics. However, those coordinates will only be continuous a priori,  which is a problem for some
arguments below; note that one cannot even calculate the metric functions in such coordinates.
It is conceivable that Clarke's implicit function
theorem~\cite{Clarke:optimization} could provide some more information in this context.
We
have not investigated this line of thought any further, as the approach in~\cite{ChGrant} provides
more general results anyway.}
shows that there
exists a neighborhood $\mcV_p\subset \mcU_p$ of the origin in $
\R^{\,\dim \mcM}$ on which the exponential map is a
diffeomorphism between $\mcV_p$ and its image
$$\mcO_p:=\exp_p (\mcV_p)\subset \mcM\;.$$
This allows one to define
 \emph{normal
coordinates} centred at $p$:

\begin{Proposition}
  \label{PC1}
  Let $(\mcM ,g)$ be a $C^3$ Lorentzian manifold with $C^{2}$ metric $g$.
  For every $p\in \mcM $ there exists an open coordinate neighborhood
  $\mcO_p$ of $p$, in which $p$ is mapped to the origin of
  $\R^{n+1}$, such that the coordinate rays $s\to sx^\mu$ are affinely parameterized
geodesics. If the metric $g$ is expressed in the resulting
coordinates $(x^\mu)=(x^0,\vec x)\in
\mcV_p$,
then
\begin{eqnarray}
\label{eq:PC1.1} & g_{\mu\nu}(0)=\eta_{\mu\nu}\;.  &
\end{eqnarray}
If $g$ is $C^3$ then we also have
\begin{eqnarray}
\label{eq:PC1.1b} &
\partial_\sigma g_{\mu\nu}(0)=0\;. &
\end{eqnarray}
Further, if the function $\sigma_p:\mcO_p\to\R$ is defined by
the formula
\bel{simpq}
\fbox{$\sigma_p(\exp_p(x^\mu)):=\eta_{\mu\nu}x^\mu
x^\nu\equiv-(x^0)^2+(\vec x)^2$}
 \;,
\ee
then
\begin{eqnarray}& \fbox{$\nabla\sigma_p\ \mbox{ is\ \ } \cases{
\mbox{ timelike } & $ \cases{ \mbox{past directed} & { on} $
\{q\,|\,\sigma_p(q)<0\;, \; x^0(q)<0\}$, \cr \mbox{future directed}
& { on} $ \{q\,|\,\sigma_p(q)<0\;, \; x^0(q)>0\}$,  }$ \cr \mbox{
null } & $\cases{ \mbox{past directed} & { on} $
\{q\,|\,\sigma_p(q)=0\;, \; x^0(q)<0\}$, \cr \mbox{future directed}
& { on} $ \{q\,|\,\sigma_p(q)=0\;, \; x^0(q)>0\}$,  }$ \cr \mbox{
spacelike } & \phantom{xxxxxxxxxxxxxxxxxi} on $  \{q\,|\,\sigma_p(q)>0\}$ .}$} \nn\\
&&\label{eq:PC1.3}
\end{eqnarray}
%
%
%
\end{Proposition}

\begin{Remark}
 \label{R25X11.1}
The definition of normal coordinates leads to $C^{k-1}$ coordinate functions if the metric is $C^k$. Hence
the metric, when expressed in normal coordinates, will be of $C^{k-2}$ differentiability class. This implies that there is a $C^2$ threshold for the introduction of normal coordinates, and that two derivatives are lost when expressing the metric in those coordinates. This is irrelevant for our purposes, as the main point here is that for $C^2$ metrics there  exists a function $\sigma_p$ satisfying \eq{P3a} below, together with the following three facts:
\begin{enumerate}
 \item $q\mapsto\sigma_p(q)$ is differentiable; \item $(q,p)\mapsto\sigma_p(q)$ is continuous; and
 \item if $g_n$ converges to $g$ in $C^2$, then the corresponding functions $
 \sigma_p$ converge in $C^0$.
\end{enumerate}
These are standard properties of solutions of ODEs (cf., e.g., \cite{Teschl}).
\end{Remark}

\begin{Remark} The coefficients of a Taylor expansion of
    $g_{\mu\nu}$ in normal
    coordinates can be expressed in terms of the Riemann tensor and
    its covariant derivatives\restrict{, see
    Section~\ref{Sshcan}} (compare~\cite{MSV,MoncriefNI}).
\end{Remark}

\proof Let us start by justifying that the implicit function
theorem can indeed be applied: Let $x^\mu$ be any coordinate
system around $p$, and let $e_a=e_a{}^\mu\partial_\mu$ be any
ON frame at $p$. Let
$$X=X^ae_a= X^ae_a{}^\mu \partial_\mu\in T_p \mcM$$ and let
 $x^\mu(s,X)$ denote the affinely parameterized
geodesic passing by $p$ at $s=0$, with tangent vector $$\dot
x^\mu(0,X):= \frac {dx^\mu(s,X)}{ds}\Big|_{s=0} =
X^ae_a{}^\mu\;.$$ Homogeneity properties of the ODE \eq{geodeq}
under the change of parameter $s\to \lambda s$ together with
uniqueness of solutions of ODE's show that for any constant
$a\ne 0$ we have
$$x^\mu(as,X/a) = x^\mu(s,X)\;.
$$ This, in turn, implies that there exist functions $\gamma^\mu$ such that
\bel{just} x^\mu(s,X) = \gamma^\mu(sX)\;.\ee From \eq{geodeq}
and \eq{just} we have
$$x^\mu(s,X) =x^\mu_0 + sX^ae_a{}^\mu + O((s|X|)^2)\;. $$ Here $x^\mu_0 $
are the coordinates of $p$, $|X|$ denotes the norm of $X$ with
respect to some auxiliary Riemannian metric on $M$, while the
$O((s|X|)^2)$ term is justified by \eq{just}. The usual
considerations of the proof that solutions of ODE's are
differentiable functions of their initial conditions show that
\beaa\frac {\partial x^\mu(s,X)}{\partial X^a} &=& \frac
{\partial (x^\mu_0 + sX^ae_a{}^\mu)}{\partial X^a}+ O(s^2)|X|
\\ & = & se_a{}^\mu + O(s^2)|X|\;.\eeaa
At $s=1$ one thus obtains \bel{just.b}\frac {\partial
x^\mu(1,X)}{\partial X^a}=e_a{}^\mu + O(|X|)\;.\ee This shows
that $\partial x^\mu/\partial X^a$ will be bijective at $X=0$
provided that $\det e_a{}^\mu\ne 0$. But this last inequality
can be obtained by taking the determinant of the equation
\bel{just.a} g(e_a,e_b)=g_{\mu\nu} e_a{}^\mu e_b^\nu\quad
\Longrightarrow \quad -1 = (\det g_{\mu\nu}) (\det
e_a{}^\mu)^2\;.\ee This justifies the use of the implicit
function theorem to obtain existence of the neighborhood
$\mcO_p$ announced in the statement of the proposition. Clearly
$\mcO_p$ can be chosen to be star-shaped with
respect to $p$. 
\Eq{just.b} and the fact that $e^\mu{}_a$ is an ON-frame show
that
$$
 g(\partial_a,\partial_b)\Big|_{X^a=0}=g_{\mu\nu} e_a{}^\mu e_b^\nu \Big|_{X^a=0}=
 \eta_{ab}\;,
$$
which establishes in \eq{eq:PC1.1}.

By construction the rays
$$
 s\to\gamma^a(s):=sX^a
$$
are affinely parameterized geodesics with tangent $\dot \gamma=
X^a\partial_a$, which gives
\beaa 0 = \left(\nabla_{\dot
\gamma}\dot \gamma \right)^a& = & \underbrace{\frac {d^2
(sX^a)}{ds^2}}_{=0} + \Gamma^a{}_{bc}(sX^d)X^b X^c
\\& = &
\Gamma^a{}_{bc}(sX^d)X^b X^c\;.\eeaa
Setting $s=0$ and differentiating this
equation twice with respect to $X^d$ and $X^e$ one obtains
$$\Gamma^a{}_{de}(0)=0\;.$$
The equation
$$0=\nabla_a g_{bc} = \partial_ag_{bc} - \Gamma^d{}_{ba}g_{dc} -
\Gamma^d{}_{ca}g_{bd}$$ evaluated at $X=0$ gives \eq{eq:PC1.1b}.

Let us pass now to the proof of the main point here, namely
\eq{eq:PC1.3}. From now on we will denote by $x^\mu$ the normal
coordinates obtained so far, and which were denoted by $X^a$ in
the arguments just done. For $x\in \mcO_p$ define \bel{PC1.3.1}
f(x) := \eta_{\mu\nu}x^\mu x^\nu\;,\ee and let
$\mcH_\tau\subset \mcO_p \setminus\{p\}$ be the level sets of
$f$: \bel{PC1.3.2} \mcH_\tau:=\{x: f(x)=\tau\;,\ x\ne 0\}\;.\ee
We will show that \bel{PC1.3.4} \mbox{the vector field
$x^\mu\partial_\mu$ is normal to the $\mcH_\tau$'s.}\ee Now,
$x^\mu\partial_\mu$ is tangent to the geodesic rays $s\to
\gamma^\mu(s):=sx^\mu$. As the causal character of the field of
tangents to a geodesic\footnote{Without loss of generality an
affine parameterization of a geodesic $\gamma$ can be chosen,
the result follows then from the calculation
$d(g(\dot\gamma,\dot\gamma))/ds=2g(\nabla_{\dot\gamma}\dot\gamma,\dot\gamma)=0$.}
is point-independent along the geodesic, we have \bel{PC1.3.3}
\begin{array}{rcl} x^\mu\partial_\mu \ \mbox{is timelike at} \ \gamma(s)&
\Longleftrightarrow & f(x)<0\;,
\\ x^\mu\partial_\mu \ \mbox{is null at} \ \gamma(s) &
\Longleftrightarrow & f(x)=0\;,\ x\ne 0\;,
\\
x^\mu\partial_\mu \ \mbox{is spacelike at} \ \gamma(s)&
\Longleftrightarrow & f(x)>0\;.
\end{array}
\ee This follows from the fact that the right-hand-side is
precisely the condition that the geodesic be timelike,
spacelike, or null, at $\gamma(0)$. Since $\nabla f$ is always
normal to the level sets of $f$, when \eq{PC1.3.4} holds we
will have \bel{PC1.3.4a} \mbox{ $x^\mu\partial_\mu$ is
proportional to $\nabla^\mu f$.}\ee This shows that
\eq{eq:PC1.3} will follow from \eq{PC1.3.3} when \eq{PC1.3.4}
holds.

It remains to establish \eq{PC1.3.4}. In order to do that,
consider any differentiable curve $\lambda\mapsto x^\mu(\lambda)$ lying on $\mcH_\tau$:
\bel{PC1.3.6} \eta_{\mu\nu}x^\mu(\lambda) x^\mu(\lambda)= \tau
 \quad \Longrightarrow \quad \eta_{\mu\nu}x^\mu(\lambda)
 \partial_\lambda x^\mu(\lambda)=0\;.
\ee
Let $\gamma^\mu(\lambda,s)$ be the following one-parameter
family of geodesic rays:
$$\gamma^\mu(\lambda,s):= s x^\mu(\lambda)\;.$$
For any function $f$ set $$ T(f)=\partial_s \left(f\circ
\gamma(s,\lambda)\right)\;,\quad X(f)=\partial_\lambda
\left(f\circ \gamma(s,\lambda)\right)\;,$$ so that
$$
 T(\lambda,s):= \Big(\partial_s \gamma^\mu(\lambda,s)
 \Big)\partial_\mu = x^\mu(\lambda,s) \partial_\mu\;, \quad X(\lambda,s):=
 \Big(\partial_\lambda \gamma^\mu(\lambda,s) \Big)\partial_\mu\;.
$$
For any fixed value of $\lambda$ the curves $s\to \gamma^\mu(\lambda,s)$ are geodesics,
which shows that $$\nabla_T T=0\;.$$ This gives $$\frac
{d(g(T,T))}{ds} = 2 g(\nabla_TT,T)=0\;,$$ hence
$$g(T,T)(s)=g(T,T)(0) = \eta_{\mu\nu} x^\mu(\lambda) x^\nu(\lambda)=\tau$$
by \eq{PC1.3.6}, in particular $g(T,T)$ is
$\lambda$-independent.

Next, for any twice-differentiable function $\psi$ we have
$$ [T,X](\psi):= T(X(\psi))-X(T(\psi)) = \partial _s \partial_\lambda
(\psi(\gamma^\mu(s,\lambda))) - \partial _\lambda \partial_s
(\psi(\gamma^\mu(s,\lambda))) =0 \;,
$$
because of the symmetry of the
matrix of second partial derivatives. It follows that
$$[T,X]=\nabla_T X-\nabla_X T = 0\;.$$
Finally, \beaa
\frac{d(g(T,X))}{ds}&=&g(\underbrace{\nabla_TT}_0,X)
+g(T,\underbrace{\nabla_TX}_{=\nabla_X T})\\
& = & g(T, \nabla_X T) = \frac 12 X(g(T,T)) = \frac 12
\partial_\lambda (g(T,T))=0\;.\eeaa This yields
$$g(T,X)(s,\lambda)= g(T,X)(0,\lambda)= \eta_{\mu\nu}x^\mu(\lambda) \partial_\lambda
x^\mu(\lambda)=0$$ again by \eq{PC1.3.6}. Thus $T$ is normal to
the level sets of $f$, which had to be established.
  \qed

As already pointed out, some regularity of the metric is
lost when going to normal coordinates; this can be avoided
using coordinates which are only approximately normal up to a
required order\restrict{ (compare Proposition~\ref{Pael1})}, which is
often sufficient for several purposes.

\begin{coco}
  It is sometimes useful to have a \emph{geodesic convexity
  property} at our disposal. This is made precise by the following
  proposition:

  \begin{Proposition}\label{Pgc} Let $\mcO$ be the domain of
  definition of a coordinate system $\{x^\mu\}$. Let $p\in \mcO$ and let
  $B_p(r)\subset \mcO$ denote an open coordinate ball of radius
  $r$
  centred at $p$. There exists $r_0>0$ such that every geodesic
  segment $\gamma:[a,b]\to \overline{B_p(r_0)}\subset \mcO$ satisfying
  $$\gamma(a),\gamma(b)\in B_p(r)\;,\quad r< r_0$$
  is entirely contained in $B_p(r)$.
  \end{Proposition}

  \proof Let $x^\mu(s)$ be the coordinate representation of
  $\gamma$, set
  $$f(s):= \sum_\mu (x^\mu-x^\mu_0)^2\;,$$
  where $x^\mu_0$ is the coordinate representation of $p$.
  We have
  \beaa \frac {df}{ds} & = &  2\sum_\mu (x^\mu-x^\mu_0)\frac
  {dx^\mu}{ds}\;,
  \\\frac {d^2f}{ds^2} & = &  2\sum_\mu \left(\frac
  {dx^\mu}{ds}\right)^2+ 2 \sum_\mu (x^\mu-x^\mu_0)\frac
  {d^2x^\mu}{ds^2}
  \\ & = &  2\sum_\mu \left(\frac
  {dx^\mu}{ds}\right)^2- 2 \sum_\mu (x^\mu-x^\mu_0)\Gamma^\mu_{\alpha\beta}\frac
  {dx^\alpha}{ds}\frac
  {dx^\beta}{ds}\;.
\eeaa Compactness of $\overline{B_p(r_0)}$ implies that there
exists a constant $C$ such that we have
$$\left|\sum_\mu (x^\mu-x^\mu_0)\Gamma^\mu_{\alpha\beta}\frac
  {dx^\alpha}{ds}\frac
  {dx^\beta}{ds}\right| \le C r_0\sum_\mu \left(\frac
  {dx^\mu}{ds}\right)^2\;.$$
It follows that $d^2f/ds^2\ge 0$ for $r_0$ small enough. This
shows that $f$ has no interior maximum if $r_0$ is
small enough, whence the result.
  \qed
\end{coco}
\ptcx{look up the uniform normal coordinates in Lee's Riemannian
  geometry}

It is convenient to introduce the following notion:

\begin{Definition}\label{Delem}
  An \emph{elementary region} is an open coordinate ball $\mcO$ within the domain of a normal
  coordinate neighborhood $\mcU_p$
  such that
  \begin{enumerate}\item
  $\mcO$ has compact closure $\overline\mcO$ in $\mcU_p$,
  and
  \item $\nabla t$ and $\partial_t$ are timelike on
      $\overline\mcU$. \ptcx{add geodesic convexity if
      needed?}
\end{enumerate}
\end{Definition}

Note that $\partial_t$ is timelike if and only if
$$g_{tt}=g(\partial_t,\partial_t)<0\;,$$
while $\nabla t$ is timelike if and only if
$$g^{tt}=g^\#(dt,dt)<0\;.$$
 Existence of elementary regions containing some point $p\in M$
follows immediately from Proposition \ref{PC1}: In normal
coordinates centred at $p$ one chooses $\mcO$ be an open
coordinate ball of sufficiently small radius.%
\restrict{ It follows from
Proposition~\ref{Pgc}
 \ptc{does not! but this is perhaps irrelevant?}
 that, by choosing the radius even smaller
if necessary, we can also require that every two points $p,q\in
\mcO$ be joined by a geodesic $\Gamma_{pq}$ contained in
$\mcO$. Further, $\Gamma_{pq}$ may be required to be unique
within the class of geodesics entirely contained in $\mcO$.
\ptc{might need elaborating upon}
}

\section{Causal paths}
 \label{SCp}

Let $(\mcM,\Lg)$ be a space-time. The basic objects in
causality theory are paths. We shall always use
\emph{parameterized paths}: by definition, these are
\emph{continuous} maps from some interval to $\mcM $. We will
use interchangedly the terms ``path", ``parameterized path", or
``curve", but we note that some authors make the distinction.
(For example, in~\cite{YCB:GRbook} a path is a map and a curve
 is the image of a path, oriented by parameterisation.)
 \ptcx{Yvonne has a nice terminology; this could be
 synchronized, but if so should be done globally through all folders as
both terms are used many times }

Let $\gamma:I\to\mcM$ and let $\mcU\subset \mcM$, we will write
$$\gamma \subset \mcU$$ whenever the image $\gamma(I)$ of $I$ by $\gamma$ is a
subset of $\mcU$.\ptcx{added, cross-check for repetitiveness
later} We will sometimes write $$\gamma\cap \mcU$$ for the path
obtained by removing from $I$ those parameters $s$ for which
$\gamma(s)\not\in\mcU$. Strictly speaking, our definition of a
path requires $\gamma$ to be defined on a connected interval,
so if the last construction gives several intervals $I_i$, then
$\gamma\cap \mcU$ will actually describe the collection of
paths $\gamma|_{I_i}$.

\begin{coco}
Some authors define a path in $\mcM $ as what would be in our terminology the
\emph{image of a
parameterized path}. In this approach one forgets about the
parameterization of $\gamma$, and identifies two paths which
have the same image and differ only by a reparameterization.
This leads to various difficulties when considering end points
of causal paths --- \emph{cf.}\/ Section~\ref{Sepeip}, or
limits of sequences of paths --- \emph{cf.}\/
Section~\ref{SAc}, and therefore we do \emph{not\/} adopt this
approach.
\end{coco}

 If $\gamma:I\to\mcM$ where $I=[a,b)$ or $I=
[a,b]$, then $\gamma(a)$ is called the \emph{starting point\/} of
the path $\gamma$, or of its image $\gamma(I)$. If $I=(a,b]$
or $I=[a,b]$, then $\gamma(b)$ is called the \emph{end point}.
(This definition will be extended in Section~\ref{Sepeip}, but
it is sufficient for the purposes here.) We shall say that
$\gamma:[a,b]\to \mcM$ is a path from $p$ to $q$ if
$\gamma(a)=p$ and $\gamma(b)=q$.

\begin{coco} In previous treatments of
causality
theory~\cite{PenroseDiffTopo,GerochDoD,HE,Wald:book,BONeill,Beem-Ehrlich:Lorentz2}
 one defines \emph{future directed timelike\/}
paths as those paths $\gamma$ which are \emph{piecewise
differentiable}, with $\dot \gamma$ timelike and future
directed wherever defined; at break points one further assumes
that both the left-sided and right-sided derivatives are
timelike. This definition turns out to be quite inconvenient
for several purposes. For instance, when studying the global
causal structure of space-times one needs to take limits of
timelike curves, obtaining thus --- by definition ---
\emph{causal future directed\/} paths. Such limits will not be
piecewise differentiable most of the time, which leads one to
the necessity of considering paths with poorer
differentiability properties. One then faces the unhandy
situation in which timelike and causal paths have completely
different properties. In several theorems separate proofs have
then to be given. The approach we present avoids this, leading
--- we believe --- to a considerable simplification of the
conceptual structure of the theory.\end{coco}

It is convenient to choose once and for all some auxiliary
Riemannian metric $\backmg$ on $\mcM$, such that
$(\mcM,\backmg)$ is complete --- such a metric always
exists~\cite{NomizuOzeki}; let $\distb$ denote the associated
distance function. A parameterized path $\gamma:I \to \mcM$
from an interval $I\subset \R$ to $M$ is called {\em locally
Lipschitzian\/} if for every compact subset $K$ of $I$ there
exists a constant $C(K)$ such that
$$\forall \ s_1,s_2\in K \quad \distb(\gamma(s_1),\gamma(s_2))\le
C(K) |s_1-s_2|\;.$$

\begin{coco}
It is natural to enquire whether the class of paths so defined
depends upon the background metric $\backmg$:

\begin{Proposition}
Let $h_1$ and $h_2$ be two complete Riemannian metrics on
$\mcM$. Then a path $\gamma:I\to \mcM$ is locally Lipschitzian
with respect to $h_1$ if and only if it is locally Lipschitzian
with respect to $h_2$.
\end{Proposition}

\proof Let $K\subset I$ be a compact set, then $\gamma(K)$ is
compact. Let $L_a$, $a=1,2$ denote the $h_a$-length of
$\gamma$, set
$$\mcK_a:=\cup_{s\in K}B_{h_a}(\gamma(s),L_a)\;,$$ where
$B_{h_a}(p,r)$ denotes a geodesic ball, with respect to the
metric $h_a$, centred at $p$, of radius $r$. Then the
$\mcK_a$'s are compact. Likewise the sets $$\hmcK_a\subset
T\mcM\;,$$ defined as the sets of $h_a$-unit vectors over
$\mcK_a$, are compact. This implies that there exists a
constant $C_K$ such that for all $X\in T_pM$, $p\in \mcK_a$, we
have
$$C_K^{-1}h_1(X,X)\le h_2(X,X) \le C_K h_1(X,X)\;.$$

Let $\gamma_{a,s_1,s_2}$ denote any minimising $h_a$-geodesic
between $\gamma(s_1)$ and $\gamma(s_2)$, then
$$\forall \ s_1,s_2\in K \qquad \gamma_{a,s_1,s_2}\subset \mcK_a\;.$$
This implies \beaa
\dist_{h_2}(\gamma(s_1),\gamma(s_2))&=&\int_{\sigma=\gamma_{2,s_1,s_2}}{\sqrt{h_2(
\dot \sigma,\dot \sigma)}}\\
&\ge &C_K^{-1}\int_{\sigma=\gamma_{2,s_1,s_2}}{\sqrt{h_1( \dot
\sigma,\dot \sigma)}}\\ &\ge
&C_K^{-1}\inf_\sigma\int_{\sigma}{\sqrt{h_1( \dot \sigma,\dot
\sigma)}}\\ &=
&C_K^{-1}\int_{\sigma=\gamma_{1,s_1,s_2}}{\sqrt{h_1( \dot
\sigma,\dot \sigma)}}\\
&=&C_K^{-1}\dist_{h_1}(\gamma(s_1),\gamma(s_2))\;. \eeaa From
symmetry with respect to the interchange of $h_1$ with $h_2$ we
conclude that for all $s_1,s_2\in K$
$$
C_K^{-1} \dist_{h_1}(\gamma(s_1),\gamma(s_2))\le
\dist_{h_2}(\gamma(s_1),\gamma(s_2))\le C_K
\dist_{h_1}(\gamma(s_1),\gamma(s_2))\;,
$$
and the result easily follows. \qed
\end{coco}

More generally, if $(N,k)$ and $(M,h)$ are Riemannian
manifolds, then a map $\phi$ is called \emph{locally
Lipschitzian}, or \emph{locally Lipschitz}, if for every
compact subset $K$ of $N$ there exists a constant $C(K)$ such
that
$$\forall \ p,q\in K \quad \dist_h(\phi(p),\phi(q))\le
C(K) \dist_k(p,q)\;.$$ A map is called \emph{Lipschitzian\/} if the
constant $C(K)$ above can be chosen independently of $K$.

The following important theorem of Rademacher will play a key
role in our considerations:

\begin{Theorem}[Rademacher]\label{TRademacher}
Let $\phi:M\to N$ be a locally Lipschitz map from a manifold
$M$ to a manifold $N$. Then:
\begin{enumerate}
\item $\phi$ is classically differentiable almost
    everywhere, with ``almost everywhere" understood in the
    sense of the Lebesgue measure in local coordinates on
    $M$.
\item The distributional derivatives of $\phi$ are in
    $L^\infty_\loc$ and are equal to the classical ones
    almost everywhere.
\item Suppose that $M$ is an open subset of $\R$ and $N$ is
    an open subset of\/ $\R^n$. Then $\phi$ is the integral
    of its distributional derivative,
    \bel{eTRad}\phi(x)-\phi(y) = \int^x_y \frac
    {d\phi}{dt}dt\;.\ee
\end{enumerate}
\end{Theorem}

\proof Point 1. is the classical statement of Rademacher, the
proof can be found in~\cite[Theorem~2, p.~235]{EvansGariepy}.
Point 2. is Theorem 5 in~\cite[p.~131]{EvansGariepy} and
Theorem~1 of~\cite[p.~235]{EvansGariepy}. (In that last theorem
the classical differentiability a.e. is actually established
for all $W^{1,p}_\loc$ functions with $p>n$). Point 3. can be
established by approximating $\phi$ by $C^1$ functions as
in~\cite[Theorem~1, p.~251]{EvansGariepy}, and passing to the
limit.\qed

 Point 2. shows that the usual properties of the
derivatives of continuously differentiable functions --- such
as the Leibniz rule, or the chain rule --- hold almost
everywhere for the derivatives of locally Lipschitzian
functions. By point 3. those properties can be used freely
whenever integration is involved.


We will use the symbol $$\dot \gamma$$ to denote the
\emph{classical}\ptcx{classical added, but maybe this is a bad
idea} derivative of a path $\gamma$, wherever defined. A
parameterized path $\gamma$ will be called {\em causal future
directed\/}
if $\gamma$ is locally Lipschitzian, with $\dot \gamma$
--- causal
 and future directed almost everywhere\footnote{Some authors
 allow constant paths to be causal, in which case the sets $J^\pm(\mcU;\mcO)$
defined below automatically contain $\mcU$. This leads to
unnecessary discussions when concatenating causal paths, so
that we find it convenient not to allow such paths in our
definition.}. Thus, $\dot \gamma$ is defined almost
everywhere; and it is causal future directed almost everywhere
on the set on which it is defined. A parameterized path
$\gamma$ will be called {\em timelike future directed} if
$\gamma$ is locally Lipschitzian, with $\dot \gamma$
--- timelike future directed almost everywhere. \emph{Past directed} parameterized paths are defined
by changing ``future" to ``past" in the definitions above.

\ptcx{make a formal statement; show that this leads to
equivalence of the notion of image and of the notion of
parameterized path, so that this is the same as what geometers
have been doing, except for a choice of origin on the path} A
useful property of locally Lischitzian paths is that they can
be parameterized by $h$-distance. Let $\gamma:[a,b)\to\mcM$ be
a path, and suppose that $\dot \gamma$ is non-zero almost
everywhere
--- this is certainly the case for causal paths. By
Rademacher's theorem the integral
$$s(t)=\int_a^t|\dot \gamma|_h(u)du
$$ is well defined. Clearly $s(t)$ is a continuous strictly
increasing function of $t$, so that the map $t\to s(t)$ is a
bijection from $[a,b)$ to its image. The new path $\hat \gamma
:=\gamma\circ s^{-1}$ differs from $\gamma$ only by a
reparametrization, so it has the same image in $\mcM$. The
reader will easily check that $|\dot{\hat \gamma}|_h=1$ almost
everywhere. Further, $\hat \gamma$ is Lipschitz continuous with
Lipschitz constant smaller than or equal to $1$: denoting by
$\distb$ be the associated distance function, we claim that
\bel{distbe} \distb(\hat\gamma(s),\hat\gamma(s'))\le |s-s'|\;.
 \ee
In order to prove \eq{distbe}, we calculate, for $s>s'$:
\bean s-s' & = & \int_{s'}^s  dt
 \\
 & = &
\int_{s'}^s \underbrace{\sqrt{h(\dot {\hat \gamma},\dot{
\gamma)}}(t)}_{=1\ \mbox{\scriptsize a.e.}} dt
 \nonumber
  \\
 & \ge  &
  \inf_{\tilde \gamma} \int_{\tilde \gamma} \sqrt{h(\dot
{\tilde \gamma},\dot{\tilde \gamma})}(t) dt =
\distb(\hat\gamma(s),\hat\gamma(s'))\;,
 \eeal{distba}
where the infimum is taken over $\tilde \gamma$'s which start
at $\gamma(s')$ and finish at $\gamma(s)$.

\section{Futures, pasts}
 \label{SCnpfp}

Let $\mcU\subset\mcO\subset \mcM$. One sets
\begin{eqnarray*}
 I^+(\mcU;\mcO)&:= &\{q\in\mcO  : 
 \mbox{\rm there exists a
timelike future directed path} \\ &  & \mbox{ from $\mcU$ to $q$
contained in $\mcO$}\} \;,
\\
 J^+(\mcU;\mcO)&:= &\{q\in\mcO : 
 \mbox{\rm there exists a
causal future directed path} \\ &  & \mbox{ from $\mcU$ to $q$
contained in $\mcO$}\} \cup \mcU\;.
\end{eqnarray*}
$I^-(\mcU;\mcO)$ and $J^-(\mcU;\mcO)$ are defined by replacing
``future" by ``past" in the definitions above. The set
$I^+(\mcU;\mcO)$ is called the \emph{timelike future of $\mcU$
in $\mcO$}, while $J^+(\mcU;\mcO)$ is called the \emph{causal
future of $\mcU$ in $\mcO$}, with similar terminology for the
timelike past and the causal past.  We will write $I^\pm(\mcU)$
for $I^\pm(\mcU;\mcM)$, similarly for $J^\pm(\mcU)$, and one
then omits the qualification ``in $\mcM$" when talking about
the causal or timelike futures and pasts of $\mcU$. We will
write $I^\pm(p;\mcO)$ for $I^\pm(\{p\};\mcO)$, $I^\pm(p)$ for
$I^\pm(\{p\};\mcM)$, \emph{etc.}

\begin{coco}
Although our definition of causal curves does not coincide with
the usual
ones~\cite{HE,PenroseDiffTopo,Beem-Ehrlich:Lorentz2,Wald:book},
it is equivalent to those. Indeed, it is easily seen that our
definition of $J^\pm$ is identical to the standard one. On the
other hand, the class of timelike curves as defined here is
quite wider then the standard one; nevertheless, the resulting
sets $I^\pm$ are again identical to the usual ones (compare
Proposition~\ref{CP3.1a}).
\end{coco}

It is legitimate to raise the question, why is it interesting
to consider sets such as $J^+(\mcO)$. The answer is two-fold:
From a mathematical point of view, those sets appear naturally
when describing the \emph{finite speed of propagation} property
of wave-type equations, such as Einstein's equations\restrict{, see
Section~\ref{}\ptc{where?} for details}. From a physical point
of view, such constructs are related to the fundamental
postulate of general relativity, that \emph{no signal can
travel faster than the speed of light}. This is equivalent to
the statement that the only events of space-times that are
influenced by an event $p\in \mcM$ are those which belong to
$J^+(\mcM)$.

\begin{Example}
\label{Ex1} Let $\mcM=S^1\times S^1$ with the flat metric
$g=-dt^2+d\varphi^2$. Geodesics of $g$ through $(0,0)$  are of
the form \be\label{Ex1a} \gamma(s)=(\alpha s \mod 2 \pi, \beta
s  \mod 2 \pi)\;, \ee
 where $\alpha$ and
$\beta$ are constants; the remaining geodesics are obtained by
a rigid translation of \eq{Ex1a}. Clearly any two points of
$\mcM$ can be joined by a timelike geodesic, which shows that
for all $p\in\mcM$ we have
$$I^+(p)=J^+(p)=\mcM\;.$$
It is of some interest to point out that for irrational
$\beta/\alpha$ in \eq{Ex1a} the corresponding geodesic is dense
in $\mcM$.
\end{Example}

There is an obvious \emph{meta-rule} in the theory of causality
that whenever a property involving $I^+$ or $J^+$ holds, then
an identical property will be true with $I^+$ replaced by
$I^-$, and with $J^+$ replaced by $J^-$, or both. This is
proved by changing the time-orientation of the manifold. Thus
we will only make formal statements for the futures.


 Example~\ref{Ex1} shows that in causally pathological
space-times the notions of futures and pasts need not to carry
interesting information. On the other hand those objects are
useful tools to study the global structure of those space-times
which possess reasonable causal properties.

We start with some elementary properties of futures and pasts:
\begin{Proposition}\label{PM1.a}
\Mcont
We have:
\begin{enumerate}\item $I^+(\mcU)\subset J^+(\mcU)$.
\item $p\in I^+(q)\quad\Longleftrightarrow\quad q\in
    I^-(p)$.
\item $\mcV\subset I^+(\mcU)\quad\Longrightarrow\quad
    I^+(\mcV)\subset I^+(\mcU)$.
\end{enumerate}
Similar properties hold with $I^+$ replaced by $J^+$.
\end{Proposition}

\proof 1. A timelike curve is a causal curve.

2. If $[0,1]\ni s \to \gamma(s)$  is a future directed causal
curve from $q$ to $p$, then $[0,1]\ni s\to \gamma(1-s)$ is a
past directed causal curve from $p$ to $q$.

3. Let us start by introducing some notation: consider
$\gamma_a:[0,1]\to \mcM$, $a=1,2$, two causal curves such that
$\gamma_1(1)=\gamma_2(0)$. We define the \emph{concatenation
operation} $\gamma_1\cup \gamma_2$ as follows: \bel{concat}
(\gamma_1\cup\gamma_2)(s)=\cases{ \gamma_1(s)\;, & $s\in
[0,1]\;,$ \cr \gamma_2(s-1)\;, & $s\in [1,2]\;.$ } \ee There is
an obvious extension of this definition when the ranges of
parameters of the $\gamma_a$'s are not $[0,1]$, or when a
finite number $i\ge 3$ of paths is considered, we leave the
formal definition to the reader.

Let, now, $r\in I^+(\mcV)$, then there exists $q\in \mcV$ and a
future directed timelike curve $\gamma_2 $ from $q$ to $r$.
Since $\mcV\subset I^+(\mcU)$ there exists a future directed
timelike curve $\gamma_1$ from some point $p\in \mcU$ to $q$.
Then the curve $\gamma_1\cup\gamma_2$ is a future directed
timelike curve from $\mcU$ to $r$. \qed

We have the following, intuitively obvious, description of
futures and pasts of points in Minkowski space-time (see Figure~\ref{Flightcone});
 in
Proposition~\ref{P3} below we will shortly prove a similar
local result in general space-times, with a considerably more
complicated proof.
\begin{figure}[ht]
\hspace{-2cm}{
 \psfrag{p}{\Large $p$}
 \psfrag{futuret}{\Large future pointing timelike}
 \psfrag{pastt}{\Large past pointing timelike}
 \psfrag{futuren}{\Large future pointing null}
 \psfrag{pastn}{\Large past pointing null}
 \resizebox{4in}{!}{\includegraphics{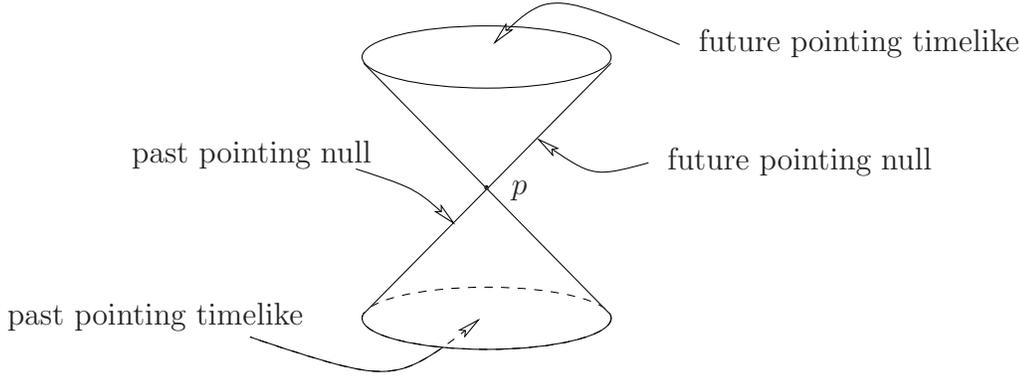}}
}
\caption{The light cone at $p$.
\label{Flightcone}
}
\end{figure}

\begin{Proposition}\label{PM1}
Let $(\mcM,g)$ be the $(n+1)$-dimensional Minkowski space-time
$\R^{1,n}:=(\R^{1+n},\eta)$, with Minkowskian coordinates
$(x^\mu)=(x^0,\vec x)$ so that
$$\eta(\partial_\mu,\partial_\nu)=\diag(-1,+1,\ldots,+1)\;.$$ Then
\begin{enumerate}
\item $I^+(0)=\{x^\mu:\eta_{\mu\nu}x^\mu x^\nu<0, \ x^0>0\}$,
\item $J^+(0)=\{x^\mu:\eta_{\mu\nu}x^\mu x^\nu\le 0,\ x^0\ge
    0\}$,
\item in particular the boundary $\dotJ^+(0)$ of $J^+(0)$
    is the union of $\{0\}$ together with all null future
    directed geodesics with initial point at the origin.
\end{enumerate}
\end{Proposition}

\proof Let $\gamma(s)=(x^\mu(s))$ be a parameterized causal
path in $\R^{1,n}$ with $\gamma(0)=0$. At points at which
$\gamma$ is differentiable we have
$$\eta(\dot\gamma,\dot\gamma)=-\left(\frac{dx^0}{ds}\right)^2 +\left|\frac{d\vec
x}{ds}\right|_\delta^2\le 0\;, \quad \frac{dx^0}{ds}\ge
\left|\frac{d\vec x}{ds}\right|_\delta^2\ge 0 \;.$$ Now, similarly
to a differentiable function, a locally Lipschitzian function
is the integral of its distributional derivative (see
Theorem~\ref{TRademacher}) hence
\begin{deqarr}
x^0(s) &=& \int_0^s \frac{dx^0}{ds} (u) du \label{PM1a}
\\
& \ge & \int_0^s  \left|\frac{d\vec x}{ds}\right|_\delta (u) du
=:\ell(\gamma_s)\;.\label{PM1b}\arrlabel{eqPM1}
\end{deqarr}
Here $\ell(\gamma_s)$ is the length, with respect  to the flat
Riemannian metric $\delta$, of the path $\gamma_s$, defined as
$$[0,s]\ni u \to \vec x(u) \in \R^n\;.$$ Let $\dist_\delta$
denote the distance function of the metric $\delta$, thus
$$\dist_\delta(\vec x,\vec y) = |\vec x - \vec y|_\delta\;,$$
it is well known that $$\ell_s\ge \dist_\delta(\vec x(s),\vec
x(0)) = |\vec x(s) - \vec x(0)|_\delta= |\vec x(s)
|_\delta\;.$$ Therefore
$$x^0(s)\ge|\vec
x(s) |_\delta\;, $$  and point 2. follows. For timelike curves the
same proof applies, with all inequalities becoming strict,
establishing point 1. Point 3. is a straightforward consequence
of point 2. \qed

\begin{coco}
There is a natural generalisation of Proposition~\ref{PM1} to
the following class of metrics on $\R\times \hyp$:
\bel{warped}g =-\varphi dt^2 + h\;,\qquad \partial_t \varphi =
\partial_t h = 0\;,\ee
where $h$ is a Riemannian metric on $\hyp$, and $\varphi$ is a
strictly positive function. Such metrics are sometimes called
\emph{warped-products},\index{warped product} with
\emph{warping function}\index{warping function} $\varphi$.

\begin{Proposition}
\label{PM1warp} Let $\mcM=\R\times \hyp$ with the metric
\eq{warped}, and let $p\in \hyp$. Then $J^+((0,p))$ is the
graph over $\hyp$ of the distance function $\dist _{\hat
h}(p,\cdot)$ of the \emph{optical metric} $$\hat
h:=\varphi^{-1}h\;,$$ while $I^+((0,p))$ is the epigraph of
$\dist_{\hat h}(p,\cdot)$,
$$I^+((0,p))=\{(t,q):t>\dist_{\hat h}(p,q)\}\;.$$
\end{Proposition}
\proof Since the causal character of a curve is invariant under
conformal transformations, the causal and timelike futures with
respect to the metric $g$ coincide with those with respect to
the metric
$$\varphi^{-1}g= -dt^2+\hat h\;.$$ Arguing as in the
the proof of Proposition~\ref{PM1}, \eq{PM1} becomes
\begin{deqarr}
x^0(s) &=& \int_0^s \frac{dx^0}{ds} (u) du \label{PM1an}
\\
& \ge & \int_0^s  \left|\frac{d\vec x}{ds}\right|_{\hat h} (u)
du =:\ell_{\hat h}(\gamma_s)\;,\label{PM1bn}
\end{deqarr}
where $\ell_{\hat h}(\gamma_s)$ denotes the length of
$\gamma_s$ with respect to $\hat h$, and one concludes as
before. \qed
\end{coco}

The next result shows that, locally, causal behaviour is
identical to that of Minkowski space-time. The proof of this
``obvious" fact turns out to be surprisingly involved:

\begin{Proposition}\label{P3}
\MCtwok
Let $\mcO_p$ be a domain of normal coordinates $x^\mu$ centered
at $p\in\mcM$ as in Proposition~\ref{PC1}. Let
$$\mcO\subset\mcO_p\;$$
be any normal--coordinate ball such that $\nabla x^0$ is
timelike on $\mcO$. Recall (compare \eq{simpq}) that the
function $\sigma_p:\mcO_p\to\R$ has been defined by the formula
\bel{lorp} \sigma_p(\exp_p(x^\mu)):=\eta_{\mu\nu}x^\mu
x^\nu\;.\ee Then
\be\label{P3a}
 \mcO\ni q=\exp_p(x^\mu) \in \cases{
I^+(p;\mcO) & $ \Longleftrightarrow \ \ \sigma_p(q) < 0\;, \
x^0> 0\;$, \cr J^+(p;\mcO) & $ \Longleftrightarrow \ \
\sigma_p(q) \le 0\;, \ x^0\ge 0\;$,
 \cr \dotJ^+(p;\mcO) & $
\Longleftrightarrow \ \ \sigma_p(q) = 0\;, \ x^0\ge 0\;$,} \ee
with the obvious analogues for pasts. In particular, a point
$q=\exp_p(x^\mu)\in \dotJ^+(p;\mcO_p)$ if and only if $q$  lies
on the null geodesic segment $[0,1]\ni s\to \gamma(s)=
\exp_p(sx^\mu)\in\dotJ^+(p;\mcO_p)$.
\end{Proposition}

\begin{Remark}\label{RP3a}
Example~\ref{Ex1} shows that $I^\pm{(p;\mcO)}$, {\em etc.},\/
cannot be replaced by $I^\pm{(p)}$, because causal paths
through $p$ can exit $\mcO_p$ and reenter it; this can actually
happen again and again.
\end{Remark}

Before proving Proposition~\ref{P3}, we note the following
straightforward implication thereof:

\begin{Proposition}\label{P3.3.7}
%
Let $\mcO$ be as in  Proposition~\ref{P3}, then
$I^+(p;\mcO) $ is open.
%
\qed
\end{Proposition}

\noindent {\sc Proof of Proposition~\ref{P3}:} As the
coordinate rays are geodesics, the implications ``$\Leftarrow$"
in \eq{P3a} are obvious. It remains to prove ``$\Rightarrow$".
We start with a lemma:

\begin{Lemma}
 \label{LP3.0}
Let  $\tau$ be a \emph{time function, i.e.},\/ a
differentiable\ptcx{do not impose past pointing gradient ?}
function with \underline{timelike} \underline{past-pointing}
gradient. For any $\tau_0$, a future directed causal path
$\gamma$ cannot leave the set $\{q:\tau(q)>\tau_0\}$; the same
holds for sets of the form $\{q:\tau(q)\ge\tau_0\}$. In fact,
$\tau$ is non-decreasing along $\gamma$, strictly increasing if
$\gamma$ is timelike.
\end{Lemma}

\proof Let $\gamma:I\to \mcM$ be a future directed
parameterized causal path, then $\tau\circ\gamma$ is a locally
Lipschitzian function, hence equals the integral of its
derivative on any compact subset of its domain of definition,
so that
\begin{eqnarray} \tau(\gamma(s_2))-\tau(\gamma(s_1)) &=& \int_{s_1}^{s_2} \frac{d(\tau\circ\gamma)}{du} (u) du
\nn 
\\
& =& \int_{s_1}^{s_2} \langle d\tau,\dot \gamma\rangle (u)du
\nn 
\\
& =& \int_{s_1}^{s_2} g(\nabla \tau,\dot \gamma)(u) du \ge
0\;,\label{P3a2}
\end{eqnarray}
since $\nabla \tau$ is timelike past directed, while
$\dot\gamma$ is causal future directed or zero wherever
defined. The function $s\to \tau(\gamma(s))$ is strictly
increasing when $\gamma$ is timelike, since then the integrand
in \eq{P3a2} is strictly positive almost everywhere.
\qed

Applying Lemma~\ref{LP3.0} to the time function $x^0$ we obtain
the claim about $x^0$ in \eq{P3a}. To justify the remaining
claims of Proposition~\ref{P3}, we recall
Equation~\eq{eq:PC1.3}
\begin{eqnarray}
& \nabla\lorp\ \mbox{ is\ \ } \cases{ \mbox{ timelike future
directed } & on $ \{q:\lorp(q)<0\;,\; x^0(q)>0\}$ , \cr \mbox{
null future directed} & on $ \{q:\lorp(q)=0\;,\;x^0(q)>0\}$ .}
 \phantom{xxxx}
\label{eq:PC1.4}
\end{eqnarray}
Let $\gamma=(\gamma^\mu):I\to\mcO$ be a parameterized future
directed causal path with $\gamma(0)=p$, then $\lorp\circ
\gamma$ is a locally Lipschitzian function, hence
\begin{eqnarray}
\nn \lorp\circ\gamma(t)&=&\int_0^t
\frac{d(\lorp\circ\gamma)(s)}{ds}ds
\\ &=&\int_0^t g(\nabla\lorp,\dot \gamma)(s)ds
\;.\label{P3b}
\end{eqnarray}We note the following:

\begin{Lemma}\label{LP3}
A future directed causal path $\gamma\subset\mcO_p$ cannot
leave the set $\{q:x^0(q)>0,\lorp(q)<0\}$.
\end{Lemma}

\proof The time function $x^0$ remains positive along $\gamma$
by Lemma~\ref{LP3.0}. If $-\sigma_p$ were also a time function
we would be done by the same argument. The problem is  that
$-\sigma_p$ is a time function only on the set where $\sigma_p$
is negative, so some care is needed; we proceed as follows: The
vector field $\nabla\lorp$ is causal future directed on
$\{x^0>0,\eta_{\mu\nu}x^\mu x^\nu\le 0\}$, while $\dot \gamma$
is causal future directed or zero wherever defined, hence
$g(\nabla\lorp,\dot \gamma)\le 0$ as long as $\gamma$ stays in
$\{x^0>0,\eta_{\mu\nu}x^\mu x^\nu\le 0\}$. By \Eq{eq:PC1.4} the
function $\lorp$ is non-increasing along $\gamma$ as long as
$\gamma$ stays in $\{x^0>0,\eta_{\mu\nu}x^\mu x^\nu\le0\}$.
Suppose that $\lorp(\gamma(s_1))<0$ and let
$$s_*=\sup\{u\in I:\lorp(\gamma(s))<0 \mbox{ \ on } \
[s_1,u]\}\;.$$ If $s_*\in I$, then $\lorp\circ\gamma(s_*)=0$ and
$\lorp\circ\gamma$ is \emph{not} non-increasing on $[s_1,s_*)$,
which is not possible since $\gamma(s)\in
\{x^0>0,\eta_{\mu\nu}x^\mu x^\nu\le0\}$ for $s\in [s_1,s_*)$.
It follows that $\lorp\circ\gamma< 0$, as desired. \qed

Proposition~\ref{P3} immediately follows for those future
direct causal paths through $p$ which \emph{do enter} the set
$\{\eta_{\mu\nu}x^\mu x^\nu<0\}$. This is the case for
$\gamma$'s such that $\dot
\gamma(0)=\left(\dot\gamma^\mu(0)\right)$ exists and is
timelike: We then have $$\gamma^\mu(s)= s\dot \gamma^\mu(0)
+o(s)\;,$$ hence
$$\eta_{\mu\nu}\gamma^\mu(s)\gamma^\nu(s)=s^2\eta_{\mu\nu}\dot\gamma^\mu(0)\dot\gamma^\nu(0)
+o(s^2) <0$$ for $s$ small enough. It follows that $\gamma$ enters
the set $\{\eta_{\mu\nu}x^\mu
x^\nu<0\}\equiv\{q:\sigma_p(q)<0\}$, and  remains there for
$|s|$ small enough. We conclude using Lemma~\ref{LP3}.

We continue with arbitrary parameterized future directed
\emph{timelike} paths $\gamma:[0,b)\to\mcM$, with
$\gamma(0)=p$, thus
$\dot \gamma$ exists and is timelike future directed for almost
all $s\in [0,b)$. In particular there exists a sequence
$s_i\to_{i\to\infty}0$ such that $\dot\gamma(s_i)$ exists and
is timelike.

Standard properties of solutions of ODE's show that for each
$q\in\mcO_p$ there exists a neighborhood $\mcW_{p,q}$ of $p$
such that the function $$\mcW_{p,q}\ni r\to \lorentz_r(q)$$  is
defined, 
continuous in $r$.
For $i$ large enough we will have  $\gamma(s_i)\in
\mcW_{p,\gamma(s)}$; for such $i$'s  we have just shown that
$$\lorentz_{\gamma(s_i)}(\gamma(s))<0
\;. $$
Passing to the limit $i\to\infty$, by continuity one obtains
\be\label{P3c} \lorentz_{p}(\gamma(s))\le 0 \;, \ee thus
\be\label{P3d}\gamma\subset\{x^0\ge0,\eta_{\mu\nu}x^\mu
x^\nu\le 0\}\;.\ee Since $\dot\gamma$ is timelike future
directed wherever defined, and $\nabla\lorp$ is causal future
directed on $\{x^0> 0, \eta_{\mu\nu}x^\mu x^\nu\le 0\}$,
\Eqsone{P3b} and \eq{P3d} show that the inequality in \eq{P3c}
must be strict.

To finish the proof, we reduce the general case to the last one
by considering perturbed metrics, as follows: let $e_0$ be any
unit timelike vector field on $\mcO$ ($e_0$ can, {\em e.g.},\/
be chosen as $\nabla x^0/\sqrt{-g(\nabla x^0,\nabla x^0)}$),
for $\epsilon
> 0$ define a family of Lorentzian metrics $g_\epsilon$ on $\mcO$
by the formula
$$g_\epsilon(X,Y)=g(X,Y)-\epsilon g(e_0,X)g(e_0,Y)\;.$$
Consider any vector $X$ 
which is causal for $g$, then
\begin{eqnarray*}
g_\epsilon(X,X)&=& g(X,X) - \epsilon (g(e_0,X))^2
\\
&\le &- \epsilon (g(e_0,X))^2<0\;,
\end{eqnarray*}
so that $X$ is timelike for $g_\epsilon$. Let
$\lorentz(g_\epsilon)_p$ be the associated functions defined as
in \Eq{lorp}, where the exponential map there is the one
associated to the metric $g_\epsilon$. Standard properties of
solutions of ODE's (see, e.g., \cite{Teschl})
imply
 that  for any compact subset $K$ of
$\mcO_p$ there exists an $\epsilon_K>0$ and a neighborhood
$\mcO_{p,K}$ of $K$ such that for all $\epsilon\in
[0,\epsilon_K]$ the functions
$$\mcO_{p,K}\ni q\to\lorentz(g_\epsilon)_p(q)$$ are defined, and
depend continuously upon $\epsilon$. We take $K$ to be
$\gamma([0,s])$, where $s$ is such that $[0,s]\subset I$, and
consider any $\epsilon $ in $(0,\epsilon_{\gamma([0,s])})$.
Since $\gamma$ is timelike for $g_\epsilon$, the results
already established show that we have
$$\lorentz(g_\epsilon)_p(\gamma(s))<0\;.$$
Continuity in $\epsilon$ implies
$$\lorp(\gamma(s))\le 0\;.$$
Since $s$ is arbitrary in $I$, Proposition~\ref{P3} is
established.
 \qed

\begin{coco}
For certain considerations it is useful to have the following:

\begin{Corollary}\label{CP3.1}Let $\mcO_p$ be a domain of normal coordinates
$x^\mu$ centered at $p\in\mcM$ as in Proposition~\ref{PC1}, and
let $\mcO\subset\mcO_p\;$ be any normal--coordinates ball such
that $\nabla x^0$ is timelike on $\mcO$. If $\gamma \subset
\mcO$ is a causal curve from $p$ to
$$q=\exp_p(x^\mu)\in \dotJ^+(p;\mcO)\;,$$ then $\gamma$ lies
entirely in $\dotJ^+(p;\mcO)$, and there exists a
reparameterization $s\to r(s)$ of $\gamma$ so that
$\gamma$ is a null-geodesic segment through $p$:
$$
 [0,1]\ni s\to \gamma(r(s))=\exp_p(sx^\mu)
 \;.
$$
\end{Corollary}

\proof Proposition~\ref{P3} shows that $\sigma_p(q)=0$. It
follows that
\begin{eqnarray}  0=\lorp\circ\gamma(t) =\int_0^t
g(\nabla\lorp,\dot \gamma)(s)ds \;.\label{P3b1a} \end{eqnarray}
Since $\nabla \lorp$ and $\dot \gamma$ are causal oppositely
directed we have $g(\nabla\lorp,\dot \gamma)\ge 0$ almost
everywhere. It thus follows from \eq{P3b1a} that
$$
g(\nabla\lorp,\dot \gamma)= 0
$$
 almost everywhere. This is only
possible if
\bel{nullParl} \nabla \lorp \sim \dot \gamma
 \ee
almost everywhere\restrict{ (see Proposition~\ref{PinvCS})}.
Thus $\nabla \lorp$ is null a.e. along $\gamma$. But $\nabla
\lorp$ is null only on  $\dotJ^+(p;\mcO)$, thus $\gamma(s)\in
\dotJ^+(p;\mcO)$ a.e.; by continuity this is true for all $s$,
and we have shown that $\gamma$ lies entirely on
$\dotJ^+(p;\mcO)$.

To continue, in normal coordinates \eq{nullParl} reads
$$
 \frac{d x^\mu(s)}{ds} = f(s)\,x^\mu(s)
$$
a.e., for some strictly positive function $f\in L^\infty$.
Define $r(s)$ by
$$
r(s)=\int _ 0^s f(t)dt \;,
$$
then $r$ is strictly increasing, hence a bijection from the
interval of definition of $\gamma$ to some interval $[0,r_0]$.
Further $r$ is Lipschitz, differentiable on a set of full
measure, and on that set it holds
$$
\frac{dr}{ds}= f
 \;.
$$
Now, this equation shows that the set where the map $r\mapsto s(r)$ might
fail to be differentiable is the image by $r$ of the set
$\Omega_1$ where $r$ fails to be differentiable, together with
the image by $r$ of the set of points $\Omega_2$ where $f$
vanishes. Both $\Omega_1$ and $\Omega_2$ have zero measure, and
the image of a negligible set by a Lipschitz map is a
negligible set. We thus obtain, almost everywhere,
\bel{nullParl2}
 \frac{d x^\mu(s(r))}{dr} = \frac{d x^\mu(s(r))}{ds}\frac{ds}{dr} = x^\mu(s(r))\;,
 \ee
so $dx^\mu/dr$ can be extended by continuity to a continuous
function.
It easily follows that $x^\mu(s(\lambda))$ is $C^1$, and the
result is obtained by integration of \eq{nullParl2}.
 \qed
\end{coco}

\begin{coco}

Penrose's approach~\cite{PenroseDiffTopo} to the theory of
causality is based on the notion of \emph{timelike or causal
trips}: by definition, a causal trip is a piecewise broken
causal geodesic. The following result can be used to show
equivalence of the definitions of $I^+$, \emph{etc.}, given
here, to those of Penrose:
 \begin{Corollary}\label{CP3.1a}
 \MCtwok
 If $q\in I^+(p)$, then there exists a future directed
 piecewise broken future directed timelike geodesic from $p$ to
 $q$. Similarly, if $q\in J^+(p)$, then there exists a future directed
 piecewise broken future directed causal geodesic from $p$ to
 $q$.
 \end{Corollary}

 \proof
 Let $\gamma:[0,1]\to\mcM$ be a parameterized future directed
 causal path with $\gamma(0)=p$ and  $\gamma(1)=q$. Continuity of $\gamma$ implies that for every
 $s\in[0,1]$ there exists $\epsilon_s>0$ such that $$\gamma(u)\in
 \mcO_{\gamma(s)}$$ for all $$u\in (s-2\epsilon_s,s+2\epsilon_s)\cap [0,1]=\left[\max(0,s-2\epsilon_s),
 \min(1,s+2\epsilon_s)\right]\;,$$
 where $\mcO_r$ is a  normal-coordinates ball centred at $r$, and satisfying the requirements
 of Proposition~\ref{P3}. Compactness of
$[0,1]$ implies that from the covering
 $\{\left(s-\epsilon_s,
 s+\epsilon_s\right)\}_{s\in[0,1]}$ a
 finite covering $\{\left(s_i-\epsilon_{s_i},
 s_i+\epsilon_{s_i}\right)\}_{i=0,\ldots,N}$ can be
 extracted, with $s_0=0$, $s_N=1$. Reordering the $s_i$'s if necessary we may assume that $s_i<s_i+1$.
 By definition we have
 $$\gamma|_{[s_i,s_{i+1}]}\subset\mcO_{\gamma(s_{i})}\;,$$
 and by Proposition~\ref{P3} there exists a causal future directed
 geodesic segment from $\gamma(s_{i})$ to $\gamma(s_{i+1})$: if
 $\gamma(s_{i+1})= \exp_{\gamma(s_{i})}(x^\mu)$, then the required
 geodesic segment is given by $$[0,1]\ni s
 \to\exp_{\gamma(s_{i})}(sx^\mu)\;.$$ If $\gamma$ is timelike, then
 all the segments are timelike.
 Concatenating the segments together provides the claimed
 piecewise broken geodesic.
 \qed
 \end{coco}

Proposition~\ref{PM1} shows that the sets $I^\pm(p)$ are open
in Minkowski space-time. Similarly it follows from
Proposition~\ref{P3} that the sets $I^\pm(p;\mcO_p)$ are open.
This turns out to be true in general:

\begin{Proposition}\label{P2}
\MCtwok
For all $\mcU\subset \mcM$ the sets
$I^{\pm}(\mcU)$ are open.
\end{Proposition}

\proof Let $q\in I^{+}(\mcU)$, and let, as in the proof of
Corollary~\ref{CP3.1a}, $s_{N-1}$ be such that $q\in
\mcO_{\gamma(s_{N-1})}$. Then
$$ \mcO_{\gamma(s_{N-1})}\cap I^+(\gamma(s_{N-1});\mcO)$$
is an open neighborhood of $q$ by Corollary~\ref{P3.3.7}.
Clearly
$$I^+(\gamma(s_{N-1});\mcO)\subset I^+(\gamma(s_{N-1}))\;.$$
Since $\gamma(s_{N-1}) \in I^{+}(\mcU)$ we have
$$ I^+(\gamma(s_{N-1}))\subset I^+(\mcU)$$
(see point 3. of Proposition~\ref{PM1.a}). It follows that
$$ \mcO_{\gamma(s_{N-1})}\cap I^+(\gamma(s_{N-1});\mcO)\subset I^+(\mcU)\;,$$
 which implies our claim. \qed

 \ptcx{material on causality for continuous metrics in
 continuous.tex, what's a good place for it to go?}

In Minkowski space-time the sets $J^{\pm}(p)$ are closed, with
\be\label{P3e} \overline{I^{\pm}(p)}= J^{\pm}(p)\;. \ee We will
show below (see Corollary~\ref{Cpushup}) that we always have
\be
 \label{P3e2} \overline{I^{\pm}(p)}\supset J^{\pm}(p)\;,
 \ee
but this requires some work. Before proving~\eq{P3e2}, let us
point out that \eq{P3e} does not need to be true in general:

\begin{Example}\label{Ex2}
Let $(\mcM,g)$ be the two-dimensional Minkowski space-time
$\R^{1,1}$ from which the set $\{x^0=1,x^1\le -1\}$ has been
removed. Then
$$
J^{+}(0;\mcM)= J^{+}(0,\R^{1,1})\setminus
\{x^0=-x^1\;,x^1\in(-\infty,1]\}\;,
$$
\emph{cf.} Figure~\ref{FC3a}.
\begin{figure}[htbp]
\begin{center}
\input{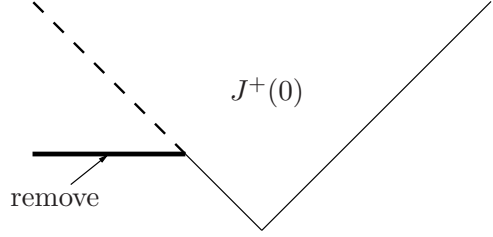}
\end{center}
 \caption{\label{FC3a}$J^+(p)$ is not closed unless
further causal regularity conditions are imposed on
$(\mcM,g)$.}
\end{figure} hence $J^+(0;\mcM)$ is neither open nor closed, and
\Eq{P3e} does not hold.
\end{Example}

We have the following:

\begin{Lemma}[``Push-up Lemma 1"]
\label{Lpushup0}
\MCtwok
For any $\Omega\subset\mcM$ we have
\bel{triinc}I^+(J^+(\Omega))= I^+(\Omega)\;.\ee
\end{Lemma}

\begin{Remark}
 \label{RLpushup0} In~\cite{ChGrant} an example is presented which shows that
the result is wrong for metrics which are merely continuous.
 \ptcx{synchronize}
\end{Remark}

\proof The obvious property $$\mcU\subset \mcV  \quad
\Longrightarrow \quad I^+(\mcU)\subset I^+( \mcV)\;$$ provides
inclusion of the right-hand-side of \eq{triinc} into the
left-hand-side. It remains to prove that
$$I^+(J^+(\Omega))\subset I^+(\Omega)\;.$$
Let $r\in I^+(J^+(\Omega))$, thus there exists a past-directed
timelike  curve $\gamma_0$ from $r$ to a point $q\in
J^+(\Omega)$. Since $q\in J^+(\Omega)$, then either $q\in
\Omega$, and there is nothing to prove, or there exists a
past-directed causal curve $\gamma:I\to\mcM$ from $q$ to some
point $p\in\Omega$. We want to show that there exists a
past-directed timelike curve $\hat \gamma$ starting at $r$ and
ending at $p$. The curve $\hat \gamma$ can be obtained by
``pushing-up" $\gamma$ slightly, to make it timelike, the
construction proceeds as follows: Using compactness, we cover
$\gamma$ by a  finite collection $\mcU_i$, $i=0,\cdots,N$, of
elementary regions $\mcU_i$ centered at $p_i\in\gamma(I)$, with
$$p_0= q\;, \quad p_i\in \mcU_{i}\cap\mcU_{i+1}\;, \quad p_{i+1}\subset J^-{(p_i)}\;, \quad p_N=p\;.$$
Let $\gamma_0:[0,s_0]\to\mcM$ be the already mentioned causal
curve from $r$ to $q\in\mcU_1$; let $s_1\ne s_0$ be close enough to
$s_0$ so that $\gamma_0(s_1)\in\mcU_1$. By
Proposition~\ref{PC1} together with the definition of
elementary regions
there exists a past directed timelike curve $\gamma_1:[0,1]\to\mcU_1$ from $\gamma_0(s_1)$ to
$p_1\in\mcU_1\cap\mcU_2$. For $s$ close enough to $1$ the curve $\gamma_1$ enters
$\mcU_2$, choose an $s_2\ne 1$ such that $\gamma_1(s_2)\in \mcU_2$,
again by Proposition~\ref{PC1} there exists a a past directed
timelike curve $\gamma_{2}:[0,1]\to\mcU_2$
from $\gamma_1(s_2)$ to
$p_2$. One
repeats that construction  iteratively obtaining a (finite) sequence of
past-directed
timelike curves $\gamma_i\subset I^+(\gamma)\cap \mcO$ such that the end point $\gamma_i(s_{i+1})$
of $\gamma_i|_{[0,s_{i+1}]}$
coincides with the starting point of $\gamma_{i+1}$.
Concatenating those curves together gives the desired path $\hat \gamma$.
\qed

We have the following, slightly stronger, version of
Lemma~\ref{Lpushup0}, which gives a sufficient condition to be
able to deform a \emph{causal} curve to a \emph{timelike} one,
keeping the deformation as small as desired:

\begin{Corollary}\label{CPushup0}
Let the metric be twice differentiable.
Consider a  \underline{causal} future directed curve
$\gamma:[0,1]\to\mcM$ from $p$ to $q$. If there exist
$s_1<s_2\in[0,1]$ such that $\gamma|_{[s_1,s_2]}$ is
\underline{timelike}, then in any neighborhood $\mcO$ of
$\gamma$ there exists a \underline{timelike} future directed
curve $\hat\gamma$ from $p$ to $q$.
\end{Corollary}

\begin{Remark}\label{Rpushup} The so-called \emph{maximising}
null geodesics can \emph{not} be deformed as above to
timelike curves, whether locally or globally. We note that all
null geodesics in Minkowski space-time are maximising.
\end{Remark}

\proof
If $s_2=1$, then Corollary~\ref{CPushup0} is
essentially a special case of Lemma~\ref{Lpushup0}: the only
difference is the statement about the neighborhood $\mcO$. This
last requirement can be satisfied  by choosing the sets
$\mcU_i$ in the proof of Lemma~\ref{Lpushup0} so that
$\mcU_i\subset \mcO$. If $s_1=0$ (and regardless of the value
of $s_2$) the result is obtained by changing time-orientation,
applying the result already established to the path
$\gamma'(s)=\gamma(1-s)$, and changing-time orientation again.
The general case is reduced to the ones already covered by
first deforming the curve $\gamma|_{[0,s_2]}$ to a new timelike
curve $\tilde \gamma$ from $p$ to $\gamma(s_2)$, and then
applying the result again to the curve $\tilde
\gamma\cup\gamma|_{[s_2,1]}$. \qed

Another result in the same spirit is provided by the following:

 \ptcx{added 14X11, check for duplication when the existence of geodesics is done, and then also in the final write-up if ever}
\begin{Proposition}
 \label{P14X11.1}
Let $\gamma$ be
causal curve from $p$ to $q$ in $\Mgtwo$
which is not a null geodesic. Then there exists a timelike curve from $p$ to $q$.
 \ptcx{this argument needs normal coordinates, problems with C one one? it explicitly addresses geodesics, so it is not allowed for metrics which are only plain, but it should still apply to Cone one metrics?}
\end{Proposition}

\proof
By Corollary~\ref{CP3.1a} we can without loss of generality assume that $\gamma$ is a piecewise broken geodesic.
If one of the geodesics forming $\gamma$ is   timelike, the result follows from Corollary~\ref{CPushup0}. It remains to consider curves which are piecewise broken null geodesics with at least one break point, say $p$. Let $q\in J^-(p)$ be close enough to $p$ so that $p$ belongs to a domain of normal coordinates $\mcO_q$ centred at $q$. Corollary~\ref{CP3.1} shows that points on $\gamma$ lying to the causal future of $p$ are not in $\dot J^+(q,\mcO_q)$, hence they are in $I^+(q,\mcO_q)$, and so $\gamma$ can be deformed within $\mcO_q$ to a timelike curve. The result follows now again from Corollary~\ref{CPushup0}.
\qed

As another straightforward corollary of  Lemma~\ref{Lpushup0}
one obtains a property of $J$, which is wrong in general for metrics which are not
$C^2$:

\begin{Corollary}\label{Cpushup}
 \MCtwok
 For any  $p\in\mcM$ we have
 $$ J^+(p)\subset\overline{I^+(p)}\;.$$
\end{Corollary}

\proof Let $q\in J^+(p)$, and let $r_i\in I^+(q)$ be any
sequence of points accumulating at $q$, then $r_i \in I^+(p)$
by Lemma~\ref{Lpushup0}, hence $q\in\overline{I^+(p)}$.\qed

\section{Extendible and inextendible
paths}\label{Sepeip}

To avoid ambiguities, recall that we only assume continuity of the metric unless explicitly indicated otherwise.

A useful concept, when studying causality, is that of a causal path
with cannot be extended any further. Recall that, from a physical
point of view, the image in space-time of a timelike path is
supposed to represent the history of some observer, and it is
sometimes useful to have at hand idealised observers which do never
cease to exist. Here it is important to have the geometrical picture
in mind, where all that matters is the image in space-time of the
path, independently of any paramaterisation: if that image ``stops",
then one can sometimes continue the path by concatenating with a
further one; continuing in this way one hopes to be able to obtain
paths which are inextendible.

In order to make things precise, let
 $\gamma:[a,b)\to
\mcM$,  be a parameterized, causal, future directed path.  A point
$p$ is called a {\em future end point} of $\gamma$ if $\lim_{s\to
b}\gamma(s)=p$. Past end points are defined in the obvious
analogous way. An {\em end point} is a point which is either a
past end point or a future end point.

Given $\gamma$ as above, together with an end point $p$, one is
tempted to extend $\gamma$ to a new path $\hat \gamma:[a,b]\to
\mcM$ defined as \bel{distb0} \hat \gamma(s)=\cases{\gamma(s), &
$s\in[a,b)\;,$ \cr p, &$s=b\;.$} \ee The first problem with this
procedure is that the resulting curve might fail to be locally
Lipschitz in general. An example is given by the timelike
future-directed path
$$[0,1)\ni \to \gamma_1(s)=(-(1-s)^{1/2},0)\in\R^{1,1}\;,$$
 which is locally Lipschitzian {on} $[0,1)$, but is \emph{not}
on $[0,1]$. (This follows from the fact that the difference
quotient $(f(s)-f(s'))/(s-s')$ blows up as $s$ and $s'$ tend to
one when $f(s)=(1-s)^{1/2}$). Recall that in our definition of a
causal curve $\gamma$, a prerequisite condition is the locally
Lipschitz character, so that the extension $\hat \gamma_1$ fails
to be causal even though $ \gamma_1$ is.

The problem is even worse if $b=\infty$: consider the timelike
future-directed curve $$[1,\infty)\ni s\to\gamma_3(s) =
(-1/s,0)\in \R^{1,1}\;.$$ Here there is no way to extend the curve
to the future, as an application from a subset of $\R$  to $\mcM$,
because the range of parameters already covers all $s\ge 1$. Now,
the image of both $\gamma_1$ and $\gamma_2$ is simply the interval
$[-1,0)\times \R$,  which can be extended to a longer causal curve
in $\R^{1,1}$ in many ways if  one thinks in terms of images
rather than of maps.

Both problems above can be taken care of by requiring that the
parameter $s$ be the proper distance parameter of some auxiliary
Riemannian metric $h$. (At this stage $h$ is not required to be
complete). This might require reparameterizing the path. From the
point of view of our definition this means that we are passing to a
different path, but the image in space-time of the new path
coincides with the previous one. If one thinks of timelike paths as
describing observers, the new observer will thus have experienced
identical events, even though he will be experiencing those events
at different times on his time-measuring device. We note, moreover,
that (locally Lipschitz) reparameterizations do not change the
timelike or causal character of paths.

\begin{cocoExa}
 \label{Exa24X11.1}
 Consider a sequence of null geodesics in $\R^{1,1}=
(\R^2,g=-dt^2+dx^2)$, with $h=d t^2+dx^2$ as the Riemannian
background metric, threading back and forth up to a space-distance
$1/n$ around the $\{x=0\}$ axis. The limit curve is
$\gamma(s)=(s/\sqrt 2,0)$ which is not $\distb$-parameterized.
\end{cocoExa}

We have already shown in Section~\ref{SCp} that a locally
Lipschitzian path can always be reparameterized by $h$-distance,
leading to a uniformly Lipschitzian path, with Lipschitz constant
one. It should be clear from the examples given above, as well as
from the examples discussed at the beginning of Section~\ref{SAc},
that it is sensible to use such a parameterization, and it is
tempting to build this requirement into the definition of a causal
path. One reason for not doing that is the existence of affine
parameterization for geodesics, which is geometrically significant,
and which is convenient for several purposes. Another reason is the
arbitrariness related to the choice of $h$. Last but not least, a
limit curve for a sequence of $\distb$-parameterized curves does not
have to be $\distb$-parameterized, as seen in Example~\ref{Exa24X11.1}.
Therefore we will not assume \emph{a priori} an $h$-distance
parameterization, but such a reparameterization will often be used
in the proofs.

Returning to \eq{distb0}, we want to show that $\hat \gamma$ will
be uniformly Lipschitz if $\distb$-parameterization is used for
$\gamma$. More generally, suppose that $\gamma$ is uniformly
Lipschitz with Lipschitz with constant $L$,
\bel{distbe1}\distb(\gamma(s),\gamma(s'))\le L|s-s'|\;.\ee
Passing with $s'$ to $b$ in that equation we obtain
$$ \distb(\gamma(s),p)\le L|s-b|\;,$$
and the Lipschitzian character of $\hat \gamma$ easily follows. We
have therefore proved:

\begin{Lemma}\label{P5.0} Let $\gamma:[a,b)\to\mcM$, $b<\infty$, be a uniformly Lipschitzian path
with an end point $p$. Then $\gamma$ can be extended to a
uniformly Lipschitzian path $\hat \gamma:[a,b]\to\mcM$, with $\hat
\gamma(b)=p$.
\end{Lemma}

Let $\gamma:[a,b)\to \mcM$, $b\in \R\cup\{\infty\}$ be a path,
then $p$ is said to be an \emph{$\omega$-limit point} of $\gamma$
if there exists a sequence $s_k \to b$ such that $\gamma(s_k)\to
p$. An end point is always an $\omega$-limit point, but the
inverse does not need to be true in general (consider
$\gamma(s)=\exp(is)\in\C$, then every point $\exp{(ix)}\in
S^1\subset \C^1$ is seen to be a $\omega$-limit point of $\gamma$ by setting
$s_k = x+2\pi k$). For $b<\infty$ and for uniformly Lipschitz
paths the notions of end point and of $\omega$-limit point
coincide:
\begin{Lemma}\label{P5.0a} Let $\gamma:[a,b)\to\mcM$, $b<\infty$, be a uniformly Lipschitzian
path. Then every $\omega$-limit point of $\gamma$ is an end point
of $\gamma$. In particular, $\gamma$ has at most one
$\omega$-limit point.
\end{Lemma}

\proof By \eq{distbe1} we have
$$\distb(\gamma(s_i),\gamma(s))\le L |s_i-s|\;,$$
and since $\distb$ is a continuous function of its arguments we
obtain, passing to the limit $i\to \infty$
$$\distb(p,\gamma(s))\le L |b-s|\;.$$
Thus $p$ is an end point of $\gamma$. Since there can be at most
one end point, the result follows. \qed

 A future directed causal curve $\gamma:[a,b)\to\mcM$ will
be said to be \emph{future extendible} if there exists $b<c\in
\R\cup\{\infty\}$ and a causal curve $\tilde \gamma:[a,c)\to\mcM$
such that \bel{443}\tilde \gamma|_{[a,b)}=\gamma\;.\ee The path
$\tilde \gamma$ is then said to be an \emph{extension} of
$\gamma$. The curve $\gamma$ will be said \emph{future
inextendible} if it is not future extendible. The notions of
\emph{past extendibility}, or of \emph{extendibility}, are defined
in the obvious way.

Extendibility in the class of causal paths  forces a causal
$\gamma:[a,b)\to \mcM$ to be uniformly Lipschitzian: This follows
from the fact that $[a,b]$ is a compact subset of the domain of
definition of any extension $\tilde \gamma$, so that
$\tilde\gamma|_{[a,b]}$ is uniformly Lipschitzian there. But then
$\tilde\gamma|_{[a,b)}$ is also uniformly Lipschitzian, and the
result follows from \eq{443}.

Whenever a uniformly Lipschitzian path can be extended by adding
an end point, it can also be extended as a strictly longer path:

\begin{Lemma}\label{P5.1} A uniformly Lipschitzian causal path
$\gamma:[a,b)\to\mcM$, $b<\infty$ is extendible if and only if it
has an end point.
\end{Lemma}

\proof Let $\hat \gamma$ be given by Proposition~\ref{P5.0}, and
let $\tilde \gamma:[0,d)$ be any maximally extended to the future,
future directed causal geodesic starting at $p$, for an
appropriate $d\in (0,\infty)$. Then $\hat \gamma \cup \tilde
\gamma$ is an extension of $\gamma$. \qed

It turns out that the paths considered in Lemma~\ref{P5.1} are
always extendible:

\begin{Theorem}\label{TP4a}
\Mcont
Let $\gamma:[a,b)\to\mcM$, $b\in \R\cup\{\infty\}$, be a future
directed causal path parameterized by $\backmg$-distance, where
$\backmg$ is any complete auxiliary Riemannian metric. Then
$\gamma$ is future inextendible if and only if $b=\infty$.
\end{Theorem}


 \proof Suppose that $b<\infty$. Let $B_h(p,r)$ denote the open
$h$-distance ball, with respect to the metric $h$, of radius $r$,
centred at $p$. Since $\gamma$ is parameterized by $h$-distance we
have, by \eq{distbe},
$$\gamma([a,b))\subset B_h(\gamma(a),b-a)
\;.$$ The Hopf-Rinow theorem~\cite{MilnorMorse,HopfRinow}  asserts that $\overline{B_h(\gamma(a),b-a)}$
is compact, therefore there exists
$p\in\overline{B_h(\gamma(a),b-a)}$ and a sequence $s_i$ such that
$$[a,b)\ni s_i\to_{i\to\infty} b \ \mbox{ \rm and } \
\gamma(s_i)\to p\;. $$ Thus $p$ is an $\omega$-limit point of
$\gamma$. Clearly $\gamma$ is uniformly Lipschitzian (with
Lipschitz modulus one), and  Lemma~\ref{P5.0a} shows  that $p$ is
an end point of $\gamma$. The result follows now from
Lemma~\ref{P5.1}. \qed

\subsection{Maximally extended geodesics}\label{sub:Inexten-geodesi}

Consider the Cauchy problem for an affinely-parameterized geodesic
$\gamma$: \bel{444}\nabla_{\dot \gamma} \dot \gamma = 0\;, \quad
\gamma(0)=p\;, \quad \dot \gamma(0)=X \;.\ee This is a
second-order ODE which, by the standard theory~\cite{Hartman},
for $C^{1,1}$ metrics, 
has unique solutions defined on a maximal interval $I=I(p,X)\ni
0$. $I$ is maximal in the sense that if $I'$ is another interval
containing $0$ on which a solution of \eq{444} is defined, then
$I'\subset I$.  When $I$ is maximal the geodesic will be called
\emph{maximally extended.} Now, it is not immediately obvious that
a \emph{maximally extended} geodesic is \emph{inextendible} in the
sense just defined: To start with, the notion of inextendibility
involves only the pointwise properties of a path, while the notion
of maximally extended geodesic involves the ODE \eq{444}, which
involves both the first and second derivatives of $\gamma$. Next,
the inextendibility criteria given above have been formulated in
terms of uniformly Lipschitzian parameterizations. While an
affinely parameterized geodesic is certainly locally Lipschitzian,
there is no \emph{a priori}  reason why it should be uniformly so,
when  maximally extended. All these issues turn out to be
irrelevant, and we have the following:

\begin{Proposition}\label{TP4ab}
\MConeone
A geodesic $\gamma:I\to\mcM$ is maximally extended as a geodesic
if and only if $\gamma$ is inextendible as a causal path.
\end{Proposition}

\proof 
Suppose, for contradiction, that $\gamma$ is a maximally extended
geodesic which is extendible as a path, thus $\gamma$ can be
extended to a path $\hat \gamma$  by adding its end point $p$ as
in \eq{distb0}. Working in a normal coordinate neighborhood
$\mcO_p$ around $p$, $\hat\gamma\cap\mcO_p$ has a last component
which is a geodesic segment which ends at  $p$. By construction of
normal coordinates the component of $\hat \gamma$ in question is
simply a half-ray through the origin, which  can be clearly be
continued through $p$ as a geodesic. This  contradicts maximality
of $\gamma$ as a geodesic. It follows that a maximally extended
geodesic is inextendible. Now, if $\gamma$ is inextendible as a
path, then $\gamma$ can clearly not be extended as a geodesic,
which establishes the reverse implication. \qed

 A result often used in
causality theory is the following:

\begin{Theorem}\label{TP4}
\Mcont
Let $\gamma$, be a future directed
causal, respectively timelike, path. Then there exists an
inextendible causal, respectively timelike, extension of
$\gamma$.
\end{Theorem}

\proof
We start with a proof assuming a $C^{1,1}$ metric, as it is simpler: If $\gamma:[a,b)\to\mcM$ is inextendible there is nothing
to prove; otherwise the path $\hat \gamma \cup \tilde \gamma$,
where $\hat \gamma$ is given by Proposition~\ref{P5.0}, and $
\tilde\gamma$ is any maximally extended future directed causal
geodesic as in the proof of Proposition~\ref{P5.1}, provides an
extension. This extension is inextendible by Proposition
\ref{TP4ab}.

When the metric is merely assumed to be continuous, one can proceed as follows:
Suppose that $\gamma$ is extendible, in particular
$\gamma$ has an end point $p$. Let $\Omega_p$ denote the
collection of all future directed, parameterized by $h$-proper
distance, timelike paths starting at $p$.  Obviously  $\Omega_p$ is
non-empty. $\Omega_p$
can be directed using the property  of ``being an extension": we
write $\gamma_1 < \gamma_2$ if $\gamma_2$ is an extension of
$\gamma_1$. The existence of inextendible paths in $\Omega_p$
easily follows from the Kuratowski-Zorn lemma.

If $\gamma_1$ is any maximal element of $\Omega_p$, then $\hat
\gamma\cup \gamma_1$, with $\hat \gamma$ given by
Lemma~\ref{P5.0}, is an inextendible future directed extension of
$\gamma$.
\qed

\section{Accumulation curves}
\label{SAc}

A key tool in the analysis of global properties of space-times is
the analysis of sequences of curves. One typically wants to obtain
a limiting curve, and study its properties. The object of this
section is to establish the existence of such limiting curves.

We wish, first, to find the ingredients needed for a useful notion
of a limit of curves. It is enlightening to start with several
examples. The first question that arises is whether to consider a
sequence of curves $\gamma_n$ defined on a common interval $I$, or
whether one should allow different domains $I_n$ for each
$\gamma_n$. To illustrate that this last option is very
unpractical, consider the family of timelike curves
\bel{gex1}(-1/n, 1/n)\ni s\to \gamma_n(s)=(s,0)\in \R^{1,1}
 \;.
 \ee
The only
sensible geometric object to which the $\gamma_n(s)$ converge is
the constant map \bel{gex2}\{0\} \ni s\to \gamma_\infty(s)=0\in
\R^{1,1}\;,\ee which is quite reasonable, except that it takes us
away from the class of causal curves. To avoid such behavior  we
will therefore assume that all the curves $\gamma_n$ have a common
domain of definition $I$.

 Next, there are various reasons why a sequence of curves might fail
to have an ``accumulation curve". First, the whole sequence might
simply run to infinity. (Consider, for example, the sequence
$$\R\ni s \to \gamma_n(s)=(s,n)\in \R^{1,1}\;.)$$ This is avoided when
one considers curves such that $\gamma_n(0)$ converges to some
point $p\in\mcM$.

Further, there might be a problem with the way the curves are
parameterized. As an example, let $\gamma_n$ be defined as
$$(-1, 1)\ni s\to \gamma_n(s)=(s/n,0)\in \R^{1,1}\;.$$
As  in \eq{gex1}, the $\gamma_n(s)$ converge to the constant map
\bel{gex2a}(-1, 1) \ni s\to \gamma_\infty(s)=0\in
\R^{1,1}\;,\ee  again not a causal curve. Another example of
pathological parameterizations is given by the family of curves
$$\R\ni s \to \gamma_n(s)=(ns,0)\in \R^{1,1}\;.$$
In this case one is tempted to say that the $\gamma_n$'s
accumulate at the path, say $\gamma_1$, if parameterization is not
taken into account. However, such a convergence  is extremely
awkward to deal with when attempting to actually prove something.
This last behavior can be avoided by assuming that all the curves
are uniformly Lipschitz continuous, with the same Lipschitz
constant. One way of ensuring this is to parameterize all
the curves by a length parameter with respect to our auxiliary
complete Riemannian metric $h$.

Yet another problem arises when considering the family of
Euclidean-distance-parameterized causal curves
$$\R\ni s \to \gamma_n(s)=(s+n,0)\in \R^{1,1}\;.$$
This can be gotten rid of by shifting the distance parameter so
that the sequence $\gamma_n(s_0)$ stays in a compact set, or
converges, for some $s_0$ in the domain $I$.

The above discussion motivates the hypotheses of the following
result:
%
\begin{Proposition}\label{Paccum} Let $(\mcM,g)$ be a $C^3$ Lorentzian manifold
with a $C^2$ metric.
Let $\gamma_n:I\to\mcM$ be a sequence of uniformly Lipschitz
future directed causal curves, and suppose that there exist
$p\in \mcM$ such that \bel{450}\gamma_n(0) \to p\;.\ee Then
there exists a  future directed causal curve $\gamma:I\to\mcM$
and a subsequence $\gamma_{n_i}$ converging to $\gamma$ in the
topology of uniform convergence on compact subsets of $I$.
\end{Proposition}

\begin{coco}
\begin{Remark}
 \label{R25X11.1xx}
The hypothesis that $g$ is $C^2$ is made to guarantee that the function $(p,q)\mapsto \sigma_q(p)$ is continuous, and depends continuously upon the metric. It is shown in~\cite{ChGrant} that the result remains true for metrics which are merely continuous.
The analysis there relies on the result here and suitable smooth approximations of the metric.
 \ptcx{needs rewording when merged}
\end{Remark}
\end{coco}

Proposition~\ref{Paccum} provides the justification for the
following definition:

\begin{Definition}\label{DSAc}
Let $\gamma_n:I\to\mcM$ be a sequence of paths in $\Mgz$. We shall say that
$\gamma:I\to\mcM$ is an \underline{accumulation curve} of the
$\gamma_n$'s, or that the $\gamma_n$'s accumulate at $\gamma$, if
there exists a subsequence $\gamma_{n_i}$ that converges to $\gamma$
uniformly on compact subsets of $I$.
\ptcx{definition changed, should
be checked for consistency}
\end{Definition}

\begin{coco}
In their treatment of causal theory, Hawking and Ellis~\cite{HE}
introduce a notion of \emph{limit curve} for paths, regardless of
parameterization, which we find very awkward to work with. A
related but slightly more convenient notion of \emph{cluster
curve} is considered in~\cite{Kriele}, where the name of ``limit
curve" is used for yet another notion of convergence. As discussed
in~\cite{Beem-Ehrlich:Lorentz2,Kriele}, those definitions lead to
pathological behavior in some situations. We have found the above
notion of ``accumulation curve" the most convenient to work with
from several points of view.

A sensible terminology, in the context of Definition~\ref{DSAc},
could be  ``$C^0_\loc$-limits of curves", but we prefer not to use
the term ``limit" in this context, as limits are usually unique,
while Definition~\ref{DSAc} allows sequences that have more than one accumulation
curve.
\end{coco}

\medskip

\noindent{\sc Proof of Proposition~\ref{Paccum}:} The hypothesis
that all the $\gamma_n$'s are uniformly Lipschitz reads
\bel{452} \distb(\gamma_n(s),\gamma_n(s'))\le L |s-s'|\;,\ee
for some constant $L$.
 This shows that the family
$\{\gamma_n\}$ is equicontinuous, and \eq{450} together with the
Arzela-Ascoli theorem implies that for every compact set $K\subset
I$  there exists a curve $\gamma_K:K\to\mcM$ and a subsequence
$\gamma_{n_i}$ which converges uniformly to $\gamma_K$ on $K$. One
can obtain a $K$-independent curve $\gamma$ by the so-called
\emph{diagonalisation} procedure.

\begin{coco} The diagonalisation procedure goes as follows: For ease of notation
we consider $I=\R$, the same argument applies on any interval with
obvious modifications. Let $\gamma_{n(i,1)}$ be the sequence which
converges to $\gamma_{[-1,1]}$; applying Arzela-Ascoli to this
sequence one can extract a subsequence $\gamma_{n(i,2)}$ of
$\gamma_{n(i,1)}$ which converges uniformly to some curve
$\gamma_{[-2,2]}$ on $[-2,2]$. Since $\gamma_{n(i,2)}$ is a
subsequence of $\gamma_{n(i,1)}$, and since $\gamma_{n(i,1)}$
converges to $\gamma_{[-1,1]}$ on $[-1,1]$, one finds that
$\gamma_{[-2,2]}$ restricted to $[-1,1]$ equals $\gamma_{[-1,1]}$.
One continues iteratively: suppose that
$\{\gamma_{n(i,k)}\}_{i\in\N}$ has been defined for some $k$, and
converges to a curve  $\gamma_{[-k,k]}$ on $[-k,k]$, then the
sequence $\{\gamma_{n(i,k+1)}\}_{i\in\N}$ is defined as a
subsequence of $\{\gamma_{n(i,k)}\}_{i\in\N}$ which converges to
some curve $\gamma_{[-(k+1),k+1]}$ on $[-(k+1),k+1]$. The curve
$\gamma$ is finally defined as
$$\gamma(s)=\gamma_{[-k,k]}(s)\;,$$ where $k$ is any number such
that $s\le k$. The construction guarantees that
$\gamma_{[-k,k]}(s)$ does not depend upon $k$ as long as $s\le k$.
\end{coco}

It remains to show that $\gamma$ is causal. Passing to the limit
$n\to \infty$ in \eq{452} one finds \bel{453}
\distb(\gamma(s),\gamma(s'))\le L|s-s'|\;.\ee For $q\in \mcM$ let
$\mcO_{q}$ be an elementary neighborhood of $q$
 as in Proposition~\ref{P3}, and let $\sigma_{q}$ be the
 associated function defined by \eq{lorp}.
 Let $s\in\R$ and consider any point $\gamma(s)\in\mcM$. Now,
  the size of the sets $\mcO_q$ can be controlled uniformly when
  $q$ varies over compact subsets of $\mcM$. It follows that for
  all $s'$ close enough to $s$ and for all $n$ large enough
  we have $\gamma_n(s')\in \mcO_{\gamma_n(s)}$. Since the
  $\gamma_n$'s are causal, Proposition~\ref{P3} shows that we have
\bel{454}\sigma_{\gamma_n(s)}(\gamma_n(s'))\le 0\;.
\ee
Passing to the limit in \eq{454} gives
  \bel{455}\sigma_{\gamma(s)}(\gamma(s'))\le 0\;.\ee
This is only possible if $\gamma$ is causal, which can be seen
  as follows: Suppose that
  $\gamma$ is differentiable at $s$. In normal coordinates on $\mcO_{\gamma(s)}$ we
  have, by definition of a derivative,
  $$\gamma^\mu(s')= \underbrace{\gamma^\mu(s)}_{=0}+\dot \gamma^\mu(s)(s'-s)+o(s'-s)\;,$$
  hence
  $$0\ge\sigma_{\gamma(s)}\left(\gamma(s')\right) \equiv \eta_{\mu\nu}\gamma^\mu(s')\gamma^\nu(s')= \eta_{\mu\nu}\dot \gamma^\mu(s)\dot
  \gamma^\nu(s)(s'-s)^2+o((s'-s)^2)\;.$$
  For $s'-s$ small enough this is only possible if
  $$\eta_{\mu\nu}\dot \gamma^\mu(s)\dot
  \gamma^\nu(s)\le 0\;,$$
  and $\dot \gamma $ is causal, as we desired to show.
\qed

Let us address now the question of inextendibility of accumulation
curves. We note the following lemma:

\begin{Lemma}
\label{Lem:inex} Let $\gamma_n$ be a sequence of
$\distb$-parameterized inextendible causal curves converging to
$\gamma$ uniformly on compact subsets of \,$\R$, then $\gamma$ is
inextendible.
\end{Lemma}

\proof
Note that the parameter range of $\gamma$ is $\R$, and the result
would follow from Theorem~\ref{TP4a} if $\gamma$ were
$\distb$-parameterized, but this might fail to be the case, as seen
in Example~\ref{Exa24X11.1}.

So we need to show that both $\gamma|_{[0,\infty)}$ and
$\gamma|_{(-\infty,0]}$ are of infinite length. As
usual it suffices to consider $\gamma|_{[0,\infty)}$, we retain the
name $\gamma$ for this last path. Suppose that this is not the case,
then there exists $a<\infty$ so that $\gamma$ is defined on $[0,a)$,
when reparameterised by $\distb$--distance. By Theorem~\ref{TP4a}
the curve $\gamma$ can be extended to a causal curve defined on
$[0,a]$, still denoted by $\gamma$.

Let $\mcU$ be an elementary neighborhood centred at $\gamma(a)$, and
let $0<b<a$ be such that $\gamma(b)\in \mcU$. By definition of
accumulation curve there exists a sequence $n_i\in \N$, a compact
interval $[-k,k]\subset \R$ and a sequence $s_i\in [-k,k]$ such that
$\gamma_{n_i}(s_i)$ converges to $\gamma(b)$, in particular we will
have $\gamma_{n_i}(s_i)\in \mcU$ for $i$ large enough. We note the
following:

 \begin{Lemma}\label{Lgh1.0} Let $\mcU$ be an elementary
 neighborhood, as defined in Definition~\ref{Delem}. There exists a constant $\ell$ such that for any
 causal curve $\gamma:I\to \mcU$
 the $h$--length $|\gamma|_h$ of $\gamma$ is bounded by $\ell$.
 \ptcx{consider adding both a lower and upper bound}
 \end{Lemma}

 To prove Lemma~\ref{Lgh1.0} we need the following variation
 of the inverse
 Cauchy-Schwarz inequality\restrict{ (compare Proposition~\ref{PinvCS})}:

 \begin{Lemma}\label{Lgh1}
 Let $K$ be a compact set and let $X$ be a continuous
 timelike vector field defined there, then there exists a \underline{strictly
 positive}
 constant $C$ such that
 for all $q\in K$ and for all causal vectors $Y\in T_q\mcM$ we have
 \bel{egh1} |g(X,Y)| \ge C |Y|_h\;.\ee
 \end{Lemma}
\proof By homogeneity it is sufficient to establish \eq{egh1} for
causal $Y\in T_q\mcM$ such that $|Y|_h=1$; let us denote by $U(h)_q$
this last set. The result follows then by continuity of the strictly
positive function
$$\cup_{q\in K} U(h)_q \ni Y \to |g(X,Y)|$$
on the compact set $\cup_{q\in K} U(h)_q$. \qed

Returning to the proof of Lemma~\ref{Lgh1.0}, let $x^0$ be the local
time coordinate on $\mcU$, since $X:=\nabla x^0$ is timelike we can use
Lemma~\ref{Lgh1} with $K=\overline{\mcU}$ to conclude that there
exists a constant $C$ such that for any causal curve $\gamma\subset
\mcU$  we have
$$|g(X,\dot \gamma)|\ge C>0$$
at all points at which $\gamma$ is differentiable. This implies, for
$s_2\ge s_1$, \beaa |x^0(s_2)-x^0(s_1)| & \ge & \int_{s_1}^{s_2}
|g(\nabla x^0, \dot \gamma)| ds
\\ &\ge & C \int_{s_1}^{s_2}  ds = C |s_2-s_1|\;. \eeaa
It follows that \bel{egh2.0} |\gamma|_h \le \ell:= \frac 2 {C}
\sup_{\mcU}|x^0|<\infty\;,\ee as desired. \qed

Returning to the proof of Lemma~\ref{Lem:inex}, it follows from
Lemma~\ref{Lgh1.0} applied to $\gamma_{n_i}$ that
$\gamma_{n_i}|_{[s_i,s_i+\ell]}$ must exit $\mcU$. This implies that
$\gamma_{n_i}|_{[-k,k+\ell]}$ cannot accumulate at a curve which has
an end point  $\gamma(b)\in \mcU$, and the result follows.
 \qed

In summary, it follows from  Lemmata~\ref{P5.1} and \ref{Lem:inex}
together with Proposition~\ref{Paccum} that:

\begin{Theorem}
\label{The:accum}
Let $(\mcM,g)$ be a $C^3$ Lorentzian manifold with a $C^2$ metric.
Every sequence of future directed, inextendible,
causal curves which accumulates at a point $p\in\mcM$ accumulates at  some
future directed, inextendible, causal curve through $p$. \qed
\end{Theorem}

One  is sometimes interested in sequences of maximally
extended geodesics:

\begin{Proposition}
\label{Pro:inexgeo} Let $\gamma_n$ be a sequence of maximally
extended geodesics accumulating at $\gamma$ in $(\mcM,g)_{C^{1,1}}$. Then $\gamma$ is a
maximally extended geodesic.
%
%
\end{Proposition}

\proof If we use a $\distb$-parameterization of the $\gamma_n$'s and of $\gamma$
such that $\gamma_n(0)\to \gamma(0)$, then by the
Arzela-Ascoli Theorem (passing to a subsequence if necessary) the
$\gamma_n$'s converge to $\gamma$, uniformly on compact subsets of
$\R$.  Let $\mcK$ be a compact neighborhood of $\gamma(0)$,
compactness of $\cup_{p\in\mcK}U_p\mcM$, where $U_p\mcM\subset
T_p\mcM$ is the set of $h$-unit vectors tangent to $\mcM$, implies
that there exists a subsequence such that $\dot \gamma_n(0)$
converges to some vector $X\in U_{\gamma(0)}\mcM\subset
T_{\gamma(0)}\mcM$. Let $\sigma:(a,b)\to\mcM$, $a\in
\R\cup\{-\infty\}$, $b\in \R\cup\{\infty\}$, be an affinely
parameterised maximally extended geodesic through $\gamma(0)$ with
initial tangent vector $X$. By continuous dependence of ODE's upon
initial values  it follows that 1) for any $a<\alpha<\beta<b$ all
the $\gamma_n$'s,  except perhaps for a finite number, are defined
on $[\alpha,\beta]$ when affinely parameterized, and 2) they
converge to $\sigma|_{[\alpha,\beta]}$ in the (uniform)
$C^1([\alpha,\beta],\mcM)$ topology. Thus $\dot\gamma_n(s)\to
\dot\sigma(s)$ uniformly on compact subsets of $(a,b)$, which
implies that a $\distb$-parameterization is preserved under taking
limits. Hence the $\gamma_n$'s, when $\distb$-parameterized,
converge uniformly to a $\distb$-reparameterization of $\sigma$ on
compact subsets of $\R$, call it $\mu$. It follows that
$\gamma=\mu$, and $\gamma$ is a maximally extended geodesic. \qed

\subsection{Achronal causal curves}
 \label{subsec:Achrona-causal-curves}

A curve  $\gamma:I\to \mcM$ is called \emph{achronal} if
$$\forall \ s,s'\in I \qquad \gamma(s)\not \in I^+(\gamma(s'))\;.$$ Any
spacelike geodesic in Minkowski space-time is achronal. More
interestingly, it follows from Proposition~\ref{PM1} that this is
also true for null geodesics. However, null geodesics do not have
to be achronal in general: consider, \emph{e.g.}, the
two-dimensional space-time $\R\times S^1$ with the flat metric
$-dt^2+dx^2$, where $x$ is an angle-type coordinate along $S^1$
with periodicity, say, $2\pi$. Then the points $(0,0)$ and
$(2\pi,0)$ both lie on the null geodesic
$$s\to (s,s \mod 2\pi)\;,$$ and are clearly timelike related to each
other.

In this section we will be interested in causal curves that are
achronal. We start with the following:

\begin{Proposition}
\label{Pro:acausalgeo}
\MCtwok
If $\gamma$ is an \emph{achronal causal}
curve, then $\gamma$ is a null geodesic.
\end{Proposition}

\proof Let $\mcO$ be any elementary neighborhood, then any
connected component of $\gamma \cap \mcO$ is a null geodesic by
Corollary~\ref{CP3.1}. \qed

\begin{Theorem}
\label{The:acaulseq}
\MCthreek
 Let $\gamma_n:I\to\mcM$ be a sequence of
 achronal causal
curves accumulating at $\gamma$, then $\gamma$ is achronal.
\end{Theorem}

\begin{Remarks}
1. Propositions~\ref{Pro:acausalgeo} and \ref{Pro:inexgeo} show that
$\gamma$ is inextendible if the $\gamma_n$'s are.

2. The theorem uses the fact that timelike futures and pasts are open, which is
not clear if the metric is not $C^2$.
\end{Remarks}

\proof 
Suppose $\gamma$ is not achronal, then there exist $s_1,s_2\in
I$ such that $\gamma(s_2)\in I^+(\gamma(s_1))$, thus there exists a
\emph{timelike} curve $\hat \gamma:[s_1,s_2]\to \mcM$ from
$\gamma(s_1)$ to $\gamma(s_2)$. Choose some $\hat s \in (s_1,s_2)$.
We have $\gamma(s_2)\in I^+(\hat \gamma(\hat s))$, and since
$I^+(\hat \gamma(\hat s))$ is open there exists an open neighborhood
$\mcO_2$ of $\gamma(s_2)$ such that $\mcO_2 \subset I^+(\hat
\gamma(\hat s))$. Similarly there exists an open neighborhood
$\mcO_1$ of $\gamma(s_1)$ such that $\mcO_1 \subset I^-(\hat
\gamma(\hat s))$. This shows that any point $p_2\in \mcO_2$ lies in
the timelike future of any point $p_1\in\mcO_1$: indeed, one can go
from $p_1$ along some timelike path to $\hat \gamma(\hat s)$, and
continue along another timelike path from $\hat \gamma(\hat s)$ to
$p_2$.

Passing to a subsequence if necessary, there exist sequences
$s_{1,n} $ and $s_{2,n}$ such that $\gamma_n(s_{1,n})$ converges to
$\gamma(s_1)$ and $\gamma_n(s_{2,n})$ converges to $\gamma(s_2)$.
Then $\gamma_n(s_{1,n})\in \mcO_1$ and $\gamma_n(s_{2,n})\in \mcO_2$
for $n$ large enough, leading to $\gamma_n(s_{2,n})\in
I^+(\gamma_n(s_{1,n}))$, contradicting achronality of $\gamma_n$.
\qed

\ptcx{prove that $\gamma$ is achronal if the $\gamma_n$'s were
maximising timelike; show that arc length is upper semicontinuous
and that lorentzian distance is lower semicontinuous (but this is probably already at least in part in
the section on lorentzian distance }

\ptcx{there is a file greg.tex, about generators of boundaries of lut futures, should be referred to
somewhere in this part and done in the continuous boundaries framework}

\section{Causality conditions}
\label{SCc}

Space-times can exhibit various causal pathologies, most of which
are undesirable from a physical point of view. The simplest
example of unwanted causal behaviour is the existence of closed
timelike curves. A space-time is said to be \emph{chronological}
if no such curves exist. An example of a space-time which is not
chronological is provided by $S^1\times \R$ with the flat metric
$-dt^2+dx^2$, where $t$ is a local coordinate defined modulo
$2\pi$ on $S^1$. Then every circle $x=\const$ is a closed timelike
curve.

\begin{coco}
The class of compact manifolds is a very convenient one from the
point of view of Riemannian geometry. The following result of
Geroch shows that such manifolds are always pathological from a
Lorentzian perspective:

\begin{Proposition}[Geroch~\cite{Geroch:topology}]\label{LGer} Every compact
space-time $\Mgtwo$ contains a closed timelike curve.
\end{Proposition}

\proof Consider the covering of $\mcM$ by the collection of open
sets $\{I^-{(p)}\}_{p\in\mcM}$, by compactness a finite covering
$\{I^-{(p_i)}\}_{i=1,\ldots,I}$ can be chosen. The possibility
$p_1\in I^-{(p_1)}$ yields immediately a closed timelike curve
through $p_1$, otherwise there exists $p_{i(1)}$ such that $p_1\in
I^-(p_{i(1)})$. Again if $p_{i(1)}\in I^-(p_{i(1)})$ we are done,
otherwise there exists $p_{i(2)}$ such that $p_{i(1)}\in
I^-(p_{i(2)})$. Continuing in this way we obtain a --- finite or
infinite --- sequence of points $p_{i(j)}$ such that
\bel{eqGP}p_{i(j)}\in I^-(p_{i(j+1)}) \;.\ee If the sequence is
finite we are done. Now, we have only a finite number of $p_i$'s
at our disposal, therefore if the sequence is finite it has to
contain repetitions:
$$p_{i(j+\ell)}= p_{i(j)}$$
for some $j$, and some $\ell>0$. It should be clear from \eq{eqGP}
that there exists a closed timelike curve through $p_{i(j)}$. \qed

\begin{Remark}\label{RLGer} Galloway~\cite{Galloway:globasp} has
shown that in compact space-times $(\mcM,g)$ there exist closed
timelike curves through any two points $p$ and $q$, under the
supplementary condition that the Ricci tensor $\Ric$ satisfies the
following energy condition: \bel{Gec} \Ric(X,X)>0 \ \mbox{ for all
causal vectors $X$.}\ee
\end{Remark}
\end{coco}

The chronology condition excludes closed \underline{timelike}
curves, but it just fails to exclude the possibility of occurrence
of \underline{closed} causal curves. A space-time is said to be
\emph{causal} if no such curves can be found. The existence of
space-times which are \emph{chronological} but not \emph{causal}
requires a little work:

\begin{Example}\label{Escc1}
Let $\mcM=\R\times S^1$ with the metric
$$g=2dt \, dx + f(t) dx^2\;,$$
where $f$ is any function satisfying
$$ f\ge 0\;,\quad \mathrm{with } \ f(t)=0 \ \mathrm{ iff }\  t=0\;.$$
(The function $f(t)=t^2$ will do.) Here $t$ runs over the $\R$
factor of $\mcM$, while $x$ is a coordinate defined modulo $2\pi$
on $S^1$. In matrix notation we have
$$[g_{\mu\nu}]= \left[\begin{array}{cc} 0 & 1 \\ 1 & f
\end{array}\right]\;,$$
which leads to the following inverse metric
$$[g^{\mu\nu}]= \left[\begin{array}{cc} -f & 1 \\ 1 & 0
\end{array}\right]\;.$$
It follows that \bel{escc1} g(\nabla t,\nabla t) = -f \le
0\;,\mathrm{with } \ g(\nabla t,\nabla t) =0 \ \mathrm{ iff }\
t=0\;.\ee Recall, now, that a function $\tau$ is called a
\emph{time function} if $\nabla \tau$ is timelike, past-pointing.
\Eq{escc1} shows that $t$ is a time function on the set $\{t\ne
0\}$. Since a time-function is strictly increasing on any causal
curve (see Lemma~\ref{LP3.0}), one easily concludes that no closed
causal curve in $\mcM$ can intersect the set $\{t\ne 0\}$. In
other words, closed causal curves --- if they do exist --- must be
entirely contained in the set $\{t=0\}$. Now, any  curve $\gamma$
contained in this last set is of the form
$$\gamma(s)=(0,x(s))\;,$$
with tangent vector
$$\dot \gamma = \dot x \partial_x \quad \Longrightarrow \quad
g(\dot \gamma, \dot \gamma) = (\dot x)^2 g(\partial_x,\partial_x)
=(\dot x)^2 g_{xx}=0\;.$$ This shows in particular that
\begin{itemize}
\item $\mcM$ does contain closed causal curves: an example is
given by $x(s)=s \mod 2 \pi$.
\item All closed causal curves are null.
\end{itemize}
It follows that $(\mcM,g)$ is indeed chronological, but not
causal, as claimed.
\end{Example}

It is desirable to have a condition of causality which is stable
under small changes of the metric. By way of example, consider a
space-time which contains a family of causal curves $\gamma_n$
with both $\gamma_n(0)$ and $\gamma_n(1)$ converging to $p$. Such
curves can be thought of as being ``almost closed". Further, it is
clear that one can produce an arbitrarily small deformation of the
metric which will allow one to obtain a closed causal curve in the
deformed space-time. The object of our next causality condition is
to exclude this behaviour. A space-time will be said to be
\emph{strongly causal} if every neighborhood $\mcO$ of a point
$p\in\mcM$ contains a neighborhood $\mcU$ such that for every
causal curve $\gamma:I\to\mcM$ the set $$\{s\in I:\gamma(s)\in
\mcU\}\subset I$$ is a connected subset of $I$. In other words,
$\gamma$ does not re-enter $\mcU$ once it has left it.

Clearly, a strongly causal space-time is necessarily causal.
However, the inverse does not always hold. An example is given in
Figure~\ref{Fscc}.
\begin{figure}[tbh]
  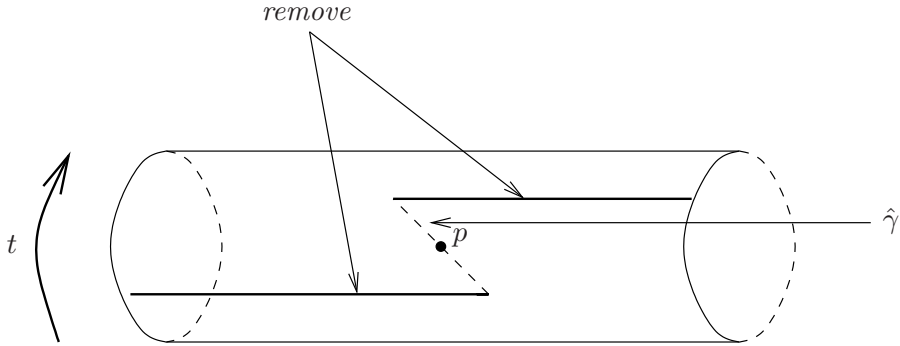
 \caption{\label{Fscc}A causal
    space-time which is not strongly causal. Here the metric is the
    flat one $-dt^2+dx^2$, with $t$ a parameter along $S^1$, so that
    the light cones are at $45^o$, in particular $\hat\gamma$ is a
    null geodesic. It should be clear that no matter how small the
    neighborhood $\mcU$ of $p$ is, there will exist a causal curve as
    drawn in the figure which will intersect this neighborhood twice.
    In order to show that $(\mcM,g)$ is causal one can proceed as
    follows: suppose that $\gamma$ is a closed causal curve in
    $\mcM$, then $\gamma$ has to intersect the hypersurfaces $\{t=\pm
    1\}$ at some points $x_\pm$, with $x_-> 1$ and $x_+<-1$.  If we
    parameterize $\gamma$ so that $\gamma(s)=(s,x(s))$ we obtain $-2>
    x_+-x_-=\int_{-1}^1 \frac {dx}{ds} ds\;,$ hence there must exist
    $s_* \in [-1,1]$ such that $dx/ds<-1$, contradicting causality of
    $\gamma$.}
\end{figure}

The definition of strong causality appears, at first sight,
somewhat unwieldy to verify, so simpler conditions are desirable.
The following provides a useful criterion: A space-time $(\mcM,g)$
is said to be \emph{stably causal} if there exists a time function
$t$ globally defined on $\mcM$. Recall
--- see Lemma~\ref{LP3.0} --- that time functions are strictly increasing on causal curves. It
then easily follows  that \emph{stable causality} implies
\emph{strong causality}:

\begin{Proposition}\label{Pscisc}
If $\Mgz$ is stably causal, then it is strongly causal.
\end{Proposition}

\proof Let $\mcO$ be a connected open neighborhood of $p\in \mcM$,
and let $\varphi$ be a nonegative smooth function such that
$\varphi(p)\ne 0 $ and such that the support
$\mathrm{supp}\varphi$ of $\varphi$ is a compact set contained in
$\mcO$. Let $\tau$ be a time function on $\mcM$, for $a\in \R$ set
$$\tau_a:=\tau + a \varphi\;.$$
As $\nabla \tau$ is timelike, the function $g(\nabla \tau,\nabla
\tau)$ is bounded away from zero on the compact set
$\mathrm{supp}\varphi$, which implies that there exists
$\epsilon>0$ small enough so that $\tau_{\pm\epsilon}$ are time
functions on $\mathrm{supp}\varphi$. Now, $\tau_{\pm\epsilon}$
coincides with $\tau$ away from $\mathrm{supp}\varphi$, so that
the $\tau_{\pm \epsilon}$'s are actually time functions on $\mcM$
as well. We set
$$\mcU:= \{q:\tau_{-\epsilon}(q)<\tau(p)<\tau_{+\epsilon}(q)\}
\;.$$ We have: \begin{itemize} \item $p\in \mcU$, therefore $\mcU$
is not empty; \item $\mcU$ is open because the $\tau_a$'s are
continuous; \item $\mcU\subset \mcO$ because $\varphi$ vanishes
outside of $\mcO$. \end{itemize}  Consider any causal curve
$\gamma$ the image of which intersects $\mcU$, $\gamma$ can enter
or leave $\mcU$ only through
\bel{eqsscc}\partial\mcU\subset\{q:\tau_{-\epsilon}(q)=\tau(p)\}
\cup \{q:\tau(p)=\tau_{+\epsilon}(q)\}\;.\ee  At a point $s_-$ at
which $\gamma(s_-)\in \{q:\tau_{-\epsilon}(q)=\tau(p)\}$ we have
$$\tau(p)=\tau_{-\epsilon}(\gamma(s_-))=\tau(\gamma(s_-))-\epsilon\varphi(\gamma(s_-))\quad \Longrightarrow
\quad \tau(\gamma(s_-)) > \tau(p)\;.$$ Similarly at a point $s_+$
at which $\gamma(s_+)\in \{q:\tau_{+\epsilon}(q)=\tau(p)\}$ we
have
$$\tau(p)=\tau_{+\epsilon}(\gamma(s_+))=\tau(\gamma(s_+))+\epsilon\varphi(\gamma(s_+))\quad \Longrightarrow
\quad \tau(\gamma(s_+)) < \tau(p)\;.$$
  As $\tau$ is
increasing along $\gamma$, we conclude that $\gamma$ can enter
$\mcU$ only through $\{q:\tau_{+\epsilon}(q)=\tau(p)\}$,  and
leave  $\mcU$ only through $ \{q:\tau_{-\epsilon}(q)=\tau(p)\}$.
Lemma~\ref{LP3.0} shows that $\gamma$ can intersect each of the
two sets at the right-hand-side of \eq{eqsscc} at most once. Those
facts obviously imply connectedness of the intersection of (the
image of) $\gamma$ with $\mcU$.\qed

\begin{coco}
There exist various alternative definitions of
stable causality which are equivalent for $C^2$ metrics, but note that equivalence is not obvious and
its proof requires work. For example,
Yvonne~Choquet-Bruhat~\cite{YCB:GRbook} defines stable causality as the
requirement of existence of a
timelike vector field $v$ such that $g-v\otimes v$ is
chronological. Hawking and Ellis~\cite{HE}
define stable causality by requiring that
$C^0$-small perturbations of the metric
preserve causality. Compare~\cite{MinguzziK,MinguzziLL}.
  \ptcx{prove equivalence; mention $K$-causality}
\end{coco}

The strongest causality condition is that of \emph{global
hyperbolicity}, considered in the next section.

\section{Global hyperbolicity}
\label{SGh}

A space-time $(\mcM,g)$ said to be \emph{globally hyperbolic} if
it is \emph{strongly causal}, and if for every $p,q\in \mcM$ the sets
$J^+(p)\cap J^-(q)$ are compact.

\begin{coco}
{It is often convenient to use the equivalent requirement of stable causality together with compactness of the  sets
$J^+(p)\cap J^-(q)$; compare Theorem~\ref{TCS}, page~\pageref{TCS}.
 \ptcx{there is a paper which says that causal is enough, by Bernal and Sanchez gr-qc/0611138, but I should not quote it without believing it}
The current
definition is the one that appears to be the most widely used.
From a Cauchy-problem point of view, a natural definition is by requiring the
existence of a Cauchy surface, compare Section~\ref{SDd} and
Theorem~\ref{TCS}, page~\pageref{TCS}. The last definition is again equivalent
for $C^2$ metrics, but the equivalence for $C^0$ metrics is not clear.
 \ptcx{is this fixable }
}
\end{coco}


It is not too difficult to show that Minkowski space-time
$\R^{1,n}$ is globally hyperbolic: first, the Minkowski time $x^0$
provides a time-function on $\R^{n,1}$; this implies strong
causality. Compactness of $J^+(p)\cap J^-(q)$ for all $p$'s and
$q$'s is easily checked by drawing pictures; it is also easy to
write a formal proof using Proposition~\ref{PM1}, this is left as
an exercise to the reader.

The notion of globally hyperbolicity provides excellent control
over causal properties of $(\mcM,g)$. This will be made clear at
several other places in this work. Anticipating, let us list a few
of those:
\begin{enumerate}
\item Let $(\mcM,g)$ be globally hyperbolic. If $J^+(p)\cap
J^-(q)\ne \emptyset$, then there exists a causal geodesic from $p$
to $q$. Similarly if $I^+(p)\cap I^-(q)\ne \emptyset$, then there
exists a timelike geodesic from $p$ to $q$. \item The Cauchy
problem for linear wave equations is globally solvable on globally
hyperbolic space-times. \item A key theorem of Choquet-Bruhat and
Geroch asserts that \emph{maximal globally hyperbolic} solutions
of the Cauchy problem for Einstein's equations are unique up to
diffeomorphism.
\end{enumerate}

We start our study of globally hyperbolic space-times with the following property:

\begin{Proposition}\label{Pgh1}
 \ptcx{to show that causal and compactness of diamonds implies globally hyperbolic
 one would need to repeat this without assuming stable causality}
 Let $\Mgthreek$ be globally
 hyperbolic, and let $\gamma_n$ be a family of causal curves
 accumulating both at $p$ and $q$. Then  there exists a
  causal curve $\gamma$, accumulation curve of the (perhaps reparameterized) $\gamma_n$'s
  which passes both through $p$ and $q$.
   \ptcx{ok for continuous metrics (but not the current proof)}
\end{Proposition}

\begin{Remark}\label{Rgh1}
The result is wrong if stable causality is assumed only. Indeed,
let $(\mcM,g)$ be the two-dimensional Minkowski space-time with the
origin removed. Let $\gamma_n$ be obtained by following a timelike
geodesic from $p=(-1,0)$ to $(0,1/n)$ and then another timelike
geodesic to $q=(1,0)$. Then $\gamma_n$ has exactly two accumulation
curves $s\to (s,0)$, with $s\in [-1,0)$ for the first one and $s\in
(0,1]$ for the second, none of which passes through both $p$ and
$q$.
\end{Remark}

 \proof Extending the $\gamma_n$'s to inextendible curves, and
 reparameterizing if necessary, we can assume that the $\gamma_n$'s are $\distb$-parameterized,
 with common domain of definition $I=\R$, and
 with $\gamma_n(0)$ converging to $p$. If $p=q$ the result has already been established
 in Proposition~\ref{Paccum}, so we assume that
  $p\ne q$. Consider the compact set
 \bel{dd0}\mcK:=\left(J^+(p)\cap J^-(q)\right)\cup \left(J^+(q)\cap
 J^-(p)\right)\ee (since a globally hyperbolic space-time is causal, one of those sets is,
 of course, necessarily
 empty). $
 \mcK$ can be covered by a finite number of elementary
 domains $\mcU_i$, $i=1,\cdots,N$. Strong causality allows us to choose the $\mcU_i$'s small enough so
 that for every $n$ the image of $\gamma_n$ is a connected subset in $\mcU_i$.
 We can choose a
 parameterization of the $\gamma_n$'s by $h$--length  so that,
 passing to a subsequence of the $\gamma_n$'s if necessary, we
 have
$\gamma_n(0)\to p$. Extending the $\gamma_n$'s if necessary we can
assume that all the $\gamma_n$'s are defined on $\R$. Now,
Lemma~\ref{Lgh1.0}, p.~\pageref{Lgh1.0}, shows that there exists a constant $L_i$ ---
independent of $n$
--- such that
 the $h$--length $|\gamma_n\cap \mcU_i|_h$ is bounded by $L_i$.
     Consequently the $h$--length $|\gamma_n\cap \mcK|_h$, with $\mcK$
as in \eq{dd0}, is bounded by \bel{egh2} |\gamma_n\cap \mcK|_h\le
L:= L_1+L_2+\ldots+L_I\;.\ee By hypothesis the $\gamma_n$'s
accumulate at $q$, therefore there exists a sequence $s_n$
(passing again to a subsequence if necessary) such that
$$\gamma_n(s_n)\to q\;.$$ \Eq{egh2} shows that the sequence $s_n$
is bounded, hence --- perhaps passing to a subsequence --- we have
$s_n\to s_*$ for some $s_*\in \R$.

At this stage we could use Proposition~\ref{Paccum}, but one might
as well argue directly: by our choice of parametrization we have
\bel{egh6} \distb(\gamma_n(s),\gamma_n(s'))\le |s-s'|\ee (see
\eq{distbe}-\eq{distba}). This shows that the family
$\{\gamma_n\}$ is equicontinuous, and \eq{egh6} together with the
Arzela-Ascoli theorem (on the compact set $[-L,L]$) implies
existence of a curve $\gamma:[-L,L]\to \mcM$ and a subsequence
$\gamma_{n_i}$ which converges uniformly to $\gamma$ on $[-L,L]$.
As $\gamma_{n_i}({s_{n_i}})$ converges both to $\gamma({s_*})$ and
to $q$ we have $$\gamma({s_*})=q\;.$$ This shows that $\gamma$ is
the desired causal curve joining  $p$ with $q$. \qed

\begin{Remark}
It should be clear from the proof above that, as emphasised in~\cite{YCB:GRbook},  a space-time is globally hyperbolic
 if and only if the length of  causal paths between two points, as measured with respect to a smooth complete Riemannian metric, is
 bounded by a number independent of the path.
\end{Remark}

As a straightforward corollary of Proposition~\ref{Pgh1} we
obtain:

\begin{Corollary}\label{Cgh1}
Let $\Mgthreek$ be globally hyperbolic, then
$$\overline{I^\pm(p)}=J^\pm(p)\;.$$
\end{Corollary}

\proof Let $q_n\in I^+(p)$ be a sequence of points accumulating at
$q$, thus there exists a sequence $\gamma_n$ of causal curves from
$p$ to $q$, then $q\in J^+(p)$ by Proposition~\ref{Pgh1}.
Hence
$$
 \overline{I^\pm(p)}\subset J^\pm(p)\;.
$$
The reverse inclusion is provided by
Corollary~\ref{Cpushup}, page~\pageref{Cpushup}.
\qed

As already mentioned, global hyperbolicity gives us control over
causal geodesics\restrict{ (see Theorem~\ref{Tgh1n}, page \pageref{Tgh1n} for
a proof)}:
 \ptcx{this is only partially written so far so might need rewording; and
irrelevant for continuous metrics unless
 one invokes the notion of limiting geodesics, and even then the timelike ones
are not clear; not sure where the proof is, though}

\begin{Theorem}\label{Tgh1}
Let $\Mgthreek$ be globally hyperbolic, if $q\in I^+(p)$,
respectively $q\in J^+(p)$, then there exists a timelike,
respectively causal, future directed geodesic from $p$ to
$q$.
\end{Theorem}

 \ptcx{file domainsofdependence}

\section{Domains of dependence}
 \label{SDd}

A set $\mcU\subset\mcM$ is said to be \emph{achronal} if
$$
 I^+{(\mcU)}\cap I^-{(\mcU)}=\emptyset
  \;.
$$
There is an obvious analogous definition of an \emph{acausal} set
$$J^+{(\mcU)}\cap J^-{(\mcU)}=\emptyset\;.$$

Let $\hyp$ be an achronal topological hypersurface in  a
space-time $(\mcM,g)$. (By a hypersurface we mean an embedded
submanifold of codimension one.) Unless explicitly indicated otherwise we will
assume that $\hyp$ has no boundary. The \emph{future domain of
dependence} $\mcDSpI$ of $\hyp$ is defined as the set of points
$p\in \mcM$ with the property that \emph{every past-directed
past-inextendible timelike curve starting at $p$ meets $\hyp$
precisely once}. The \emph{past domain of dependence} $\mcDSmI$ is
defined by changing \emph{past-directed past-inextendible} to
\emph{future-directed future-inextendible} above. Finally one sets
\bel{dd1}
 \mcDSI:= \mcDSpI\cup\mcDSmI
 \;.
\ee
The ``precisely" in
``precisely once" above follows  of course already from
achronality of $\hyp$; the repetitiveness in our definition is deliberate, to
emphasize
the property. We always have
$$\hyp\subset\mcDpmSI\;.$$

\begin{coco}
We have found it useful to build in the fact that $\hyp$ is a
topological hypersurface in the definition of $\mcDSpI$. Some
authors do not impose this restriction~\cite{GerochDoD}, which can
lead to  various pathologies. From the point of view of
differential equations the only interesting case is that of a
hypersurface anyway.

The domain of dependence is usually denoted by $\mcD(\hyp)$ in the literature,
and we will sometimes write so. We have added the subscript $I$ to emphasise
that the definition is based on timelike curves.
Hawking and Ellis~\cite{HE} define the domain of dependence using
causal curves instead of timelike ones, we will denote the resulting domains of
dependence by $\mcDpmSJ$, etc., if need arises. On the other hand timelike
curves are used by Geroch~\cite{GerochDoD} and by
Penrose~\cite{PenroseDiffTopo}. For $C^3$ metrics and spacelike acausal
hypersurfaces $\hyp$, the resulting sets
differ by a boundary.
The definition with causal curves
has the advantage that the resulting set $\mcDSJ$ is  open when
$\hyp$ is an acausal topological hypersurface. However, this
excludes piecewise null hypersurfaces as Cauchy surfaces, and this
is the reason why we use the definition based on timelike curves in the current
treatment.
It appears that the definition using causal curves is easier to handle when
continuous metrics are considered~\cite{ChGrant}.
\end{coco}

The following examples are instructive, and are left as exercices
to the reader; note that some of the results proved later in this
section might be helpful in verifying our claims:
\begin{figure}[t]
\begin{center} {\psfrag{hyp}{\Huge $\!\hyp$}
\psfrag{remove}{\Huge remove}
\psfrag{dpluss}{\Huge $\!\!\!\!\!\!\!\!\mcDSpI$}
\psfrag{dminus}{\Huge $\mcDSmI$}
\resizebox{5in}{!}{\includegraphics{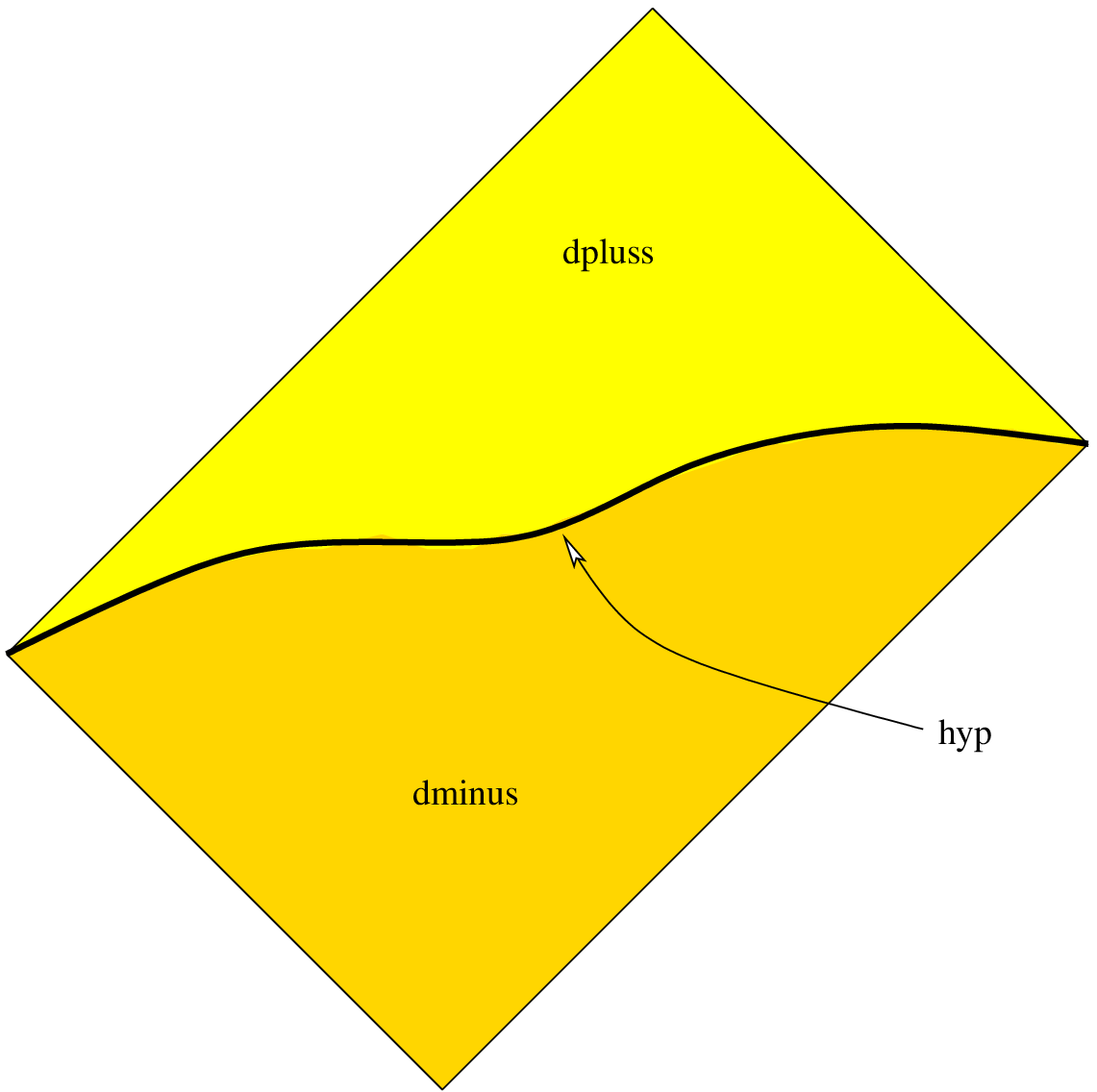}%
\includegraphics{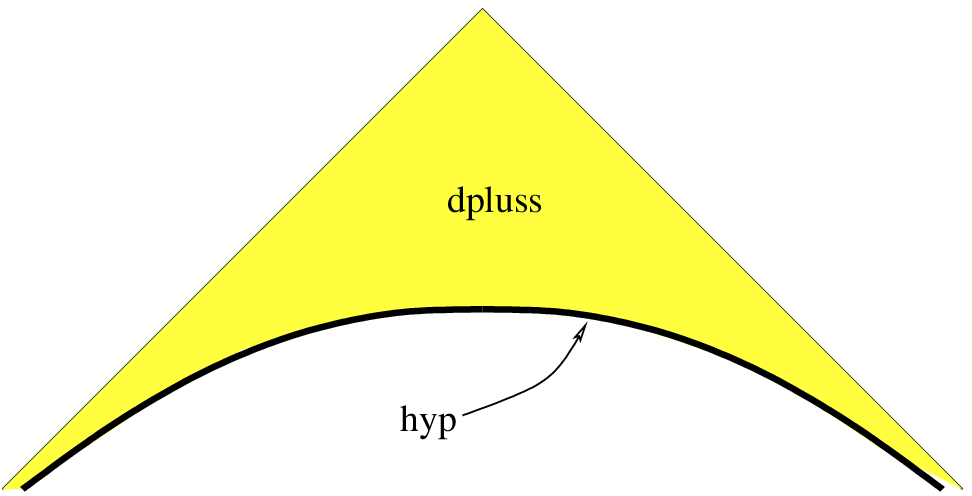}%
\includegraphics{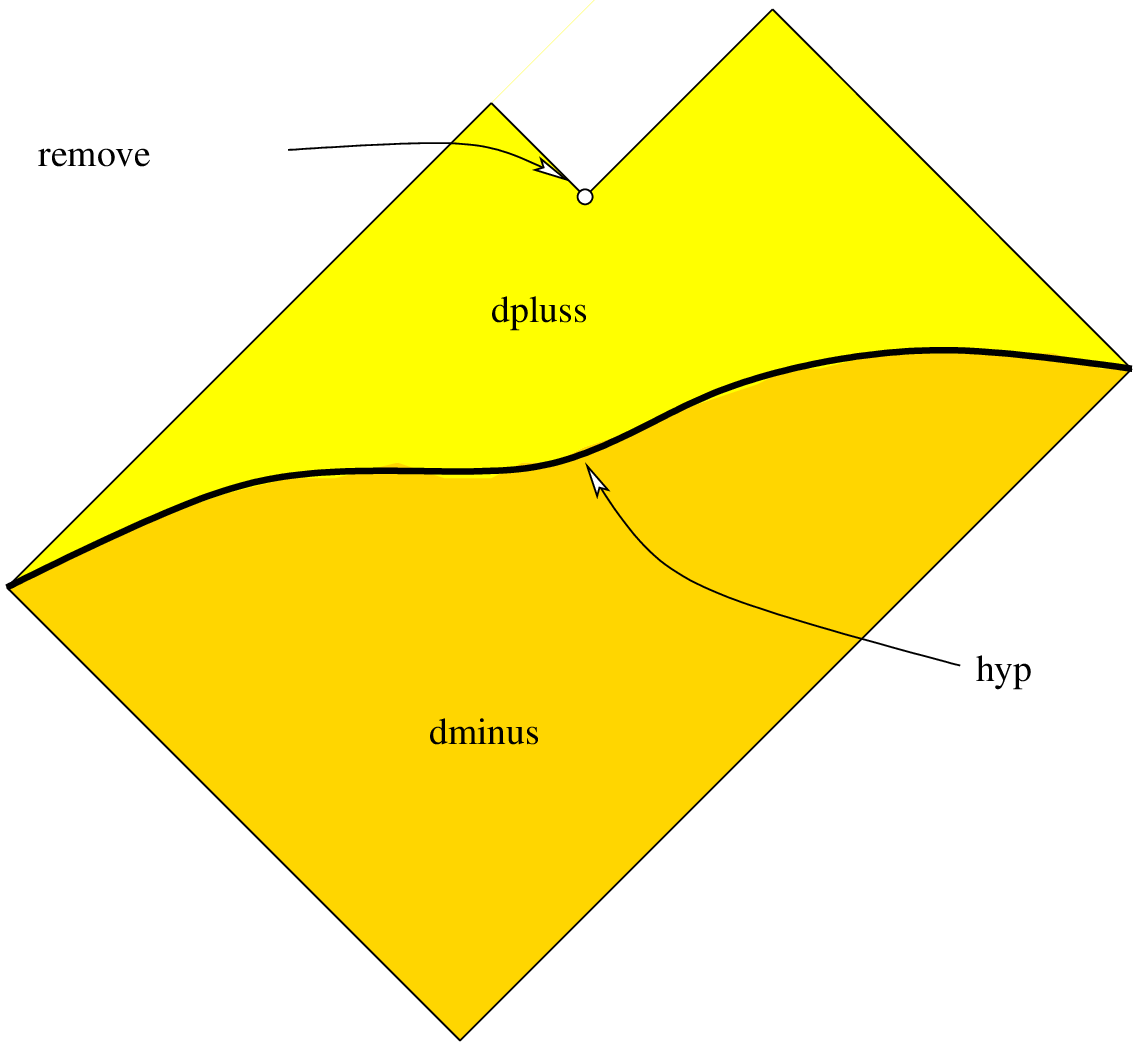}}
}
\caption{Examples of domains of dependence. \label{Fdod}}
\end{center}
\end{figure}%

\begin{Example}\label{Exdd1}
Let $\hyp=\{x^0=0\}$ in Minkowski space-time $\R^{1,n}$, where $x^0$
is the usual time coordinate on $\R^{1,n}$. Then $\mcDSI=\R^{1,n}$.
Thus both $\mcDSpI$ and $\mcDSmI$ are non-trivial, and their union
covers the whole space-time.\end{Example}

\begin{Example}\label{Exdd1a}
Let $\hyp=\{$the set of points in $\R^n$ with rational
coordinates$\}\subset \{x^0=0\}$ in Minkowski space-time $\R^{1,n}$,
where $x^0$ is the usual time coordinate on $\R^{1,n}$. Then
``$\mcDSpI=\hyp$'', in the sense that ``the set of points
$p\in \mcM$ with the property that {every past-directed
past-inextendible timelike curve starting at $p$ meets $\hyp$
precisely once}" coincides with $\hyp$. Such examples are the reason why we
assumed that $\hyp$ is a hypersurface in the definition of $\mcDSI$.
\end{Example}

\begin{Example}\label{Exdd2}
Let $\hyp=\{x^0-x^1=0\}$ in Minkowski space-time, where the
$x^\mu$'s are the usual Minkowskian coordinates on $\R^{1,n}$. Then
$\mcDSpI=\mcDSmI=\mcDSI=\hyp$ (this can be proved  using, \eg,
Lemma~\ref{Lpushup} below).
\end{Example}

\begin{Example}\label{Exdd3}
Let $\hyp=\{x^0=|x^1|\}$ in Minkowski space-time $\R^{1,n}$. Then
$\mcDSmI=\hyp$, $\mcDSpI=\{x^0\ge |x^1|\}$. The fact that
$\mcDSmI=\hyp$ makes $\mcDSmI$ rather uninteresting. On the other
hand $\mcDSpI$ coincides with the causal future of $\hyp$.
\end{Example}

\begin{Example}
\label{Exdd2a}
Let $\hyp=\dot J^+(0)$, the forward
light-cone of the origin in
 Minkowski space-time $\R^{1,n}$. Then $\mcDSpI=J^+(0)$ is the
 forward causal cone of the origin, while $\mcDSmI=\hyp$. On the other hand,
 if one removes the origin from $\dot J^+(0)$, so that $\hyp=\dot
 J^+(0)\setminus\{0\}$, then $\mcDSpI=\mcDSmI=\hyp$.
\end{Example}

\begin{Example}\label{Exdd4}
Let $\hyp=\{\eta_{\mu\nu}x^\mu x^\nu=-1\;, x^0>0\}$ be the upper
component of the unit spacelike hyperboloid in Minkowski
space-time. Then $\mcDSI=J^+(0)\;.$  Thus both $\mcDSmI$ and
$\mcDSpI$ are non-trivial, however $\mcDSmI$ does not cover the
whole past of $\hyp$.
\end{Example}

As a warm-up, let us prove the following elementary property of
domains of dependence:

\begin{Proposition}\label{Pdd1}
 \MCtwok
Let $p\in \mcDSpI$, then
$$I^-(p)\cap J^+(\hyp)\subset \mcDSpI\;.$$
\end{Proposition}

\proof Let $q\in I^-(p)\cap J^+(\hyp)$, thus there exists a
past-directed timelike curve $\gamma_0$ from $p$ to $q$.
Let $\gamma_1$ be a  past-inextendible timelike curve $\gamma_1$
starting at $q$. The curve $\gamma:=\gamma_0\cup\gamma_1$ is a
past-inextendible past-directed timelike curve starting at $p$,
thus it meets $\hyp$ precisely once at some point $r\in\hyp$.
Suppose that $\gamma$ passes through $r$ before passing through
$q$, as $q\in J^+(\hyp)$  Lemma~\ref{Lpushup0}
 \ptcx{this uses push-ups, could work with lut futures}
shows
that $r\in I^+(\hyp)$, 
contradicting achronality of $\hyp$. This shows that $\gamma$ must
meet $\hyp$ after passing through $q$, hence $\gamma_1$ meets
$\hyp$ precisely once.\qed

Let $\hyp$ be achronal, we shall say that a set $\mcO$ forms a
one-sided future neighborhood of $p\in \hyp$ if there exists an
open set $\mcU\subset \mcM$ such that $\mcU$ contains $p$ and
$$\mcU\cap J^+(\hyp)\subset \mcO\;.$$
As $I^-(p)$ is open,  Proposition~\ref{Pdd1} immediately implies:

\begin{Corollary}\label{Cdd0}
 \MCtwok
Suppose that $\mcDSpI\ne \hyp$,
consider any point $p\in \mcDSpI\setminus\hyp$. For any $q\in \hyp
\cap I^-(p)$
 \ptcx{needs $\cI$?}
the set $\mcDSpI$ forms a one-sided future
neighborhood of $q$.
\qed
\end{Corollary} 

Transversality considerations near $\hyp$ should make it clear that
the hypothesis of Corollary \ref{Cdd0} is satisfied for  achronal, $C^1$,
\underline{spacelike} hypersurfaces without boundary, and
therefore for such $\hyp$ the set $\mcDSI$ forms a neighborhood of
$\hyp$. Example~\ref{Exdd2} shows that this will not be the case for
general $\hyp$'s.
%
%

 The next theorem shows that achronal topological
hypersurfaces can
 be used to produce globally hyperbolic space-times:

\begin{Theorem}\label{Tdd1}
Let $\hyp$ be an achronal hypersurface
in $\Mgthreek$, and suppose that  the interior $\imcDSI$ of the domain
of dependence $\mcDSI$ of $\hyp$ is not empty. Then $\imcDSI$ equipped
with the metric obtained by restriction from $g$ is globally
hyperbolic.
\end{Theorem}

\proof We need first to show that a causal curve can be pushed-up
by an amount as small as desired to yield a timelike curve:

\begin{Lemma}[``Push-up Lemma II"]
\label{Lpushup}
\MCtwok
 Let $\gamma:\R^+\to \mcM$ be a past-inextendible past-directed
 \underline{causal}
curve starting at $p$, and let $\mcO$ be a neighborhood of the
image $\gamma(\R^+)$ of $\gamma$. Then for every $r\in I^+({q})\cap\mcO$ there
exists a past-inextendible past-directed
 \underline{timelike}
curve $\hat\gamma$ starting at $r$ such that
\beal{27IV11.1}
 &
 \hat \gamma \subset
 I^+(\gamma)\cap \mcO\;,
 &
\\
 &
 \forall\ s\in [0,\infty) \qquad I^-(\hat \gamma (s))\cap \gamma(\R^+) \ne
\emptyset\;.
 &
\eeal{27IV11.2}
%
\end{Lemma}

\proof The construction is essentially identical to that of the
proof of Lemma~\ref{Lpushup0}, p.~\pageref{Lpushup0}, except that we will have
to deal
with a countable collection of curves, rather than a finite number.
One also needs to make sure that the final curve is inextendible.
As usual, we parameterize $\gamma$ by $h$--distance as measured
from $p$. Using an exhaustion of $[0,\infty)$ by compact intervals
$[m,m+1]$ we cover $\gamma$ by a countable collection
$\mcU_i\subset \mcO$, $i\in \N$ of elementary regions $\mcU_i$
centered at $$p_i=\gamma(r_i)$$  with
$$
 p_1= p\;, \quad p_i\in \mcU_{i}\cap\mcU_{i+1}\;, \quad p_{i+1}\subset
J^-{(p_i)}
  \;.
$$
We further  impose the following condition on the
$\mcU_i$'s:  if  $r_i\in[j,j+1)$,
then the corresponding $\mcU_i$
is contained in a $h$-distance ball $B_h(p_i,1/(j+1))$.

Let $\gamma_0:[0,s_0]\to\mcM$ be a past directed causal curve from
$r$ to $p\in\mcU_1\cap\mcU_2$; let $s_1$ be close enough to $s_0$
so that
$$\gamma_0(s_1)\in\mcU_2\;.$$ Proposition~\ref{PC1}, p.~\pageref{PC1}, together
with
the definition of elementary regions shows that
 there exists a past directed timelike curve
$\gamma_1:[0,1]\to\mcU_1\subset\mcO$ from $q$ to
 $p_2$. (In particular $\gamma_1\setminus \{p\} \subset I^+(p)\subset
I^+(\gamma)$).
 Similarly, for any $s\in [0,1]$ there exists a a past directed
 timelike curve $\gamma_{2,s}:[0,1]\to\mcU_2\subset\mcO$
  from $\gamma_1(s)$ to
 $p_2$. We choose $s=:s_2$ small enough so that
 $$\gamma_1(s_2)\in \mcU_3\;.$$
 One
 repeats that construction  iteratively, obtaining a sequence of
 past-directed
 timelike curves $\gamma_i\subset I^+(\gamma)\cap  \mcU_i\subset I^+(\gamma)\cap
\mcO$
 such that the end point of
 $\gamma_i$ lies in $\mcU_{i+1}$ and
 coincides with the starting point of $\gamma_{i+1}$.  Concatenating those
curves together gives the desired path $\hat
  \gamma$. Since every path $\gamma_i$ lies in $I^+(\gamma)\cap
  \mcO$,  so does their union.

 Since $\gamma_i\subset \mcU_i\subset B_h(p_i,1/(j+1))$ when
 $r_i\in[j,j+1)$
 we obtain, for ${r\in[j,j+1)}$,
$$
\distb(\gamma(r),\hat\gamma)\le
\distb(\gamma(r),\gamma(r_i))+\distb(\gamma(r_i),\gamma_i)\le
\frac 2 {j+1}\;,$$ where we have ensured that
$\distb(\gamma(r),\gamma(r_i))< 1/(j+1)$ by choosing $r_i$
appropriately. It follows that
\bel{edd0}
\distb(\gamma(r),\hat\gamma)\le \frac 2 {r}\;.
 \ee

To finish the proof, suppose that $\hat \gamma:[0,s_*)\to\mcM$ is
extendible,  call $\hat p$ the end point of $\hat \gamma$. By
\eq{edd0}
$$
\lim_{r\to \infty} \distb(\gamma(r),\hat p)=0\;.
$$
Thus $\hat p$ is an end point of $\gamma$, which together with
Theorem~\ref{TP4a}  contradicts inextendibility of $\gamma$.
 \qed

By the definition of domains of dependence, inextendible
\emph{timelike} curves through $p\in \mcDSpI$
 intersect all  the sets $\hyp$, $I^+(\hyp)$, and  $I^-(\hyp)$. This is wrong in
general for inextendible
 \emph{causal} curves through
 points in
$\mcDSpI\setminus\imcDpSI$, as shown on Figure~\ref{FNGNIC}.
\begin{figure}[th]
\begin{center}
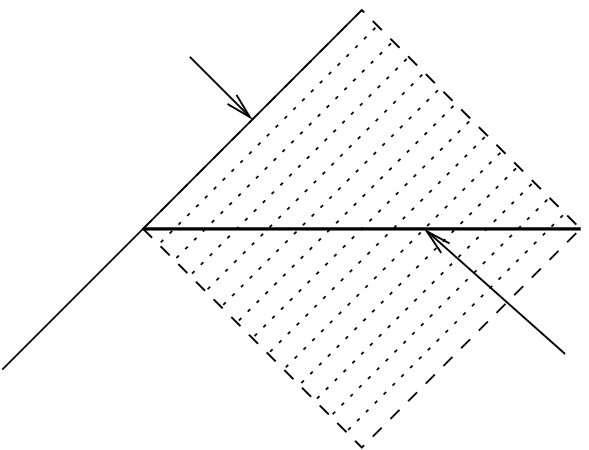
\end{center}
 \caption{\label{FNGNIC} Let $\hyp=\{t=0,x\in (-1,1)\}\subset \R^{1,1}$, then
 $\mcDSI$ is the closed dotted diamond region \emph{without} the two rightermost
and leftermost points that
 lie on the closure
 of $\hyp$. The past directed null geodesic $\gamma$ starting at $(1,1) \in
\mcDSpI$ does not intersect $\hyp$.}
\end{figure}
 Nevertheless we have:

\begin{Lemma}\label{Ldd1} If $p\in\imcDSI$, then every inextendible
\underline{causal} curve $\gamma$ through $p$ intersects $\hyp$,
$I^-(\hyp)$ and $I^+(\hyp)$.
\end{Lemma}

\begin{Remark} In contradistinction with timelike curves, for causal curves the
intersection of $\gamma$ with $\hyp$ does not have to be a point. An
example is given by the hypersurface $\hyp$ of Figure~\ref{FPNC}.
 \ptcx{the figures should be numbered sequentially with the theorems
 and stuff}
\begin{figure}[hbt]
\begin{center}
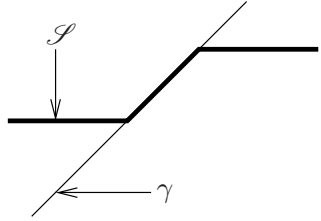
\end{center}
 \caption{\label{FPNC}A null geodesic $\gamma$ intersecting an achronal
topological hypersurface
 $\hyp$ at more than one point.}
\end{figure}
\end{Remark}

\proof Changing time-orientation if necessary we may suppose that
$p\in\mcDSpI$. Let $\gamma:I\to\mcM$ be any past-directed
inextendible causal curve  through $p$. Since $p$ is an interior
point of $\mcDSpI$ there exists $q\in I^+(p)\cap \mcDSpI$. By the
Push-up Lemma~\ref{Lpushup} with $\mcO=\mcM$  there exists a
past-inextendible past-directed timelike curve $\hat \gamma$
starting at $q$ which lies to the future of $\gamma$.
The inextendible timelike curve $\hat \gamma$ enters $I^-(\hyp)$, and so does
$\gamma$ by \eq{27IV11.2}.

If $p\in \hyp$, we can repeat the argument
above with the time-orientation changed, showing that $\gamma$
enters $I^+(\hyp)$ as well, and we are done.

Otherwise $p\not\in\hyp$, then $p$
is necessarily in $I^+(\hyp)$, hence $\gamma$ meets $I^+(\hyp)$ as well.
Now, each of the two disjoint sets
$$I_\pm:=\{s\in I: \gamma(s)\in I^\pm(\hyp)\}\subset \R$$ is
open in the connected interval $I$. They  cover $I$ if $\gamma$
does not meet $\hyp$, which implies that either $I_+$ or  $I_-$
must be empty when $\gamma\cap \hyp=\emptyset$.
But we have shown that both $I_+$
and $I_-$ are not empty, and so $\gamma$  meets
$\hyp$, as desired.
%
\hfill\qed

\medskip

Returning to the proof of Theorem~\ref{Tdd1}, suppose that $\imcDSI$
is not strongly causal. Then there exists $p\in \imcDSI$ and a
sequence $\gamma_n:\R\to\imcDSI$ of inextendible causal curves which
exit the $h$-distance geodesic ball $B_h(p,1/n)$ (centred at $p$ and
of radius $1/n$) and renter $B_h(p,1/n)$ again.  Changing
time-orientation of $\mcM$ if necessary, without loss of generality
we may assume that $p\in I^-(\hyp)\cup \hyp$. Note that the property
``leaves and reenters" is invariant under the change $\gamma_n(s)\to
\gamma_n(-s)$, so that changing the orientation of (some of) the
$\gamma_n$'s if necessary, there is no loss of generality in
assuming the $\gamma_n$'s to be future-directed. Finally, we
reparameterize the $\gamma_n$'s by $h$--distance, with
$\gamma_n(0)\in B_h(p,1/n)$. Then, there exists a sequence $s_n>0$
such that $\gamma_n(s_n)\in B_h(p,1/n)$, with $\gamma_n(0)$ and
$\gamma_n(s_n)$ lying on  different connected components of
$\gamma\cap B_h(p,1/n)$.

Let $\mcO$ be an elementary neighborhood of $p$, as in
Definition~\ref{Delem}, p.~\pageref{Delem}, and let $n_0$ be large enough so
that
$\overline {B_h(p,1/n_0)} \subset \mcO$. Note that the local
coordinate $x^0$ on $\mcO$ is monotonous along every connected
component of $\gamma_n\cap \mcO$ which implies, for $n\ge n_0$, that
any causal curve which exits and reenters $B_h(p,1/n)$ has also to
exit and reenter $\mcO$. This in turn guarantees the existence of an
$\epsilon>0$ such that $s_n>\epsilon$ for all $n\ge n_0$.

Let $\gamma$ be an accumulation curve
through $p$ of the $\gamma_n$'s, passing to a subsequence if necessary the
$\gamma_n$'s
converge uniformly to $\gamma$ on compact subsets of $\R$. The curve
$\gamma$ is causal and $p$ is, by hypothesis, in the interior of
$\mcDSI$, we can therefore invoke Lemma~\ref{Ldd1} to conclude that
there exist $s_\pm\in\R$ such that $\gamma(s_-)\in I^-(\hyp)$ and
$\gamma(s_+)\in I^+(\hyp)$.  Since $\hyp$ is achronal and $\gamma$ is
future directed we must have $s_-<s_+$.  Since the $I^\pm(\hyp)$'s are open, and
since  (passing to a subsequence if necessary)
$\gamma_n(s_\pm)\to\gamma(s_\pm)$, we have $\gamma_n(s_\pm)\in
I^\pm(\hyp)$ for $n$ large enough.

The situation which is simplest to exclude is the one where the sequence
$\{s_n\}$ is
bounded. Then there exists $ s_*\in \R$ such that, passing again to
a subsequence if necessary, we have $s_n\to s_*$. Note that
$\gamma_n(s_*)\to p$ and that $s_*\ge \epsilon$.  Since
$\gamma_n|_{[0,s_*]}$ converges uniformly to $\gamma|_{[0,s_*]}$, we
obtain an inextendible periodic causal curve $\gamma'$ through $p$
by repetitively circling from $p$ to $p$ along
$\gamma|_{[0,s_*]}$.
By Lemma~\ref{Ldd1} $\gamma'$ meets all of $\hyp$, $I^+(\hyp)$ and $I^-(\hyp
)$, which is
clearly incompatible with periodicity of $\gamma'$ and achronality
of $\hyp$.
 (In detail: there exist points $q_\pm\in \gamma'\cap
I^\pm(\hyp )$. Following backwards $\gamma'$ from $p$ to $q_+$ we
obtain $q_+\in J^-(p)$. But $I^-(q_+)\cap \hyp \ne \emptyset$, and
Lemma~\ref{Lpushup0} implies that $I^-(p)\cap \hyp \ne \emptyset$.
Since $p\in I^-(\hyp)\cup\hyp$, this contradicts achronality of
$\hyp$.)

Note that if $p\in I^-(\hyp)$ we would need to have $s_n \le s_+$ for
$n$ large enough: Otherwise, for $n$ large, we could follow
$\gamma_n$ in the future direction from $\gamma_n(s_+)\in I^+(\hyp)$
to $\gamma_n(s_n)\in I^-(\hyp)$, which is not possible if $\hyp$ is
achronal. But then the sequence $s_n$ would be bounded, which has already been excluded.
So $p\in I^-(\hyp)$ cannot occur.

There  remains the  possibility $p\in \hyp$. We
then must have $\gamma_n|_{[0,\infty)}\cap I^-(\hyp)=\emptyset$,
otherwise we would obtain a contradiction with achronality of $\hyp$
by following $\gamma_n$ to the future  from $p=\gamma_n(0)\in \hyp$
to a point where $\gamma_n$ intersects $ I^-(\hyp)$. Set
$$
\hat \gamma_n (s) = \gamma_n(s+s_n)
 \;,
$$
then $\hat \gamma_n$ accumulates at $p$ since $\hat \gamma_n(0)\to
p$, therefore there exists an inextendible accumulation curve  $\hat
\gamma:\R\to\mcM$ of the $\hat \gamma_n$'s passing through $p$. As
$$
 \hat
\gamma_n({[-s_n,\infty)})\cap I^-(p)=\gamma_n({[0,\infty)})\cap
I^-(\hyp)=\emptyset
$$
 we have $\hat \gamma(\R) \cap I^-(p)=\emptyset$ as
well. This is, however, not possible if $p\in \imcDSI$ by
Lemma~\ref{Ldd1}. We see that the possibility that $p\in\hyp$ cannot occur either. We
conclude that $\imcDSI$ is strictly causal, as desired.

To finish the proof, we need to prove compactness of the sets
of the form
$$J^+(p)\cap J^-(q)\;,\quad p,q\in \imcDSI\;.$$
If $p$ and $q$ are such that this set is empty or equals $\{p\}$
there is nothing to prove. Otherwise, consider a sequence $r_n \in
J^+(p)\cap J^-(q)$. One of the following is true:
\begin{enumerate}
\item we have $r_n\in I^-(\hyp)\cup \hyp$ for all $n\ge n_0$, or
\item there exists a subsequence, still denoted by $r_n$, such that $r_n\in
I^+(\hyp)$.
\end{enumerate}
In the second case we change time-orientation, pass to a
subsequence, rename $p$ and $q$, reducing the analysis to the first
case. Note that this leads to $p\in I^-(\hyp)\cup \hyp$.

By definition, there exists a future directed causal curve
$\hat\gamma_n$ from $p$ to $q$ which passes through $r_n$,
\bel{edd1}
 \hat \gamma_n(s_n)=r_n\;.
 \ee
Let $\gamma_n$ be any $\distb$-parameterized, inextendible future
directed causal curve extending $\hat\gamma_n$, with
$\gamma_n(0)=p$. Let $\gamma$ be an inextendible accumulation curve
of the $\gamma_n$'s, then $\gamma$ is a future inextendible causal
curve through
$$
 p\in (\mcDSmI\cup \hyp)\cap \imcDSI\;.$$
By Lemma~\ref{Ldd1} there
exists $s_+$ such that $\gamma(s_+)\in I^+(\hyp)$. Passing to a
subsequence, the $\gamma_n$'s converge uniformly to $\gamma$ on
$[0,s_+]$, which implies that for $n$ large enough the
$\gamma_n|_{[0,s_+]}$'s  enter $I^+(\hyp)$. This, together with
achronality of $\hyp$, shows that the sequence $s_n$ defined by
\eq{edd1} is bounded; in fact we must have $0\le s_n\le s_+$.
Eventually passing to another subsequence we thus have $s_n\to
s_\infty$ for some $s_\infty\in \R$. This implies
$$
 r_n\to \gamma(s_\infty)\in J^+(p)\cap  J^-(q)\;,
$$
which had to be established.
 \qed

\ptcx{a definition of topological hypersurface with boundary commented out, as
part of removing an associated inconclusive lemma}
\restrict{
Let $\hyp$ be a topological hypersurface in $\mcM$ with boundary%
\footnote{By this we mean that near every point of $\hyp$ there exists an open
$C^0$-coordinate
neighborhood $\mcO$ in $\mcM$ so that $\hyp\cap\mcO$ is represented as
$\{x^0=0\}$ or $\{x^0=0\;, \ x^1\ge 0\}$.}%
, then the
interior of $\hyp$ in the topology of $\mcM$ is of course
empty. However, $\hyp$ can be equipped with the induced topology, in
which open subsets of $\hyp$ are defined as intersections of open
subsets of $\mcM$ with $\hyp$. We will use the symbol $\ihyp$ for
the interior of $\hyp$ with respect to this topology, hoping that
this will not lead to confusions for the reader. The following
strengthening of Lemma~\ref{Ldd1} is often useful:
 \ptcx{say something somewhere about coverings of $\mcM$ which render spacelike
 hypersurfaces achronal}

\begin{Lemma}\label{Ldd1a} A point $p\in\mcM$ is in $ \imcDSI$ if and only if
every inextendible
\underline{causal} curve $\gamma$ through $p$ intersects $\ihyp$,
$I^-(\ihyp)$ and $I^+(\ihyp)$.
 \ptcr{I do not know if this is correct, if it is, then it should be
 proved and should replace Lemma~\ref{Ldd1}; at this stage it can be safely
discarded as it is not used anywhere;
 also reword the following paragraph if kept}
\end{Lemma}

\proof
 $\Longrightarrow$ Let $\gamma$ be as in the statement of the lemma, then
$\gamma$ intersects $\hyp$ by Lemma \ref{Ldd1a}. Since $\imcDSI$ is
open there exists a point $r\in \imcDSI$ to the future of $p$, then
$I^-(r)\cap \hyp$ is an open neighborhood of $\gamma \cap \hyp$.

$\Longleftarrow$ By definition $p\in \mcDSI$. Suppose that
$p\not\in\mcDSI$, then there exists a sequence of  inextendible
timelike curves $\gamma_n$ with $\gamma_n(0)\to p$ such that
$\gamma_n\cap \hyp=\emptyset$. Let $\gamma$ be an accumulation
curve, then there exists $s_*$ such that $\gamma(s_*) \in \hyp$.
 \ptc{ so what ? }

 \qed
}
%
%
%
%

We have the following
characterisation of \emph{interiors} of domains of dependence:

\begin{Theorem}
 \label{Tdd3}
 \MCthreek
Let $\hyp$ be a differentiable acausal spacelike hypersurface.
A point $p\in\mcM$ is in $ \imcDpSI$ if and only if
\bel{eTdd3}
 \mbox{ the set }\ {J^-(p)\cap \hyp}\ \mbox{ is
 non-empty, and compact as a subset of $\ihyp$.}
 \ee
\end{Theorem}

\begin{Remark}\label{RTdd3} The set $\hyp=\{t=0,x\in[-1,1]\}\subset \R^{1,1}$
(compare Figure~\ref{FNGNIC}, but note that a different $\hyp$ was
meant there) shows that the condition \eq{eTdd3} cannot be
replaced by the requirement that the set $
\overline{I^-(p)\cap \hyp} $ { is non-empty, and \emph{compact as
a subset of} $\mcM$.}
\end{Remark}

 \proof
For $p\in \imcDpSI$ compactness of  $\overline{I^-(p)\cap \hyp}$
can be established by an argument very similar to that given in
the last part of the proof of Theorem~\ref{Tdd1}, the details are
left to the reader.

In order to prove the reverse implication assume that \eq{eTdd3}
holds, then there exists a  future directed causal curve
$\gamma:[0,1]\to\mcM$  from some point $q\in\hyp$ to $p$. Set
$$I:= \{t\in (0,1]: \gamma(s)\in \imcDpS\ \mbox{\ for all $0<s\le t$}\}\subset
(0,1]\;.$$
Now, elementary considerations show that for $C^1$, spacelike, acausal
hypersurfaces we have $\gamma(t)\in \imcDpSI$ for  $t>0$ small enough, hence $I$
is not empty.
Clearly $I$ is open in $(0,1]$. In order to show that it equals
$(0,1]$ set
$$
t_*:= \sup I\;.
$$
Consider any past-inextendible past-directed
causal curve $\hat \gamma$ starting at $\gamma(t_*)$. For $t<t_*$
let $\hat \gamma_t$ be a family of past-inextendible causal
push-downs of $\hat\gamma$ which start at $\gamma(t)$, and which
have the property that
$$
 \distb(\gamma_t(s),\gamma(s))\le  |t-t_*|
 \ \mbox{ for} \  0\le s \le 1/|t-t_*|\;.
$$
Then $\hat\gamma_t$ intersects
$\hyp$ at some point $q_t\in J^-(p)$. Compactness of
$J^-(p)\cap\hyp$ implies that the curve $t\to q_t \in \hyp$
accumulates at some point $q_*\in \hyp$, which clearly is the
point of intersection of $\gamma$ with $\hyp$. This shows that
every causal curve $\gamma$ through $\gamma(t_*)$ meets $\hyp$, in
particular $\gamma(t_*)\in \mcDSpI$. So $I$ is both open and closed in $(0,1]$,
hence $I=(0,1]$, and the result is proved.
\qed

\section{Cauchy horizons}
\label{SCauchyHorizons}
\newcommand{\mcHpSI}{{\mcH}^+_I(\hyp)}
\newcommand{\mcHmSI}{{\mcH}^-_I(\hyp)}
\newcommand{\mcHSI}{{\mcH}_I(\hyp)}

\ptcx{this section needs expanding,
also this is a special case of acausal boundaries, so maybe that's
the way to handle that}

\begin{Definition}
\label{DCauchyHorizon} Let $\hyp$ be an achronal topological hypersurface. The
\emph{future Cauchy horizon} $\mcHpSI$ of $\hyp$ is defined as
 \ptcx{changed to $\mcHpSI_J$ needs crosschecking}
$$
\mcHpSI= \mcDSpI\setminus I^-(\mcDSpI)
 \;,
$$
with an obvious corresponding definition for the \emph{past Cauchy
horizon} $\mcHmSI$. One defines the \emph{Cauchy horizon} as
$$
 \mcHSI= \mcHmSI\cup\mcHpSI
 \;.
$$
\end{Definition}

\begin{coco}
Our definition follows that of Penrose~\cite{PenroseDiffTopo}.
Similarly to the domains of dependence, the usual notation for Cauchy horizons
is $\mcH$ and not $\mcH_I$, and we will sometimes write so. The analogous
definition of future
Cauchy horizon  with
$\mcDSpI$   replaced by $\mcDSpJ$ leads in general to essentially
different sets for continuous metrics~\cite{ChGrant}.
\end{coco}

It is instructive to consider a few examples:

\begin{Example}\label{Exdd1c}
Let $\hyp=\{x^0=0\}$ in Minkowski space-time $\R^{1,n}$, where $x^0$
is the usual time coordinate on $\R^{1,n}$. Then $\mcHSI=\emptyset$.
\end{Example}

\begin{Example}\label{Exdd1cc}
Let $\hyp$ be the open unit ball in $\R^n$, viewed as a subset of $\{x^0=0\}$ in
Minkowski space-time $\R^{1,n}$, where $x^0$
is the usual time coordinate on $\R^{1,n}$. Then $\mcHpSI $ is the intersection
of the past-light cone of the point $(x^0=1,\vec x = 0)$ with $\{x^0\ge 0\}$.
\end{Example}

\begin{Example}\label{Exdd4c}
Let $\hyp=\{\eta_{\mu\nu}x^\mu x^\nu=-1\;, x^0>0\}$ be the upper
component of the unit spacelike hyperboloid in Minkowski
space-time. Then $\mcHpSI=\emptyset $, while $\mcHmSI $ coincides with the
future
light-cone of the origin.
\end{Example}

\begin{Example}\label{Exdd4cc}
\restrict
{Let $\hyp$ be the three-dimensional sphere $\{t=0\}$ in the Taub-NUT space-time,
described in Appendix~\ref{ATaubNut}.  Then
 \ptcx{do something with this; one can use the text which is restricted}
}
The Taub-NUT space-times~\cite{Taub,NUT} provide examples of space-times with an achronal spacelike $S^3$ and with two corresponding past and future Cauchy horizons, each diffeomorphic to $S^3$~\cite{MisnerTaub}.
\end{Example}

For any open set set $\Omega$ one has $\Omega\setminus
I^{-}(\Omega)=\emptyset$, which shows that
\bel{SCH1}
 \imcDSI\cap \mcHpSI=\emptyset
 \;.
\ee
It follows that $\mcHpSI$ is a subset of the topological boundary
$\partial\mcDSpI$ of $\mcDSpI$:
\bel{mCHboun}
 \mcHpSI \subset \partial\mcDSpI:= \overline{\mcDSpI}\setminus \imcDpS
  \;.
 \ee

 \ptcx{this might be overlapping with the joint analysis with greg in the continuous case?
 crossrefer appropriately if kept}
\medskip

The important notion of \emph{generators} of  horizons stems from the following result in which we assume, for simplicity,
 \ptcx{do a more general statement allowing more general hypersurfaces, with a different proof?; one could use  Lemma~\ref{Ldd1} in the general case of not necessarily spacelike hypersurfaces} that $\hyp$ is differentiable and spacelike:

\begin{Proposition}
\label{PCH1}
Let $\hyp$ be a spacelike $C^1$ hypersurface in $\Mgthreek$.
For any $p\in \mcHpSI$ there exists a  past directed
null geodesic  $\gamma_p\subset \mcHpSI$ starting at $p$
which either does not have an endpoint in $\mcM$, or has an endpoint
on $\ohyp\setminus \hyp$.
 \ptcx{ could use limit-geodesics
in general}
\end{Proposition}

\begin{Remark}
 There might be more than one such geodesic for some points on the horizon.
\end{Remark}

\proof
Let $p\in \mcHpSI$, then there exists a sequence of points $p_n\not \in \mcDSpI$ which converge to $p$, and past inextendible timelike curves $\gamma_n$ through $p_n$ that do not meet $\hyp$.
Let $\gamma $ be an accumulation curve of the $\gamma_n$ through $p$. Then $\gamma $ does not meet $ \hyp$: indeed, if it did, then the $\gamma_n$'s would be meeting $\hyp$ as well for all $n$ large enough.
   \ptcx{needs a comment, think about the set-up, needs an acausal? think of a piece
   of hyp which is null and ends there, with all generators sliding on it}
If $\gamma $ meets $\ohyp $, we let $\gamma_p$ be the segment of $\gamma$ from $p$ to the intersection point with $\ohyp $, otherwise we let $\gamma_p=\gamma$.

The curve $\gamma_p$ is achronal: otherwise  $\gamma_p$  would enter the interior of $\mcDSpI$, but then it would have to intersect $\hyp$ by Lemma~\ref{Ldd1}.
We can thus invoke  Proposition~\ref{Pro:acausalgeo}, p.~\pageref{Pro:acausalgeo}, to conclude  that $\gamma_p$ is a null geodesic.

It remains to show that $\gamma_p\subset \mcHpSI$. Let $\gamma$ be an inextendible past-directed timelike curve through a point $q$ on $\gamma_p$, with $q\not\in\ohyp$. Let $\mcO$ be a neighborhood of $q$ that does not meet $\hyp$,  and let $r\in \mcO $ be a point on $\gamma$ lying to the timelike past of $q$. By Corollary~\ref{CPushup0} there exists a timelike curve $\gamma_1$ from $p$ to $r$. Consider the past-inextendible timelike curve, say $\gamma_2$, obtained by following $\gamma_1$ from $p$ to $r$, and then following $\gamma$ to the past. Since $p\in \mcHpSI$ the curve $\gamma_2$ has to meet $\hyp$. As $\gamma_1$ does not meet $\hyp$, it must be the case that $\gamma$ meets $\hyp$, and so $q\in\mcHpSI$.
\qed

For any $p\in \mcHpSI$ let $\hat \gamma_p$ denote a maximal future
extension of the geodesic segment $\gamma_p$ of
Proposition~\ref{PCH1}, and set $\tilde \gamma_p= \gamma_p\cap
\mcHpSI$. (Note that $\hat \gamma_p$ might exit $\mcHpSI$ when
followed to the future, an example of this can be seen in
Figure~\ref{FNGNIC}.) Then $\tilde \gamma_p$ is called \emph{a
generator of $\mcHpSI$.} Using this terminology, Proposition
\ref{PCH1} can be reworded as the property that every $p\in \mcHpSI$
is either an interior point or a future end point of a generator of
$\mcHpSI$. If $\ohyp =\hyp$, then generators of $\mcHpSI$ do not have
past end points, remaining forever on $\mcHpSI$ to the past.
%


\subsection{Semi-convexity of future horizons}
 \label{Aconv}

 \ptcx{  UPDATE FOR
CONTINUOUS CAUSALITY? section synchronized on 14X11}
A hypersurface $\mcH\subset \mcM$ will be said to be
{\em future null geodesically ruled} if every point $p\in \mcH$
belongs to a future inextensible null geodesic $\Gamma \subset
\mcH$; those geodesics are called {\em the generators} of
$\mcH$. We emphasize that the generators are allowed to have past
 \ptcx{the generator buisiness seems to require geodesics; perhaps approximate ones? but
 then leave the smooth treatment here, and update sometimes later?}
endpoints on $\mcH$, but no future endpoints.  {\em Past null
geodesically ruled} hypersurfaces are defined by changing the time
orientation. Examples of future geodesically ruled hypersurfaces
include past Cauchy horizons ${\cal
  D}^-(\hyp)$ of achronal sets $\hyp$ of Proposition~\ref{PCH1} (compare~\cite[Theorem~5.12]{PenroseDiffTopo})
   and black hole event horizons $\dot
J^-(\scri^+)$~\cite[p.~312]{HE}.
 \ptcx{synchronize with the relevant section}

Note that our definition involves explicitly geodesics, and therefore throughout this section we assume that the metric is twice-continuously differentiable.
It should be kept in mind that the notion of the generator of a horizon in space-times with merely continuous metrics is not completely clear, so allowing metrics of lower differentiability might require a reformulation of the problem.
 \ptcx{can I do this?}

We always assume that the space-time dimension $\mathrm{dim} \mcM$ is $n+1$.

Suppose that $\cal O$ is a domain in
$\R^n$. Recall that a continuous function $f: \cal O\to \R$ is
called semi--convex if there exists a $C^2$ function $\phi: \cal
O\to \R$ such that $f+\phi$ is convex. We shall say that the graph
of $f$ is a semi--convex hypersurface if $f$ is semi--convex. A
hypersurface $\mcH$ in a manifold $\mcM$ will be said semi--convex if
$\mcH$ can be covered by coordinate patches ${\cal U}_\alpha$ such
that $\mcH\cap {\cal U}_\alpha$ is a semi--convex  graph for each
$\alpha$.

Consider an achronal hypersurface $\mcH\ne \emptyset$ in a globally
hyperbolic space--time $(\mcM,g)$.
 \ptcx{$\Mgz$?}
Let $t$ be a time function on $\mcM$
which induces a diffeomorphism of $\mcM$ with $\R \times \hyp$ in the
standard way~\cite{GerochDoD,Seifert}, with the level sets
$\hyp_\tau\equiv \{p| t(p)=\tau\}$ of $t$ being Cauchy surfaces.
As usual we identify $\hyp_0$ with $\hyp$ and, in the
identification above, the curves $\R\times \{q\}$, $q\in\hyp$, are
integral curves of $\nabla t$. Define
 \ptcx{problem with the integral curves of nabla t if the metric is only continuous, but
 this is most likely not needed}
\begin{equation}
  \label{sh}
  \sHh =\{q\in \hyp\ |\  \R\times \{q\}\ \mbox{\rm intersects}\ \mcH \}\ .
\end{equation}
For $q\in\sHh $ the set $(I\times {q} )\cap \mcH$ is a point by
achronality of $\mcH$, which will  be denoted by $(f(q),q)$.
Thus an achronal hypersurface $\mcH$ in a globally hyperbolic
space--time is a graph over $\sHh $ of a function $f$. The
invariance-of-the-domain theorem shows that $\sHh $ is an open
subset of $\hyp$. We have the following:

\begin{theorem}
\label{semiConvT1}\chindex{horizons!semi-convexity of} Let
$\mcH\ne \emptyset$ be an achronal future null geodesically
ruled hypersurface in a globally hyperbolic space--time
$(\mcM=\R\times\hyp,g)_{C^2}$. Then $\mcH$ is the graph of a
semi--convex function $f$ defined on an open subset $\sHh $ of
$\hyp$, in particular $\mcH$ is semi--convex.
\end{theorem}

\begin{proof}
  As discussed above, $\mcH$ is the graph of a function $f$.  The idea
  of the proof is to show that $f$ satisfies a variational principle,
  the semi--concavity of $f$ follows then by a standard argument. Let
  $p\in \mcH$ and let $\cal O$ be a coordinate patch in a neighborhood
  of $p$ such that $x^0 = t$, with $\cal O$ of the form $I\times
  B(3R)$, where $B(R)$ denotes a coordinate ball centered at $0$ of
  radius $R$ in $\R^3$, with $p=(t(p), 0)$. Here $I$ is the range of
  the coordinate $x^0$, we require it to be a bounded interval the
  size of which will be determined later on. We further assume that
  the curves $I\times \{\vx\}$, $\vx\in B(3R)$, are integral curves of
  $\nabla t$. Define
  $$
  {\cal U}_0=\{\vec x \in B(3R)| \mathrm{\ the\ causal\ path }\
  I\ni t \to (t,\vec x) \ \mathrm{intersects }\ \mcH\}\ .
  $$
  We note that ${\cal U}_0$ is non--empty, since $0\in {\cal U}_0$.
  Set \be\label{Hsection} \mcH_\sigma=\mcH\cap \hyp_\sigma \ , \ee and
  choose $\sigma$ large enough so that ${\cal O}\subset I^-(
  \hyp_\sigma)$.
  \ptcx{requires making cal O smaller? and problems if I not open? take $\check I$, but then needs to make cal O smaller?}
Now $p$ lies on a future inextensible generator
 \ptcx{C two metric?}
  $\Gamma$ of $\mcH$, and global hyperbolicity of $(\mcM,g)$ implies that
  $\Gamma\cap \hyp_\sigma $ is nonempty, hence $\mcH_\sigma$ is
  nonempty.

  For $\vx\in B(3R)$ let $\cP(x)$ denote the collection of piecewise
  differentiable future directed null curves $\Gamma: [a,b]\to \mcM$ with
  $\Gamma(a)\in \R\times \{\vx\}$ and $\Gamma(b)\in \mcH_\sigma$. We
  define \be \label{tau} \tau(\vx) = \sup_{\Gamma \in \cP(x)}
  t(\Gamma(a))\ . \ee We emphasize that we allow the domain of
  definition $[a,b]$ to depend upon $\Gamma$, and that the ``$a$"
  occurring in $ t(\Gamma(a))$ in \eq{tau} is the lower bound for the
  domain of definition of the curve $\Gamma$ under consideration.

We have the following result (compare~\cite{AP,GMP,Perlick}):

\begin{Proposition}[Fermat principle]\label{P1}\chindex{Fermat's
    principle} For $\vx\in\cU_0$ we have
  $$
  \tau(\vx)=f(\vx)\ .$$
\end{Proposition}
\begin{proof}
  Let $\Gamma$ be any generator of $\mcH$ such that
  $\Gamma(0)=(f(\vx),\vx)$, clearly $\Gamma\in\cP(x)$ so that $
  \tau(\vx)\ge f(\vx)$. To show that this inequality has to be an
  equality, suppose for contradiction that $ \tau(\vx)>f(\vx)$, thus
  there exists a null future directed curve $\Gamma$ such that
  $t(\Gamma(0))> f(\vx)$ and $\Gamma(1)\in\mcH_\sigma \subset \mcH$.
  Then the curve $\tilde \Gamma$ obtained by following $\R\times
  \{\vx\}$ from $(f(\vx),\vx)$ to $(t(\Gamma(0)),\vx)$ and following
  $\Gamma$ from there on is a causal curve with endpoints on $\mcH$
  which is not a null geodesic. By Proposition~\ref{P14X11.1} the curve
  $\tilde \Gamma$ can be deformed to a timelike curve with the same
  endpoints,
  \ptcx{requires C2 or C11 cause of the deformation lemma}
  which is impossible by achronality of $\mcH$. \qed
\end{proof}

The Fermat principle, Proposition \ref{P1}, shows that $f$ is a
solution of the variational principle \eq{tau}. Now this variational
principle can be rewritten in a somewhat more convenient form as
follows: The identification of $\mcM$ with $\R\times\hyp$ by flowing
from $\hyp_0\equiv \hyp$ along the gradient of $t$ leads to a
global decomposition of the metric of the form
$$
g = \alpha(-dt^2 + h_t)\ , $$ where $h_t$ denotes a $t$--dependent
family of Riemannian metrics on $\hyp$. Any future directed
differentiable null curve $\Gamma(s)=(t(s),\vg(s))$ satisfies
$$
\frac{dt(
s)
}{ds} 
=\sqrt{h_{t(s)}(\dotg,\dotg)}\ , $$ where $\dotg$ is a shorthand for
$d\vg(s)/ds$. It follows that for any $\Gamma\in\cP(x)$ it holds
that
\begin{eqnarray*}
  t(\Gamma(a)) & = & t(\Gamma(b)) - \int _a ^b \frac{dt}{ds} ds \\ & =
  & 
  \sigma- \int _a ^b \sqrt{h_{\tilde \sigma}(\dotg,\dotg)} ds\ .
\end{eqnarray*}
This allows us to rewrite \eq{tau} as \be \label{taun} \tau(\vx) =
\sigma - \mu(\vx)\ , \qquad \mu(\vx)\equiv\inf_{\Gamma \in
  \cP(x)} \int _a ^b \sqrt{h_{t(s)}(\dotg,\dotg)} ds\ .  \ee We note that
in static space--times $\mu(\vx)$ is the Riemannian distance from
$\vx$ to $\mcH_\sigma$. In particular Equation \eq{taun} implies the
well known fact, that in globally hyperbolic static space--times
Cauchy horizons of open subsets of level sets of $t$ are graphs of
the distance function from the boundary of those sets.

Let $\vg:[a,b]\to \hyp$ be a piecewise differentiable path, for
any $p\in \R\times \{\vg(b)\}$ we can find a null future directed
curve $\hat\gamma :[a,b] \to \mcM$ of the form $\hat\gamma(s)
=(\phi(s),\vg(s))$ with future end point $p$ by solving the problem
\be \cases{ \phi(b)=t(p)\ , & \cr
  \displaystyle{\frac{d\phi(s)}{ds}}=\sqrt{h_{\phi(s)}(\dotg(s),\dotg(s))}\
  . & } \label{nulll}\ee
  %
  %
  %
  %
The path $\hat\gamma$ will be called the {\em null lift of $\gamma$
with endpoint $p$}.

As an example of application of Proposition \ref{P1} we recover the
following well known result~\cite{PenroseDiffTopo}:

\begin{Corollary}\label{C1} $f$ is Lipschitz continuous
  on any compact subset of its domain of definition.
\end{Corollary}

\begin{proof}
For $\vec y,\vec z\in B(2R)$ let $K\subset \R\times B(2R)$ be a
compact set which contains all the null lifts $\Gamma_{\vec y, \vz}$ of
  the coordinate segments $[\vec y,\vec z]:=\{\lambda \vec y
  +(1-\lambda)\vec z\ , \ \lambda \in [0,1]\}$ with endpoints
  $(\tau(\vec z),\vec z)$. Define \be \label{Cdef} \hat C = \sup
  \{\sqrt{h_p(n,n)}| p\in K, |n|_\delta =1\}\;, \ee where the supremum
  is taken over all points $p\in K$ and over all vectors $n\in T_p\mcM$
  the coordinate components $n^i$ of which have Euclidean length
  $|n|_\delta$ equal to one.  Choose $I$ to be a bounded interval
  large enough so that $K\subset I\times B(2R)$ and, as before, choose
  $\sigma$ large enough so that $I\times B(2R)$ lies to the past of
  $\hyp_\sigma$.  Let $\vec y,\vec z\in B(2R)$ and consider the
  causal curve $\Gamma=(t(s),\gamma(s))$ obtained by following the
  null lift $\Gamma_{{\vec y}, \vz}$ in the parameter interval $s\in[0,1]$,
  and then a generator of $\mcH$ from $(\tau(\vec z),\vec z)$ until
  $\mcH_\sigma$ in the parameter interval $s\in[1,2]$. Then we have
$$
  \mu(\vz)=\int_1^2 \sqrt{h_{\tilde \sigma}(\dotg,\dG)} ds
  \;.
$$
  Further
  $\Gamma\in \cP(x)$ so that
\begin{eqnarray}
  \mu({\vec y}) & \le & \int_0^2 \sqrt{h_{\tilde \sigma}(\dotg,\dotg)}
  ds\nonumber \\ & = & \int_0^1 \sqrt{h_{\tilde \sigma}(\dotg,\dotg)} ds +
  \int_1^2 \sqrt{h_{\tilde \sigma}(\dotg,\dotg)} ds \nonumber\\ & \le &
  \hat C |{\vec y}-\vz|_\delta + \mu(\vz) \label{tin} \ , \end{eqnarray}
where $|\cdot|_\delta$ denotes the Euclidean norm of a vector, and
with $\hat C$ defined in \eq{Cdef}. Setting 1) first ${\vec y} = \vx$,
$\vz = \vx + {\vec h}$ in \eq{tin} and 2) then $\vz = \vx$, ${\vec y} = \vx +
{\vec h}$, the Lipschitz continuity of $f$ on $B(2R)$ follows.  The
general result is obtained now by a standard covering argument. \qed
\end{proof}

Returning to the proof of Theorem \ref{semiConvT1}, for $\vx\in
B(R)$
  let $\Gamma_{\vx}$ 
  be a generator of $\mcH$ such that $\Gamma_{\vx}(0)=(\tau(\vx),\vx)$,
  and, if we write $\Gamma_{\vx}(s)= (\phi_{\vx}(s),\gamma_{\vx}(s))$,
  then we require that $\gamma_{\vx}(s)\in B(2R) $ for $s\in[0,1]$.
  For $s\in[0,1]$ and ${\vec h}\in B(R)$ let $\gamma_{\vx,\pm}(s)\in\hyp$
  be defined by
  $$\gamma_{\vx,\pm}(s) = \gamma_{\vx}(s)\pm (1-s){\vec h} = s
  \gamma_{\vx}(s)+(1-s)(\gamma_{\vx}(s)\pm {\vec h})\in B(2R)\ . $$
  We note
  that
  $$\gamma_{\vx,\pm}(0)=\vx \pm {\vec h}\ , \qquad \gamma_{\vx,\pm}(1)=
  \gamma_{\vx}(1)\ , \qquad \dot \gamma_{\vx,\pm} - \dot \gamma_{\vx}=
  \mp {\vec h}\ . $$
  Let $\Gamma_{\vx,\pm}=(\pxpm,\gamma_{\vx,\pm})$ be the
  null lifts of the paths $\gamma_{\vx,\pm}$ with endpoints
  $\Gamma_{\vx}(1)$. Let $K$ be a compact set containing all the
  $\Gamma_{\vx,\pm}$'s, where $\vx$ and ${\vec h}$ run through $ B(R)$. Let
  I be any bounded interval such that $I\times B(2R)$ contains $K$. As
  before, choose $\sigma$ so that $I\times B(2R)$ lies to the past of
  $\hyp_\sigma$, and let $b$ be such that $\Gamma_{\vx}(b)\in
  \mcH_\sigma$. (The value of the parameter $b$ will of course depend
  upon $\vx$). Let $\Gamma_\pm$ be the null curve obtained by
  following $\Gamma_{\vx,\pm}$ for parameter values $s\in [0,1]$, and
  then $ \Gamma_{\vx}$ for parameter values $s\in [1,b]$.  Then
  $\Gamma_\pm\in {\cP}(\vx\pm {\vec h})$ so that we have
  $$
  \mu(\vx\pm {\vec h})\le \int_0^1 \sqrt{h_\pxpm(\dot
    \gamma_{\vx,\pm},\dot \gamma_{\vx,\pm})} ds + \int_1^b
  \sqrt{h_\pxpm(\dot \gamma_{\vx},\dot \gamma_{\vx})} ds\ .
  $$
  Further
  $$\mu(\vx)= \int_0^b \sqrt{h_\pxhere(\dot \gamma_{\vx},\dot
    \gamma_{\vx})} ds\ , $$
  hence \begin{eqnarray} \lefteqn{
      \frac{\mu(\vx+ {\vec h})+ \mu(\vx- {\vec h})}{2} - \mu(\vx) \le }&&
    \nonumber \\ &&\int_0^1 \left(\frac{\sqrt{h_\pxpm(\dot
          \gamma_{\vx,+},\dot \gamma_{\vx,+})} + \sqrt{h_\pxpm(\dot
          \gamma_{\vx,-},\dot \gamma_{\vx,-})}}{2} -
      \sqrt{h_\pxhere(\dot \gamma_{\vx},\dot \gamma_{\vx})} \right)ds
    \; . \nonumber\\ && \label{concineq} \end{eqnarray} Since
  solutions of ODE's with parameters are differentiable functions of
  those, we can write
\begin{equation}
  \label{star}
  \phi _\pm(s)=\pxhere+\psi_i(s)h^i +r(s,h), \qquad |r(s,h)| \le C
|h|_\delta^2\ ,
\end{equation}
for some functions $\psi_i$, with a constant $C$ which is
independent of $\vx,{\vec h}\in B(R)$ and $s\in [0,1]$. Inserting
\eq{star} in \eq{concineq}, second order Taylor expanding the
function $\sqrt{h_{\phi _\pm(s)}(\dot \gamma_{\vx,\pm},\dot
  \gamma_{\vx,\pm})}(s)$ in all its arguments around
$(\pxhere,\gamma_{\vx}(s),\dot \gamma_{\vx}(s))$ and using
compactness of $K$ one obtains \be \label{concineq2} \frac{\mu(\vx+
{\vec h})+ \mu(\vx- {\vec h})}{2} - \mu(\vx) \le C |{\vec h}|_\delta^2 \ , \ee
for
some constant $C$. Set
$$
\psi(\vx)=\mu(\vx) - C |\vx|_\delta^2 \ . $$ Equation \eq{concineq2}
shows that
$$\forall \vx, {\vec h} \in B(R) \qquad \psi(\vx)\ge \frac{\psi(\vx+ {\vec h})+
  \psi(\vx- {\vec h})}{2} \ . $$
A standard argument implies that $\psi$ is concave. It follows that
$$f(\vx)+ C |\vx|_\delta^2=\tau(\vx)+ C
|\vx|_\delta^2=\sigma-\mu(\vx)+ C |\vx|_\delta^2=\sigma-\psi(\vx)$$
is convex, which is what had to be established. \qed
\end{proof}

\section{Cauchy surfaces}
 \label{S10X11.1}

A topological hypersurface $\hyp$ is said to be a \emph{Cauchy
surface} if
$$\mcDSJ=\mcM\;.
$$
(Note that it does not matter whether $\mcDSJ$ or $\mcDSI$ is chosen in the definition when the metric is $C^2$.)
 \ptcx{what about continuous?}
Theorem~\ref{Tdd1}, p.~\pageref{Tdd1}, shows that a necessary condition for this
equality is that $\mcM$ be globally hyperbolic. A celebrated
theorem, due independently to  Geroch and Seifert, shows that this
condition is also sufficient:

\begin{Theorem}[Geroch~\cite{GerochDoD}, Seifert~\cite{Seifert}]\label{TCS} A space-time $\Mgthreek$ is
globally hyperbolic if and only if there exists on $\mcM$ a time
function $\tau$ with the property that all its level sets are
Cauchy surfaces. The function $\tau$ can be chosen to be smooth if the
manifold is.
\end{Theorem}

\proof The proof uses volume functions, defined as follows: let
$\varphi_i$, $i \in \N$, be any partition of unity on $\mcM$, set
$$V_i:= \int_\mcM \varphi_i d\mu\;,$$
where $d\mu$ is, say, the Riemannian measure associated to the
auxiliary Riemannian metric $h$ on $\mcM$. Define
$$\nu:= \sum_{i\in \N} \frac1 {2^iV_i}\varphi_i \;.$$
Then $\nu$ is smooth, positive, nowhere vanishing, with
$$\int_M \nu\,d\mu=1\;.$$
Following Geroch, we define
$$V_{\pm}(p) := \int _{J^\pm(p)}\nu\,d\mu\;.$$
We clearly have
$$\forall p\in \mcM \quad 0 < V_\pm(p)< 1\;.$$
The functions $V_\pm$ may fail to be continuous in general,
an
example is given in Figure~\ref{Fmissing}.
\begin{figure}[tbh]
\begin{center}
\includegraphics[width=.5\textwidth]{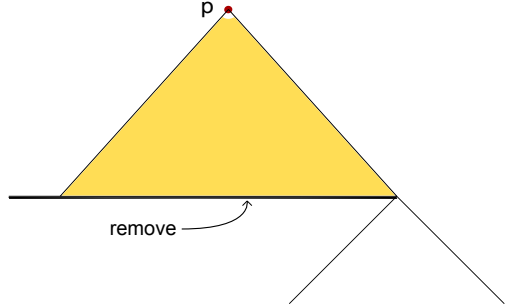}
\caption{The volume function $V_-$ is discontinuous at $p$. \label{Fmissing}}
\end{center}
\end{figure}
%
%
It turns out that such
behavior cannot occur under the current conditions:

\begin{Lemma}\label{LCS} On $C^2$ globally hyperbolic space-times the functions $V_\pm$ are continuous.
\end{Lemma}

\proof Let $p_i$ be any sequence converging to $p$,
and let the symbol $\varphi_\Omega$ denote the characteristic
function of a set $\Omega$. Let $q$ be any point such that $q\in
I^-(p)\Leftrightarrow p\in I^+(q)$, since $I^+(q)$ forms a
neighborhood of $p$
 we have $p_i\in I^+(q)\Leftrightarrow q\in
I^-(p_i)$ for $i$ large enough. Equivalently, \bel{eCS1}\forall
i\ge i_0 \qquad \varphi_{I^{-}(p_i)}(q)= 1=
\varphi_{I^-(p)}(q)\;.\ee
Since the right-hand-side of \eq{eCS1}
is zero for $q\not\in I^-(p)$ we obtain \bel{eCS2} \forall q
\qquad \liminf_{i\to\infty} \varphi_{I^{-}(p_i)}(q)\ge
\varphi_{I^-(p)}(q)
 \;.
\ee
By Corollary~\ref{Cgh1} $J^-(p)$ differs
from $I^-(p)$ by a topological hypersurface
so that
\bel{eCS3}  \liminf_{i\to\infty}
\varphi_{J^{-}(p_i)}\ge \varphi_{J^-(p)}\ \mbox{\ a.e.}\ee To
obtain the inverse inequality, let $q$ be such that
$$ \limsup_{i\to\infty} \varphi_{J^{-}(p_i)}(q)=1\;,$$
hence there exists a sequence $\gamma_j$ of future directed,
$\distb$-parameterised  causal curves from $q$ to $p_{i_j}$. By
Proposition~\ref{Pgh1} there exists a future directed accumulation
curve of the $\gamma_j$'s from $q$ to $p$. We have thus shown the
implication
$$ \limsup_{i\to\infty} \varphi_{J^{-}(p_i)}(q)=1 \quad \Longrightarrow \quad \varphi_{J^-(p)}(q)=1\;.$$
Since the function appearing at the left-hand-side of the
implication above can only take values zero or one, it follows
that \bel{eCS5} \limsup_{i\to\infty} \varphi_{J^{-}(p_i)}\le
\varphi_{J^-(p)}\;.\ee \Eqs{eCS1}{eCS5} show that
$$\lim_{i\to\infty} \varphi_{J^{-}(p_{i})} \ \mbox{ exists a.e., and equals }\
\varphi_{J^-(p)} \ \mbox{ a.e.}$$ Since
$$0\le \varphi_{J^-(p)}\le 1 \in {\mycal L}^1(\nu\, d\mu)\;,$$
the Lebesgue dominated convergence theorem gives
$$V_-(p)= \int_{\mcM} \varphi_{J^-(p)}\nu\,d\mu = \lim_{i\to\infty}
\int_{\mcM} \varphi_{J^-(p_i)}\nu\,d\mu =
\lim_{i\to\infty}V_-(p_i)\;. $$ Changing time orientation one also
obtains continuity of $V_+$.\qed

We continue with the following observation:

\begin{Lemma}\label{LCS2} $V_-$ tends to zero along any
past-inextendible causal curve $\gamma:[a,b)\to\mcM$.
\end{Lemma}

\proof Let $X_i$ be any partition of $\mcM$ by sets with compact
closure, the dominated convergence theorem shows that
\bel{eCS7}\lim_{k\to\infty}\sum_{i\ge k} \int_{X_i}\nu\,d\mu =
0\;.\ee Suppose that there exists $k<\infty$ such that
$$\forall s \quad J^-(\gamma(s))\cap \left(\cup_{i=1}^k X_i\right)\ne
\emptyset\;.$$ Equivalently, there exists a sequence $s_i\to b$
such that $$\gamma(s_i)\in K:= \overline{\cup_{i=1}^k X_i}\;.$$
Compactness of $K$ implies that there exists (passing to a
subsequence if necessary) a point $q_\infty\in K$ such that
$\gamma(s_i)\to q_\infty$. Strong causality of $\mcM$ implies that
there exists an elementary neighborhood $\mcO$ of $q_\infty$ such
that $\gamma\cap\mcO$ is connected, and Lemma~\ref{Lgh1.0} shows
that $\gamma\cap \mcO$ has finite $h$-length, which contradicts
inextendibility of $\gamma$ (compare Theorem~\ref{TP4a}). This
implies that for any $k$ we have
$$ J^-(\gamma(s))\cap \left(\cup_{i=1}^k X_i\right)=
\emptyset$$ for $s$ large enough, say $s\ge s_k$. In particular
$$s\ge s_k \quad \Longrightarrow \quad \int_{J^-(\gamma(s))\cap \left(\cup_{i=1}^k X_i\right)}\nu\,d\mu=0\;.$$
 This implies $$\forall s\ge s_k \quad V_-(\gamma(s))=
\int_{J^-(\gamma(s))\cap \left(\cup_{i=k+1}^\infty
X_i\right)}\nu\,d\mu\le \sum_{i\ge k+1} \int_{X_i}\nu\,d\mu\;,$$
which, in view of \eq{eCS7}, can be made as small as desired by
choosing $k$ sufficiently large. \qed

We are ready now to pass to the proof of Theorem~\ref{TCS}. Set
$$\tau:= \frac{V_-}{V_+}\;.$$
Then $\tau$ is continuous by Lemma~\ref{LCS}. Let $\gamma:(a,b)\to
\mcM$ be any inextendible future-directed causal curve. By
Lemma~\ref{LCS2}
$$\lim_{s\to b} \tau(\gamma(s))=\infty\;, \qquad\lim_{s\to a}
\tau(\gamma(s))=0\;. $$ Thus $\tau$ runs from $0$ to $\infty$ on
any such curves, in particular $\gamma$ intersects every level set
of $\tau$ at least once.  From the
definition of the measure $\nu \, d\mu$ it should be clear that
$\tau$ is actually strictly increasing on any causal curve, hence
the level sets of $\tau$ are met by causal curves precisely once.

The differentiability properties of $\tau$ constructed above are
not clear. It thus remains to show that   $\tau$ can be modified, if necessary, so that it is
as smooth as the atlas of $\mcM$ allows (except perhaps for analyticity).
 This can be done as follows
(compare~\cite{BernalSanchez,Seifert,BGP}): By \cite{NavarroMinguzzi}
 \ptcx{add crossreferencing}
there
exists a smooth metric $\hg$ with cones wider than those of $g$ so that
$(\mcM,\hg)$ is globally hyperbolic. We can thus apply the construction just
carried-out to construct a $\hg$--time function $\hat \tau$.
The property that the light-cones of $\hg$ are strictly wider than those of $g$
implies that the
$g$-gradient of $\hat \tau$ is everywhere timelike. A small smoothing of $
\hat
\tau$, using convolutions in local coordinates, leads to the desired smooth time
function. (Note that the smoothness of $\tau$ depends only upon the smoothness
of $\mcM$, regardless of the smoothness of the metric.)
\qed

An important corollary of Theorem~\ref{TCS} is:

\begin{Corollary}\label{CCS}
A globally hyperbolic space-time $\Mgthreek$ is necessarily diffeomorphic to
$\R\times\hyp$, with the coordinate along the $\R$ factor having
timelike gradient.
\end{Corollary}

\proof Let $X$ by any smooth
timelike vector field on $X$ --- if both the metric and
the time function $\tau$ of Theorem~\ref{TCS} are smooth then
$\nabla \tau$ will do, but any other choice works equally well.
Choose any number $\tau_0$ in the range of $\tau$. Define a
bijection $\varphi: \mcM\to \R\times \hyp$ as follows: for $p\in \mcM$ let
$q(p)$ be the point on the level set $\hyp_0=\{r\in \mcM:
\tau(r)=\tau_0\}$ which lies on the integral curve of $X$ through
$p$. Such a point exists because any inextendible timelike curve
in $\mcM$ meets $\hyp_0$; it is unique by achronality of $\hyp_0$.
The map $\varphi$ is continuous by continuous dependence of ODE's
upon initial values. If $\tau$ is merely continuous, one can
invoque the invariance of domain theorem
 \ptcx{give ref}
to prove that $\varphi$ is a homeomorphism; if $\tau$ is
differentiable, its level sets are differentiable,  $X$ meets those
level sets transversely, and the fact that $\varphi$ is a
diffeomorphism follows from the implicit function theorem. \qed

It is not easy  to decide whether or not a hypersurface $\hyp$ is
a Cauchy hypersurface, except in \emph{spatially compact}
space-times:

\begin{Theorem}[Budi\v{c} \emph{et al.}\/~\cite{BILY}, Galloway~\cite{Galloway:cauchy}]
\label{TBILY}
 Let $(\mcM,g)$ be a smooth globally hyperbolic space-time and suppose that
$\mcM$ contains a smooth, compact, connected spacelike hypersurface
$\hyp$. Then $\hyp$ is a Cauchy surface for $\mcM$.\ptcx{give proof}
\end{Theorem}
\begin{coco}
\begin{Remark} Some further results concerning
Cauchy surface criteria can be found
in~\cite{Galloway:cauchy,Harris}.
\end{Remark}
\end{coco}

\ptcx{add the new guys reference, with their Cauchy surface
criterion, and smoothing argument}

The following  shows the key role of global hyperbolicity for the wave equation:

\begin{Theorem}
\label{Tkeygh} Let $\hyp$ be a smooth spacelike hypersurface in a smooth
space-time $(\mcM,g)$. Then the Cauchy problem for the wave
equation has a unique globally defined solution on $\imcDSI$ for all smooth
initial data on $\hyp$.
\end{Theorem}

The theorem is well known to the community, but we note that an adequate reference does not seem to be available.\restrict{
 \ptc{use Hahn Banach a la Hoermander?} A sketch of the proof proceeds as follows:}

It is of interest to enquire what happens with solutions of the wave equation when Cauchy horizons occur:

First, examples are known where solutions of wave equations blow up at the event horizon, cf., e.g., \cite{ChIM,KichenassamyRendall}.
 \ptcx{discuss and/or crossrefer to the examples section}

Next, a simple example where solutions extend smoothly to solutions of the wave equation, but uniqueness fails, proceeds as follows: let $\hyp$ be the unit ball within the hypersurface $\{t=0\}$ in Minkowski space-time. Let $p=(1,\vec 0)$ and let $q=(-1,\vec 0)$, then the Cauchy horizon is the union of two inverted cones with tips at $p$ and $q$:
$$
 \mcHSI= (\dot J^-(p)\cap \{t>0\})\cup \dot J^+(q) \cap \{t<0\})
 \;.
$$
Any two distinct solutions of the wave equation on Minkowski space-time which have the same Cauchy data on $\hyp$  coincide on $\mcDSI$, and provide examples of solutions which differ beyond the event horizon.

To conclude, we have both global existence and uniqueness of solutions of the Cauchy problem for the wave equation in globally hyperbolic spacetimes. On the other hand, uniqueness or existence are problematic when the space-time is not globally hyperbolic space-times.

\section{Some applications}
Any formalism is only useful to something if it leads to
interesting applications. In this section we will list some of
those.

We start by pointing out the already-mentioned Theorem~\ref{Tkeygh}. Its counterpart for the Einstein equations is the celebrated Choquet-Bruhat -- Geroch theorem:
 \ptcx{should probably be removed and replaced by a crossreference; but note that this one is smooth
 while \ref{T11XI11.1} is for Sobolev}

\begin{Theorem}[Choquet-Bruhat, Geroch \cite{ChoquetBruhatGeroch69}]
 \label{T11XI11.1x}
Consider  a smooth triple
$(\hyp ,\gamma,K)$,
where $\hyp$ is an $n$-dimensional manifold, $\gamma  $ is a
Riemannian metric on $\hyp $, and $K  $ is a symmetric two--tensor on
$\hyp $, satisfying the general relativistic vacuum constraint equations.
Then
there exists a {\em unique up to isometries} vacuum
space--time $(M,g)$, called the {\em maximal globally hyperbolic vacuum development of
$(\hyp ,\gamma,K)$}, with an embedding $i:\hyp \rightarrow M$ such
that
$i^*g=\gamma$, and such that $K$ corresponds to the extrinsic curvature tensor (``second fundamental form")
of
$i(\hyp )$ in $M$. $(M,g)$ is {\em inextendible} in the class of
globally hyperbolic space--times with a vacuum metric.
\end{Theorem}

This theorem is the starting point of many studies in mathematical general relativity.
Similarly to the wave equation, examples show that uniqueness fails beyond horizons.
 \ptcx{expand eventually}

To continue, we say that $(\mcM,g)$ satisfies the \emph{timelike focussing
condition}, or \emph{timelike convergence condition}, if the Ricci
tensor satisfies
\bel{pec1} R_{\mu\nu}n^\mu n^\nu \ge 0 \ \mbox{ for all
\underline{timelike} vectors } \ n^\mu\;.\ee By continuity, the
inequality in \eq{pec1} will also hold for  causal vectors.
Condition \eq{pec1} can of course be rewritten as a condition on
the matter fields using the Einstein equation, and is satisfied in
many cases of interest, including vacuum general relativity, or
the Einstein-Maxwell theory, or the Einstein-Yang-Mills theory.
This last two examples actually have the property that the
corresponding energy-momentum tensor is trace-free; whenever this
happens, \eq{pec1} is simply the requirement that the energy
density of the matter fields is non-negative for all
observers:\bel{tfcc}8\pi T_{\mu\nu} n^\mu n^\nu =
(R_{\mu\nu}-\frac 12 \underbrace{R}_{=0\ \mbox{\scriptsize if}\
g^{\alpha\beta}T_{\alpha\beta}=0} g_{\mu\nu})n^\mu n^\nu =
R_{\mu\nu}n^\mu n^\nu \;.\ee

We say that $(\mcM,g)$ satisfies the \emph{null energy condition}
if
\bel{nec1} R_{\mu\nu}n^\mu n^\nu \ge 0 \ \mbox{ for all
\underline{null}
vectors } \ n^\mu\;.\ee Clearly, the timelike focussing condition
implies the null energy condition. Because $g_{\mu\nu}n^\mu
n^\nu=0 $ for all null vectors, the $R$ term in the calculation
\eq{tfcc} drops out regardless of whether or not $T_{\mu\nu}$ is
traceless, so the null energy condition is equivalent to
positivity of energy density of matter fields without any
provisos.

The simplest \emph{geodesic incompleteness theorem}
is:\ptcx{Hawking has a claim on this theorem too? give proofs?}

\begin{Theorem}[Geroch's geodesic incompleteness
theorem~\cite{Geroch:singularity}] Let $(\mcM,g)$ be a smooth globally
hyperbolic satisfying the timelike focussing condition, and
suppose that $\mcM$ contains a compact
 \ptcx{is this needed? remove smoothness if giving proof}
Cauchy surface $\hyp$ with strictly
negative mean curvature:
$$\trh K <0\;,$$
where $(h,K)$ are the usual Cauchy data induced on $\hyp$ by $g$.
Then $(\mcM,g)$ is future timelike geodesically incomplete.
\end{Theorem}

Yet another incompleteness theorem involves \emph{trapped surfaces},
this requires introducing
some terminology:
 \ptcx{should be done in terms of null geometry, and what follows should
 be a translation, after the theorem}
Let $\hyp$ be a spacelike hypersurface in $(\mcM,g)$, and
consider a surface $S\subset \hyp$. We shall also assume that $S$
is two-sided in $\hyp$, this means that there exists a globally
defined field $m$ of unit normals to $S$ within $\hyp$. There are
actually two such fields, $m$ and $-m$, we arbitrarily choose one
and call it {\em outer pointing}. In situations where $S$ does
actually bound a compact region, the outer-pointing one should of
course be chosen to point away from the compact region. We let $H$
denote the mean extrinsic curvature of $S$ within $\hyp$:
\bel{meancur} H:= D_i m^i\;,\ee
where $D$ is the covariant extrinsic of the metric $h$ induced on
$\hyp$. We say that $S$ is \emph{outer-future-trapped}
if
\bel{outertr} \theta_+:=H +K_{ij}(g^{ij}-m^im^j)\le 0\;,\ee
with an obvious symmetric definition for
\emph{inner-future-trapped}:
\bel{innertr} \theta_-:=-H +K_{ij}(g^{ij}-m^im^j)\ge 0\;,\ee
(One also has the obvious \emph{past} version thereof, where the
sign in front of the $K$ term should be changed.) A celebrated
theorem of Penrose\footnote{Penrose's theorem is slightly more
general, using a definition of $\theta_\pm$ which involves
a discussion of null geometry which we prefer to avoid here. This
is the reason why we have stated this theorem in the current
form.} asserts that:

\begin{Theorem}[Penrose's geodesic incompleteness  theorem~\cite{Penrose:singularity}]
Let $(\mcM,g)$ be a smooth
 \ptcx{remove smoothness if giving proof}
globally hyperbolic space-time satisfying the
\emph{null energy condition}, and suppose that $\mcM$ contains a
\emph{non-compact} Cauchy surface $\hyp$. If there exists a
\emph{compact} trapped surface within $S$ which is both
\emph{inner-future-trapped} and \emph{outer-future-trapped}, then
$(\mcM,g)$ is geodesically incomplete.
\end{Theorem}

The significance of this theorem stems from the fact, that the Schwarzschild
solution, as well as the non-degenerate Kerr black holes, possess trapped
surfaces beyond the horizon. A small perturbation of the metric will preserve
this. It follows that the geodesic incompleteness of these black holes is not
an accident of the large isometry group involved, but is stable under
perturbations of the metric.

More generally, future-trapped surfaces signal the existence of black holes.
Formal statements to this effect require  the introduction of the notion of a
black hole, as well as several global regularity conditions, and
will therefore not be given here.\ptcx{should be done in the black
holes section, give crossreference}

We close the list of applications with the \emph{area theorem}, which plays a
fundamental role in black-hole thermodynamics:

\begin{Theorem}[\cite{HE,ChDGH}]
\MCinfty
 Let $\mcE$ be a future geodesically complete
acausal null hypersurface, and let $\hyp_1$, $\hyp_2$ be two
spacelike acausal hypersurfaces. If
$$\mcE\cap \hyp_1 \subset J^-(\mcE\cap \hyp_2)\;,$$
then $$ \mbox{Area}(\mcE\cap \hyp_1) \le \mbox{Area}(\mcE \cap
\hyp_2)\;.$$
\end{Theorem}

\ptcx{show existence of maximising geodesics in globally hyperbolic
space-times}

\ptcx{a sections on lorentzian length, restricted,   should probably be ignored at the moment}
\restrict
{

\section{The Lorentzian distance function}
 \label{SLlf2}
 Let $ \Omega(p,q)$ be the set of all  future
directed causal curves from $p$ to $q$. We define the
\emph{Lorentzian distance function} $d : M \times M \to [0,
\infty]$ as follows:
\bel{lorddef} d(p,q) = \left\{%
\begin{array}{ll}
    \sup \{ \ell(\gamma): \gamma \in \Omega(p,q) \} , & \hbox{$q\in J^+(p)$,} \\
    0, & \hbox{otherwise,} \\
\end{array}%
\right.    \ee
 where  $\sup$ is understood
in $\R\cup\{\infty\}$. We shall sometimes write $d_g$ for $d$ when
the need to indicate explicitly the metric arises.

\begin{coco}
It is tempting to define $d(p,q)=-\infty$ for $q\not \in J^+(p)$,
but this leads to a function $d$ which is never continuous
(compare Proposition~\ref{Pro:contlord}), which is the reason for
using \eq{lorddef}.
\end{coco}

It is legitimate to ask the question, what would happen if one
took the infimum rather than the supremum  in \eq{lorddef}. The
result is not very interesting, the reader should easily convince
himself that one can approach a $C^0$ curve $\gamma$ as close as
desired by threading back and forth near $\gamma$ along null
geodesics, each of which has zero Lorentzian length. So taking the
infimum in \eq{lorddef} always gives zero.

In any case the calculation of $\sigma$ in Minkowski space-time,
which we are about to do, should make it clear that the right
thing to do is to take the supremum:

\begin{Example}
Let $q=(x^\mu)=(x^0,\vec x)\in \R^{1,n}$ be in the timelike future
of the origin, thus $x^0>|\vec x|_\delta$ by
Proposition~\ref{PM1}. Let us make a Gram-Schmidt
orthonormalization starting from $(x^\mu)$, viewed as a vector
tangent to $\R^{1,n}$ at $0$, thus $x^\mu e_\mu =
\sqrt{(x^0)^2-|\vec x|_\delta^2}e_0$. Using normal coordinates
based on the new basis $e_\mu$ one obtains a Minkowskian
coordinate system in which $q=(t,\vec 0)$, with
\bel{tdef}t=\sqrt{(x^0)^2-|\vec x|_\delta^2}\;.\ee Consider, now
any causal curve $\gamma(s) = (\gamma^0(s),\vec \gamma(s))$, from
the origin to $q$, causality gives
$$|\dot \gamma^0|\ge |\dot {\vec \gamma}|_\delta\;,$$ with $|\dot
\gamma^0|$ without zeros (since the inequality is strict for
timelike vectors, and if it is an equality than neither side can
vanish, otherwise $\dot \gamma$ would vanish, which is not allowed
for causal vectors). This shows that we can reparameterize
$\gamma$ so that
$$\gamma^0(s)=s\;,\quad s\in [0,t]\;.$$ In this parameterization
we have
$$\ell(\gamma)= \int_0^t \sqrt{1- |\dot {\vec \gamma}|^2_\delta}\;,$$
and this formula makes it clear that the supremum is attained on
the path
$$\gamma(s)=(s,\vec 0) \quad \Longrightarrow \quad \ell (\gamma) = t \;.$$
It follows that
\bel{lordMin} d_\eta(p,q) = \left\{%
\begin{array}{ll}
    \sqrt{(x^0)^2-\sum_i(x^i)^2 }\;, & \hbox{$q\in J^+(p)$,} \\
    0, & \hbox{otherwise.} \\
\end{array}%
\right.    \ee We note that \eq{lordMin} follows also from
Proposition~\ref{Pelmax} below.
\end{Example}

\begin{Example}
 The two-dimensional space-time $\mcM=S^1\times S^1$ with the
flat metric $-dt^2+dx^2$, where $t$ is a $\mod 2 \pi$--coordinate
on the first $S^1$ factor, and $x$ is a similar coordinate on the
second factor, provides an example where $d(p,q)=\infty$ for all
$p,q\in \mcM$. This is seen by noting that if $\gamma$ is a
timelike curve from $p=(t_0,x_0)$ to $q$, then one can obtain a
causal curve of length $2\pi n+\ell(\gamma)$ be going $n$ times
around the timelike circle $x=x_0$, and then following $\gamma$
from $p$ to $q$.
\end{Example}

 It follows immediately from its definition
that the Lorentzian distance function obeys the reverse triangle
inequality: for $p \in J^-(q)$ and $q\in J^-(r)$ it holds that
\bel{revtrin} d(p,r) \geq d(p,q) + d(q,r) \, .
\ee Indeed, if $p \in J^-(q)$ and $q\in J^-(r)$, then the class
$\Omega(p,r)$ of future directed causal curves from $p$ to $r$
contains the class $\Omega(p,q)\cup \Omega(q,r)$, where the union
$\cup$ is understood as the concatenation operation on paths. It
follows that the sup of $\ell$ over $\Omega(p,q)$ is greater than
or equal to the sup of $\ell$ over $\Omega(p,q)\cup \Omega(p,r)$,
which implies \eq{revtrin}

The function $d$  needs not to be continuous in general, as seen
in the following example:
\begin{Example}[\protect{\protect\cite[p.~141]{Beem-Ehrlich:Lorentz2}}]
\label{Exnotcontlord} Let $\mcM$ be $$\{(t,x):0\le t\le
2\}\setminus \{(1,x):-1\le x\le 1\}\;,$$
 with the identifications
$(x,0)\sim(x,2)$\restrict{ (see Figure~\ref{Fnotcontlord})\ptcx{make figure}},
equipped with the flat metric $g=-dt^2+dx^2$. Let $p$ be the
origin and let $q=(0,1/2)$, set $p_n = (0,1/n)\to p$. Clearly
$p_n\in I^+(p_n)$ which implies $d(p_n,p_n)=\infty$, and also
$d(p_n,q)=\infty$ for $n>2$. By Proposition~\ref{Pelmax}    we
have $d(p,q)=1/2$, so that the function $d(\cdot,q)$ is not
continuous at $p$.
\end{Example}
\begin{figure}
\caption{\label{Fnotcontlord} The space-time of
Example~\ref{Exnotcontlord}.}
\end{figure}

Let us show that  $d$ is \emph{lower semi-continuous}:

\begin{Proposition}
\label{Pro:lowsemcontlord} If $p_n\to p$ and $q_n\to q$ then
\bel{lowsemcontlord} d(p,q)\le \lim\inf d(p_n,q_n)\;.\ee (This implies in particular that $d$ is continuous
on $d^{-1}(\{\infty\})$.)
\end{Proposition}

\proof If $d(p,q)=0$ there is nothing to prove. Otherwise let
$\gamma:[0,1]\to\mcM$ be any causal curve from $p$ to $q$ with
non-zero Lorentzian length. Let us start by showing that we can
always deform $\gamma$ to a new causal curve $\tilde \gamma$,
which is \emph{timelike near $p$ and $q$},  reducing the length by
less than $\epsilon/2$ (if at all). Indeed, let $s_*>0$ be the
last point on $\gamma$ such that $\gamma(s)\not \in I^+(p)$, then
$s<1$, and by Proposition~\ref{Pro:acausalgeo} $\gamma|_{[0,s_*]}$
is a null geodesic, in particular $\ell(\gamma|_{[0,s_*]})=0$. By
definition of $s_*$ there exists a sequence $s_i\searrow s_*$ such
that $\gamma(s_i)\in I^+(\gamma(s_*))$. Since $$\ell(\gamma)=
\ell(\gamma|_{[s_*,1]})\;,$$  by choosing $i$ large enough we will
have $\ell(\gamma|_{[s_i,1]})\ge \ell(\gamma)-\epsilon/4$. One can
use the path  constructed in the proof of Lemma~\ref{Lpushup0} to
replace $\gamma|_{[0,s_i]}$ by a timelike curve from $p$ to
$\gamma(s_i)$, leading to a new curve from $p$ to $q$ which is
timelike near $p$ and which has length not less than
$\ell(\gamma)-\epsilon/4$. A similar construction near $q$ leads
to a causal path $\tilde \gamma:[0,1]\to \mcM$ from $p$ to $q$,
timelike near $p$ and $q$, such that
$$\ell(\tilde \gamma)\ge \ell(\gamma)-\epsilon/2\;.$$

Let, now,  $s_->0$ be small enough so that $$\ell(\tilde
\gamma|_{[0,s_-]})<\epsilon/4\;.$$ Since $\tilde \gamma$ is
timelike near $p$ the set $I^-(\gamma(s_-))$ is an open
neighborhood of $p$ and therefore, for $n$ large enough, there
exists a timelike curve from $p_n$ to $\gamma(s_-)$. Similarly we
can find $s_+$ close to $q$ so that $$\ell(\tilde
\gamma|_{[s_+,1]})<\epsilon/4\;,$$ and there exists a timelike
curve from $\gamma(s_+)$ to $q_n$ for $n$ large enough.
Concatenating those curves results in a causal curve from $p_n$ to
$q_n$ with length not less than $\ell(\gamma)-\epsilon$. Taking a
supremum over $\gamma$'s proves \eq{lowsemcontlord} except for
$-\epsilon$ at the left-hand-side. However, as $\epsilon$ is
arbitrary, \eq{lowsemcontlord} follows.

If $d(p,q)=\infty$, then \eq{lowsemcontlord} gives $\liminf
d(p_n,q_n)=\infty$, which clearly implies $\lim
d(p_n,q_n)=\infty$, as desired.\qed

We shall say that a future directed causal path $\gamma:I\to \mcM$
is \emph{maximising} if \bel{maxcp}\forall \ s,s'\in I\;,\ s<s'
\qquad d(\gamma(s),\gamma(s'))=\ell(\gamma|_{[s,s']})\;.\ee We
have the following simple observation:

\begin{Proposition}
\label{Pmaxgeo} $d(p,q)$ is attained on $\gamma:[0,1]\to \mcM$ if
and only if  $\gamma$ is maximising.
\end{Proposition}

\proof  The triangle inequality gives \bel{trineqmax} d(p,q)\ge
d(p,\gamma(s_1))+d(\gamma(s_1),\gamma(s_2))+d(\gamma(s_2),q)\;.\ee
Since $d(p,q)$ is attained on $\gamma$ we have
\beaa d(p,q)&=&\int_0^1 \sqrt{g(\dot \gamma,\dot \gamma)}(s)ds
\\ & = & \int_0^{s_1} \sqrt{g(\dot \gamma,\dot \gamma)}(s)ds+
\underbrace{\int_{s_1}^{s_2} \sqrt{g(\dot \gamma,\dot
\gamma)}(s)ds}_{=\ell(\gamma|_{[s_1,s_2]})}+ \int_{s_2}^1
\sqrt{g(\dot \gamma,\dot \gamma)}(s)ds\;.\eeaa If $\gamma$ were
not maximising, then there would exist $s_1,s_2\in [0,1]$ and a
causal curve $\sigma$ from $\gamma(s_1)$ to $\gamma(s_2)$ with
Lorentzian length larger than $\ell(\gamma|_{[s_1,s_2]})$. One
would then obtain a causal path longer than $\gamma$ by following
$\gamma$ from $p$ to $\gamma(s_1)$, then following $\sigma$, and
then following $\gamma$ from $\gamma(s_2)$ to $q$. This
contradicts the fact that $d(p,q)$ is attained on  $\gamma$. It
follows that
$\ell(\gamma|_{[s_1,s_2]})=d(\gamma(s_1),\gamma(s_2))$ for all
$s_1,s_2\in [0,1]$. \qed

We also note:

\begin{Proposition}
\label{Pmaxgeo2} A null geodesic $\gamma$ is maximising if and
only if it is achronal.
\end{Proposition}

\proof $\Longrightarrow$: From the definition of $d$ we have
$d(\gamma(s_1),\gamma(s_2))>0$ whenever $\gamma(s_2)\in
I^+(\gamma(s_1))$, therefore all the terms at the right-hand-side
of \eq{trineqmax} are zero if $d(p,q)$ vanishes. It follows that
$\gamma$ is achronal, and the fact that $\gamma$ is a null
geodesic follows from Proposition~\ref{Pro:acausalgeo}.

 $\Longleftarrow$: Since $\gamma$ is achronal the inequality in
 \eq{trineqmax} is an inequality, with all terms vanishing.
 \qed

A Jacobi field along a geodesic $\gamma$ is a solution of
\emph{Jacobi equation}:
\bel{Jacobi} \frac {D^2Z}{ds^2}(s)
= R(\dot \gamma, Z)\dot \gamma\;.\ee
The point $\gamma(s_2)$ is said to be \emph{conjugate} to
$\gamma(s_1)$  if there exists a non-trivial solution of
\eq{Jacobi} which vanishes at both points.

A key result concerning maximising curves is:

\begin{Theorem}
 \label{Tmaxgeo3}
\MCtwok
If $\gamma$ is maximising, then $\gamma$ is a
geodesic without conjugate points. \ptc{the proof of the theorem
to be finished}
\end{Theorem}

\proof  Proposition~\ref{Pmaxgeo} reduces the problem to showing
that if $d(p,q)$ is attained on $\gamma$, then $\gamma$ is a
geodesic without conjugate points. This is proved by variational
arguments, as follows. We start with a Lemma:

\begin{Lemma}
\label{Ldiff} Let $q\in I^+(p)$, and let $\Omega_{C^{1,1}}(p,q)$
denote the class of $C^{1,1}$ timelike paths from $p$ to $q$. Then
$$\sup_{\gamma \in \Omega(p,q)}\ell(\gamma)=\sup_{\gamma \in
\Omega_{C^{1,1}}(p,q)}\ell(\gamma)\;.$$
\end{Lemma}

 \qed

It follows from Proposition~\ref{Pncj0} in Appendix~\ref{Sstenc}
that the exponential map ceases to be a local diffeomorphism at
conjugate points. This implies, by definition, that there are no
points conjugate to the origin in elementary neighborhoods,
leading to:\ptcr{this does not need the conjugate points story and
can be proved using Gauss coordinates on a normal neighborhood so
it could be used in the proof of the previous theorem; note that there is a
serious problem with crossreferencing, as this label does not exist}

\begin{Proposition}\label{Pelmax} Let $\mcO$ be an elementary neighborhood
centred at $p$. For $q\in \mcO$ the geodesic segment from $p$ to
$q$ is the longest causal curve from $p$ to $q$ entirely contained
in $\mcO$. Equivalently, the distance function $\sigma(p,\cdot)$
within the space-time $(\mcO,g|_{\mcO})$ coincides with
$\sqrt{-\sigma_p}$ on $J^+(p,\mcO)$, and vanishes elsewhere. \qed
\end{Proposition}

Let us show that accumulation curves of maximising paths are
maximising:

\begin{Proposition}
\label{Pmaxaccum} Let $\gamma_n$ be a sequence of maximising
curves accumulating at $\gamma$, then $\gamma$ is a maximising
geodesic.
 \ptcr{this is already done in a previous section, so crossrefer and not use}
\end{Proposition}

\proof By Theorem~\ref{Tmaxgeo3} the $\gamma_n$'s are geodesics,
so that $\gamma$ is a geodesic by Proposition~\ref{Pro:inexgeo}.
Reparameterizing if necessary, we can assume that the $\gamma_n$'s
are affinely parameterised and defined on a common compact
interval $I$. Passing to a subsequence if necessary, the
$\gamma_n$'s converge to $\gamma$ in $C^1(I,\mcM)$, so that
$\ell(\gamma_n)\to\ell(\gamma)$. Suppose that $\gamma$ is not
maximising between $p=\lim\gamma_n(s_1)$ and $q=\lim
\gamma_n(s_2)$, then there exists a causal curve $\tilde \gamma$
from $p$ to $q$ with $$r:=\ell(\tilde
\gamma)-\ell(\gamma|_{[s_1,s_2]})>0\;,$$ in particular $q\in
I^+(p)$. Since  $\ell(\gamma_n)\to\ell(\gamma)$ we have
$\ell(\gamma_n)<\ell(\gamma)-r/2$ for $n$ large enough. However,
the construction in the proof of
Proposition~\ref{Pro:lowsemcontlord} provides, for $n$ large
enough, a causal curve from $\gamma_n(s_1)$ to $\gamma_n(s_2)$
strictly longer than $\ell(\gamma)-r/2$ if we choose $\epsilon$
there to be smaller than $r/2$, contradicting the maximising
property of $\gamma_n$. \qed

We are almost ready to show that $d$ is continuous on globally
hyperbolic space-times. We start with the proof of a slightly
stronger version of Theorem~\ref{Tgh1}:
 \ptcr{note that the wording near that theorem needs changing if not done}

\begin{Theorem}\label{Tgh1n}
Let $(\mcM,g)$ be globally hyperbolic, if $q\in I^+(p)$,
respectively $q\in J^+(p)$, then there exists a timelike,
respectively causal, future directed, maximising geodesic from $p$
to $q$, on which $d(p,q)$ is attained.
\end{Theorem}

\proof Let $\gamma_n$ be any sequence of $\distb$-parameterized
causal paths from $p$ to $q$ such that $\ell(\gamma_n)\to
\dist_g(p,q)$. By Proposition~\ref{Pgh1} there exists a causal curve
$\gamma$ through $p$ and $q$ which is an accumulation curve of the
$\gamma_n$'s.

We want to  show that $\gamma$ has to be a geodesic --- this uses
the fact that $\gamma$ is the longest causal curve from $p$ to $q$.
If $q\in I^+(p)$ then $\dist_g(p,q)>0$ and the resulting geodesic
has to be timelike.\ptc{the details of this should be filled in;
perhaps choose Gauss coordinates associated with $x^0=0$, then the
local inequality is more obvious then the normal coordinates?}

Assume, first, that $\gamma$ is timelike, suppose that there exist
two points $\gamma(s_1)$ and $\gamma(s_2)$  such that
$\dist_g(\gamma(s_1),\gamma(s_2))>0$ and such that $\gamma(s_2)$
lies within the domain of definition of normal coordinates centred
on $\gamma(s_1)$. Without loss of generality we can assume that
$\gamma(s_2)=(\lambda,\vec 0)$. Then
 \ptc{the idea is to reduce this to a Minkowskian problem by the
 triangle inequality, but I do not quite see it now; to be finished}

\qed

\begin{Proposition}
\label{Pro:contlord} On globally hyperbolic space-times $d$ is
continuous, everywhere finite.
\end{Proposition}

\proof Let $q_n\in I^+(p_n)$ and suppose that  $(p_n,q_n)$
converge to $(p,q)$. Let $\gamma_n$ be any maximising geodesic
from $p_n$ to $q_n$, then (passing to a subsequence if necessary),
by Proposition~\ref{Pmaxaccum}, $\gamma_n$ converges to a
maximising geodesic from $p$ to $q$ in $C^2$ topology. This
implies that $\ell(\gamma_n)\to \ell(\gamma)$, proving the result.
\qed

 \ptc{greg would like to see a proof of upper semi-continuity of the
 lorentzian arc length functional, and note that a beginning is already in the file above}

\ptcr{file NullCords, in principle complete but irrelevant without
further text}

%
%

\ptcr{here could come the file Completeness? about the Choquet-Bruhat Geroch
theorem}
 \ptcr{Conformal Completions file as an input, duplications? but it is in the
original levoca directory, not in James Grant, needs inputting eventually but
irrelevant for the moment}

}

{\noindent\sc Acknowledgements} The author is grateful to Gregory Galloway and James Grant for useful discussions.

\backmatter
\bibliographystyle{amsplain}
\bibliography{../references/hip_bib,%
../references/reffile,%
../references/newbiblio,%
../references/newbiblio2,%
../references/bibl,%
../references/howard,%
../references/bartnik,%
../references/myGR,%
../references/newbib,%
../references/Energy,%
../references/dp-BAMS,%
../references/prop2,%
../references/besse2,%
../references/netbiblio}

\def\polhk#1{\setbox0=\hbox{#1}{\ooalign{\hidewidth
  \lower1.5ex\hbox{`}\hidewidth\crcr\unhbox0}}} \def\cprime{$'$}
  \def\cprime{$'$}
\providecommand{\bysame}{\leavevmode\hbox to3em{\hrulefill}\thinspace}
\providecommand{\MR}{\relax\ifhmode\unskip\space\fi MR }
\providecommand{\MRhref}[2]{%
  \href{http://www.ams.org/mathscinet-getitem?mr=#1}{#2}
}
\providecommand{\href}[2]{#2}
\begin{thebibliography}{10}

\bibitem{AP}
F.~Antonacci and P.~Piccione, \emph{A {F}ermat principle on {L}orentzian
  manifolds and applications}, Appl.\ Math.\ Lett. \textbf{9} (1996), 91--95.

\bibitem{BahouriChemin}
H.~Bahouri and J.-Y. Chemin, \emph{{\'E}quations d'ondes quasiline\'aires et
  estimations de {S}trichartz}, Am.\ Jour.\ Math. \textbf{121} (1999),
  1337--1377.

\bibitem{BGP}
C.~B{\"a}r, N.~Ginoux, and F.~Pf{\"a}ffle, \emph{Wave equations on {L}orentzian
  manifolds and quantization}, ESI Lectures in Mathematics and Physics,
  European Mathematical Society (EMS), Z\"urich, 2007. \MR{MR2298021}

\bibitem{Beem-Ehrlich:Lorentz2}
J.~K. Beem, P.~E. Ehrlich, and K.~L. Easley, \emph{Global {Lorentzian}
  geometry}, 2 ed., Pure and Applied Mathematics, vol. 202, Marcel Dekker, New
  York, 1996.

\bibitem{BernalSanchez}
A.N. Bernal and M.~S{\'a}nchez, \emph{Smoothness of time functions and the
  metric splitting of globally hyperbolic space-times}, Commun.\ Math.\ Phys.
  \textbf{257} (2005), 43--50. \MR{MR2163568 (2006g:53105)}

\bibitem{BILY}
R.~Budic, J.~Isenberg, L.~Lindblom, and P.~Yasskin, \emph{On the determination
  of the {C}auchy surfaces from intrinsic properties}, Commun.\ Math.\ Phys.
  \textbf{61} (1978), 87--95.

\bibitem{YCB:GRbook}
Y.~Choquet-Bruhat, \emph{General relativity and the {E}instein equations},
  Oxford Mathematical Monographs, Oxford University Press, Oxford, 2009.
  \MR{MR2473363 (2010f:83001)}

\bibitem{ChoquetBruhatGeroch69}
Y.~Choquet-Bruhat and R.~Geroch, \emph{Global aspects of the {C}auchy problem
  in general relativity}, Commun.\ Math.\ Phys. \textbf{14} (1969), 329--335.
  \MR{MR0250640 (40 \#3872)}

\bibitem{SCC}
P.T. Chru\'{s}ciel, \emph{On uniqueness in the large of solutions of {E}instein
  equations (``{S}trong {C}osmic {C}ensorship'')}, Australian National
  University Press, Canberra, 1991.

\bibitem{ChDGH}
P.T. Chru\'{s}ciel, E.~Delay, G.~Galloway, and R.~Howard, \emph{Regularity of
  horizons and the area theorem}, Annales Henri Poincar\'e \textbf{2} (2001),
  109--178, arXiv:gr-qc/0001003. \MR{MR1823836 (2002e:83045)}

\bibitem{ChGrant}
P.T. Chru\'{s}ciel and J.~Grant, \emph{{On Lorentzian causality with continuous
  metrics}},  (2011), arXiv:1110.xxxx [gr-qc].

\bibitem{ChIM}
P.T. Chru\'{s}ciel, J.~Isenberg, and V.~Moncrief, \emph{Strong cosmic
  censorship in polarized {G}owdy space--times}, Class.\ Quantum Grav.
  \textbf{7} (1990), 1671--1680.

\bibitem{Clarke:optimization}
F.H. Clarke, \emph{Optimization and nonsmooth analysis}, second ed., Society
  for Industrial and Applied Mathematics (SIAM), Philadelphia, PA, 1990.

\bibitem{EvansGariepy}
L.C. Evans and R.F. Gariepy, \emph{Measure theory and fine properties of
  functions}, CRC Press, Boca Raton, FL, 1992.

\bibitem{Galloway:cauchy}
G.J. Galloway, \emph{{Some results on Cauchy surface criteria in Lorentzian
  geometry}}, Illinois Jour.\ Math. \textbf{29} (1985), 1--10.

\bibitem{Galloway:globasp}
Gregory~J. Galloway, \emph{{Some global aspects of compact space-times}},
  Arch.\ Math. \textbf{42} (1984), 168--172.

\bibitem{Geroch:singularity}
R.~Geroch, \emph{{Singularities in closed universes}}, Phys.\ Rev.\ Lett.
  \textbf{17} (1966), 445--447.

\bibitem{Geroch:topology}
\bysame, \emph{Topology in general relativity}, Jour.\ Math.\ Phys. \textbf{8}
  (1967), 782--786. \MR{MR0213139 (35 \#4004)}

\bibitem{GerochDoD}
\bysame, \emph{Domain of dependence}, Jour.\ Math. Phys. \textbf{11} (1970),
  437--449.

\bibitem{GMP}
F.~Giannoni, A.~Masiello, and P.~Piccione, \emph{A variational theory for light
  rays in stably causal {L}orentzian manifolds: regularity and multiplicity
  results}, Commun.\ Math.\ Phys. \textbf{187} (1997), 375--415.

\bibitem{GuilleminPollack}
V.~Guillemin and A.~Pollack, \emph{Differential topology}, Prentice--Hall,
  Englewood Cliffs, N.J, 1974.

\bibitem{Harris}
S.G. Harris, \emph{What is the shape of space in a spacetime?}, Differential
  geometry: geometry in mathematical physics and related topics ({L}os
  {A}ngeles, {CA}, 1990), Proc. Sympos. Pure Math., vol.~54, Amer. Math. Soc.,
  Providence, RI, 1993, pp.~287--296. \MR{1216546 (94e:53065)}

\bibitem{Hartman}
P.~Hartman, \emph{Ordinary differential equations}, J. Wiley \&\ Sons,
  Baltimore, 1973.

\bibitem{HE}
S.W. Hawking and G.F.R. Ellis, \emph{The large scale structure of space-time},
  Cambridge University Press, Cambridge, 1973, Cambridge Monographs on
  Mathematical Physics, No. 1. \MR{MR0424186 (54 \#12154)}

\bibitem{HopfRinow}
H.~Hopf and W.~Rinow, \emph{{Ueber den Begriff der vollst\"aendingen
  differentialgeometrischen Fl\"ache}}, Comment.\ Math.\ Helv. \textbf{3}
  (1931), 209--225.

\bibitem{KichenassamyRendall}
S.~Kichenassamy and A.~Rendall, \emph{Analytic description of singularities in
  {G}owdy space-times}, Class.\ Quantum Grav. \textbf{15} (1998), 1339--1355.

\bibitem{KlainermanRodnianski:r2}
S.~Klainerman and I.~Rodnianski, \emph{The causal structure of microlocalized
  rough {E}instein metrics}, Ann.\ of Math.\ (2) \textbf{161} (2005),
  1195--1243. \MR{MR2180401 (2007d:58052)}

\bibitem{KlainermanRodnianski:r1}
\bysame, \emph{Rough solutions of the {E}instein-vacuum equations}, Ann.\ of
  Math.\ (2) \textbf{161} (2005), 1143--1193. \MR{MR2180400 (2007d:58051)}

\bibitem{Kriele}
M.~Kriele, \emph{Spacetime}, Lecture Notes in Physics. New Series m:
  Monographs, vol.~59, Springer-Verlag, Berlin, 1999, Foundations of general
  relativity and differential geometry. \MR{2001g:53126}

\bibitem{KeyeMartin}
K.~Martin, \emph{Compactness of the space of causal curves}, Class.\ Quantum
  Grav. \textbf{23} (2006), 1241--1251. \MR{MR2205482}

\bibitem{Maxwell:Compact}
D.~Maxwell, \emph{Rough solutions of the {E}instein constraint equations on
  compact manifolds}, Jour.\ Hyperbolic Diff.\ Equ. \textbf{2} (2005),
  521--546, arXiv:gr-qc/0506085. \MR{MR2151120 (2006d:58027)}

\bibitem{Maxwell:Rough}
\bysame, \emph{Rough solutions of the {E}instein constraint equations}, J.
  Reine Angew. Math. \textbf{590} (2006), 1--29, arXiv:gr-qc/0405088.
  \MR{MR2208126 (2006j:58044)}

\bibitem{MilnorMorse}
J.~Milnor, \emph{Morse theory}, Annals of Mathematics Studies, vol.~51,
  Princeton Univ.\ Press, 1963.

\bibitem{MinguzziLL}
E.~Minguzzi, \emph{Chronological spacetimes without lightlike lines are stably
  causal}, Comm. Math. Phys. \textbf{288} (2009), 801--819. \MR{2504855
  (2010e:83089)}

\bibitem{MinguzziK}
\bysame, \emph{{$K$}-causality coincides with stable causality}, Comm. Math.
  Phys. \textbf{290} (2009), 239--248. \MR{2520513 (2010i:53133)}

\bibitem{MinguzziSanchez}
E.~Minguzzi and M.~S{\'a}nchez, \emph{The causal hierarchy of spacetimes},
  Recent developments in pseudo-{R}iemannian geometry, ESI Lect. Math. Phys.,
  Eur. Math. Soc., Z\"urich, 2008, pp.~299--358. \MR{2436235 (2010b:53128)}

\bibitem{MisnerTaub}
C.W. Misner and A.~Taub, \emph{A singularity-free empty universe}, Soviet.\
  Phys.\ JEPT \textbf{28} (1969), 122--133.

\bibitem{MoncriefNI}
V.~Moncrief, \emph{An integral equation for spacetime curvature in general
  relativity}, Surveys in differential geometry. {V}ol. {X}, Surv. Differ.
  Geom., vol.~10, Int. Press, Somerville, MA, 2006, pp.~109--146. \MR{2408224
  (2009h:53166)}

\bibitem{MSV}
U.~Mueller, C.~Schubert, and A.E.M. van~de Ven, \emph{{A closed formula for the
  Riemann normal coordinate expansion}}, Gen.\ Rel.\ Gravitation \textbf{31}
  (1999), 1759--1781.

\bibitem{NavarroMinguzzi}
J.J.B. Navarro and E.~Minguzzi, \emph{{The stability of global hyperbolicity}},
   (2011), arXiv:1108.5210 [gr-qc].

\bibitem{NUT}
E.~Newman, L.~Tamburino, and T.~Unti, \emph{Empty-space generalization of the
  {S}chwarzschild metric}, Jour.\ Math.\ Phys. \textbf{4} (1963), 915--923.
  \MR{MR0152345 (27 \#2325)}

\bibitem{NomizuOzeki}
K.~Nomizu and H.~Ozeki, \emph{The existence of complete {R}iemannian metrics},
  Proc.\ Amer.\ Math.\ Soc. \textbf{12} (1961), 889--891. \MR{MR0133785 (24
  \#A3610)}

\bibitem{BONeill}
B.~O'Neill, \emph{Semi-{R}iemannian geometry}, Pure and Applied Mathematics,
  vol. 103, Academic Press, New York, 1983. \MR{MR719023 (85f:53002)}

\bibitem{Penrose:singularity}
R.~Penrose, \emph{Gravitational collapse and space-time singularities}, Phys.\
  Rev.\ Lett. \textbf{14} (1965), 57--59. \MR{0172678 (30 \#2897)}

\bibitem{PenroseDiffTopo}
\bysame, \emph{Techniques of differential topology in relativity}, SIAM,
  Philadelphia, 1972, (Regional Conf.\ Series in Appl.\ Math., vol. 7).

\bibitem{Perlick}
V.~Perlick, \emph{On {F}ermat's principle in general relativity: I. {T}he
  general case}, Class.\ Quantum Grav. \textbf{7} (1990), 1319--1331.

\bibitem{Seifert}
H.J. Seifert, \emph{Smoothing and extending cosmic time functions}, Gen.\ Rel.\
  Grav. \textbf{8} (1977), 815--831. \MR{MR0484260 (58 \#4185)}

\bibitem{SmithTataru}
H.~Smith and D.~Tataru, \emph{Sharp local well-posedness results for the
  nonlinear wave equation}, Ann.\ of Math. (2) \textbf{162} (2005), 291--366.
  \MR{MR2178963 (2006k:35193)}

\bibitem{SorkinWoolgar}
R.D. Sorkin and E.~Woolgar, \emph{A causal order for space-times with {$C\sp
  0$} {L}orentzian metrics: proof of compactness of the space of causal
  curves}, Class.\ Quantum Grav. \textbf{13} (1996), 1971--1993. \MR{MR1400951
  (97e:53123)}

\bibitem{Taub}
A.H. Taub, \emph{Empty space-times admitting a three parameter group of
  motions}, Ann.\ of Math.\ (2) \textbf{53} (1951), 472--490. \MR{MR0041565
  (12,865b)}

\bibitem{Teschl}
G.~Teschl, \emph{Ordinary differential equations and dynamical systems}, 2011,
  \url{http://www.mat.univie.ac.at/~gerald/ftp/book-ode/ode.pdf}.

\bibitem{Wald:book}
R.M. Wald, \emph{General relativity}, University of Chicago Press, Chicago,
  1984.

\bibitem{WangRicci}
Q.~Wang, \emph{{On Ricci coefficients of null hypersurfaces with time foliation
  in Einstein vacuum space-time}}, arXiv:1006.5963.

\bibitem{WangCones}
\bysame, \emph{On the geometry of null cones in {E}instein-vacuum spacetimes},
  Ann.\ Inst. H. Poincar\'e Anal.\ Non Lin\'eaire \textbf{26} (2009), 285--328.
  \MR{2483823 (2011b:53172)}

\end{thebibliography}
\end{document}